\documentclass[
	paper=a4,
	pagesize=auto,
	fontsize=10pt,
]{scrartcl}

\usepackage[utf8]{inputenc}
\usepackage[
	english,
]{babel}

\KOMAoptions{%
   DIV=10,
}%
\KOMAoptions{%
   numbers=noenddot
}%
\KOMAoptions{
   twoside=semi,  
   twocolumn=false, 
   headinclude=false,%
   footinclude=false,%
   mpinclude=false,%
   headlines=2.1,%
   headsepline=true,%
   footsepline=false,%
   cleardoublepage=empty 
}%
\KOMAoptions{%
   parskip=relative, 
}%
\KOMAoptions{%
}%
\KOMAoptions{%
}%
\KOMAoptions{%
   bibliography=totoc 
}%
\KOMAoptions{%
}%
\KOMAoptions{%
}%
\KOMAoptions{%
}%

\makeatletter
\newcommand{\LoadPackagesNow}{}
\newcommand{\LoadPackageLater}[2][]{%
   \g@addto@macro{\LoadPackagesNow}{%
      \usepackage[#1]{#2}%
   }%
}
\makeatother

\usepackage{xspace}
\usepackage{xargs}
\usepackage{ifthen}
\usepackage{etoolbox}

\usepackage{authblk}

\usepackage{calc}
\usepackage[
	margin=3cm,
	]{geometry}
\usepackage{framed}
\usepackage{mdframed}
\usepackage{pbox}
\usepackage[shortlabels]{enumitem}
\usepackage{array}
\usepackage{multirow}
\usepackage{longtable}
\usepackage{arydshln}
\usepackage{hhline}
\usepackage{afterpage}
\usepackage{pdflscape}
\usepackage{rotating}
\usepackage{booktabs}
\usepackage{setspace}
\usepackage[
   bottom,      
   stable,      
   perpage,     
   ragged,      
   multiple,    
   flushmargin,
   hang,
]{footmisc}

\usepackage{mathpazo}

\usepackage[%
   headsepline,
   automark,
   komastyle,
]{scrpage2}

\clearscrheadings
\clearscrplain
\lehead{\pagemark}
\rohead{\pagemark}
\rehead{\headmark}
\lohead{\shortauthor}
\automark[section]{section} 

\usepackage{relsize}
\usepackage[normalem]{ulem}      
\usepackage{soul}		           
\usepackage{url}
\usepackage{lipsum}

\usepackage[table,dvipsnames]{xcolor}
\usepackage[%
]{graphicx}

\usepackage{epstopdf}
\usepackage{wrapfig}
\usepackage{caption}
\usepackage{subcaption}
\usepackage[
   tbtags,    
   sumlimits,  
   nointlimits, 
   namelimits, 
   reqno,     
]{amsmath} %

\usepackage{amsfonts}
\usepackage{mathrsfs} 
\usepackage{dsfont}
\usepackage{amssymb}
\usepackage{units}
\LoadPackageLater{amsthm}
\LoadPackageLater{thmtools}
\usepackage[fixamsmath,disallowspaces]{mathtools}
\mathtoolsset{showonlyrefs}
\mathtoolsset{centercolon=true}
\usepackage{bm} 
\usepackage{bbm} 
\makeatletter
\g@addto@macro\bfseries{\boldmath}
\makeatother
\allowdisplaybreaks[4]
\numberwithin{equation}{section}

\usepackage[toc]{appendix}

\usepackage{csquotes}
\usepackage[
	style=alphabetic,
	sortcites=true,
	giveninits=true,
	maxbibnames=99,
	backend=biber,
	maxalphanames=5,
]{biblatex}

\renewbibmacro{in:}{}
\DeclareFieldFormat{pages}{#1}

\usepackage[textsize=tiny,english,colorinlistoftodos]{todonotes}

\definecolor{pdfurlcolor}{rgb}{0,0,0.6}
\definecolor{pdffilecolor}{rgb}{0.7,0,0}
\definecolor{pdflinkcolor}{rgb}{0,0,0.6}
\definecolor{pdfcitecolor}{rgb}{0,0,0.6}
\usepackage[
  colorlinks=true,         
  urlcolor=pdfurlcolor,    
  filecolor=pdffilecolor,  
  linkcolor=pdflinkcolor,  
  citecolor=pdfcitecolor,  %
  raiselinks=true,			 
  breaklinks,              
  verbose,
  hyperindex=true,         
  linktocpage=true,        
  hyperfootnotes=false,     
  bookmarks=true,          
  bookmarksopenlevel=1,    
  bookmarksopen=true,      
  bookmarksnumbered=true,  
  bookmarkstype=toc,       
  plainpages=false,        
  pageanchor=true,         
  pdfdisplaydoctitle=true, 
  pdfstartview=FitH,       
  pdfpagemode=UseOutlines, 
  pdfpagelabels=true,           
  pdfpagelayout=OneColumn, 
]{hyperref}

\LoadPackagesNow

\newcommand{\ifargdef}[3][{}]{\ifthenelse{\equal{#2}{}}{#1}{#3}}

\newlength{\hangind}
\newcommand{\myhangindent}[1]{\settowidth{\hangind}{\widthof{#1}}\hangindent=\the\hangind}

\setkomafont{section}{\Large\rmfamily\bfseries}
\setkomafont{subsection}{\large\rmfamily\bfseries}
\setkomafont{subsubsection}{\rmfamily\bfseries}
\setkomafont{paragraph}{\rmfamily\bfseries}
\setkomafont{pageheadfoot}{\smaller\scshape\rmfamily}

\hypersetup{
	pdfauthor={Maximilian März, Claire Boyer, Jonas Kahn, Pierre Weiss},
	pdftitle={Sampling Rates for $\ell^1$-Synthesis},
	pdfsubject={Article},
	pdfcreator={PDF-LaTeX},
}

\newenvironment{model}[2][]{%
\refstepcounter{theorem}%
\ifstrempty{#1}%
{\mdfsetup{%
frametitle={%
\tikz[baseline=(current bounding box.east),outer sep=0pt]
\node[anchor=east,rectangle,fill=gray!30]
{\strut Model~\arabic{section}.\arabic{theorem}};}}
}%
{\mdfsetup{%
frametitle={%
\tikz[baseline=(current bounding box.east),outer sep=0pt]
\node[anchor=east,rectangle,fill=gray!30]
{\strut Model~\arabic{section}.\arabic{theorem}:~#1};}}%
}%
\mdfsetup{innertopmargin=10pt,linecolor=gray!20,roundcorner=2pt,
linewidth=2pt,topline=true,backgroundcolor=gray!10,
frametitleaboveskip=\dimexpr-\ht\strutbox\relax
}
\begin{mdframed}[skipabove=5pt,skipbelow=2ex,nobreak=true]\relax%
\label{#2}}{\end{mdframed}}

\newenvironment{highlight}{\begin{quote}\itshape}{\end{quote}}

\newenvironment{expstep}
{\begin{itemize}[label={$\blacktriangleright$},leftmargin=1.5em]}
{\end{itemize}}

\newcommand{\expkwd}[1]{\noindent\textbf{#1}}

\newtheoremstyle{claim}
	{\topsep}{\topsep}%
	{\itshape}
	{}
	{}
	{}
	{.5em}
	{{\bfseries\boldmath\thmname{#1} \thmnumber{#2}} \thmnote{(#3)}}

\newtheoremstyle{definition}
	{\topsep}{\topsep}%
	{}
	{}
	{}
	{}
	{.5em}
	{\textbf{\thmname{#1} \thmnumber{#2}} \thmnote{(#3)}}
	
\newtheoremstyle{algorithm}
	{\topsep}{\topsep}%
	{}
	{}
	{\bfseries\boldmath}
	{}
	{\newline}
	{\thmname{#1} \thmnumber{#2} \thmnote{(#3)}}

\mdfdefinestyle{emphframe}{linecolor=black, innertopmargin = 3pt, splittopskip = \topskip, skipbelow=6pt, skipabove=6pt}

\declaretheorem[style=claim,numberwithin=section]{theorem}
\declaretheorem[style=claim,sibling=theorem]{lemma}

\declaretheorem[style=claim,sibling=theorem]{corollary}
\declaretheorem[style=claim,sibling=theorem]{proposition}
\declaretheorem[style=definition,sibling=theorem]{definition}

\declaretheorem[style=definition,sibling=theorem,qed=$\Diamond$]{remark}
\declaretheorem[style=definition,sibling=theorem,qed=$\Diamond$]{example}
\declaretheorem[style=algorithm,sibling=theorem,%
	preheadhook={\begin{mdframed}[style=emphframe,nobreak=true] \setcounter{mpfootnote}{\value{footnote}}},%
	postfoothook=\end{mdframed}]{experiment}
\declaretheorem[style=algorithm,sibling=theorem,%
	preheadhook={\begin{mdframed}[style=emphframe] \setcounter{mpfootnote}{\value{footnote}}},%
	postfoothook=\end{mdframed}]{algorithm}

\newcommand{\opleft}[1]{\mathopen{}\left#1}
\newcommand{\opright}[1]{\right#1\mathclose{}}
\newcommandx{\braces}[4]{%
\ifstrequal{#3}{normal}{#1#4#2}{%
\ifstrequal{#3}{auto}{\left#1#4\right#2}{%
\ifstrequal{#3}{opauto}{\opleft#1#4\opright#2}{%
#3#1#4#3#2}}}%
}
\newcommandx{\opannot}[3][3=\downarrow]{\stackrel{\mathclap{\substack{#1 \\ #3 \vspace{2pt}}}}{#2}}
\newcommandx{\lineannot}[3][3=\rightarrow]{\mathllap{\boxed{\text{\textsmaller{#1}}} #3} #2}
\newcommandx{\multilineannot}[4][4=\rightarrow]{\mathllap{\boxed{\parbox{#1}{\RaggedRight\textsmaller{#2}}} #4} #3}
 
\newcommand{\N}{\mathbb{N}} 
\newcommand{\R}{\mathbb{R}} 

\newcommand{\E}{\mathbb{E}}
\newcommand{\Prob}{\mathbb{P}}

\newcommand{\suchthat}[1][normal]{\ifstrequal{#1}{normal}{\mid}{#1|}} 

\newcommand{\setcompl}[1]{#1^c} 

\newcommandx{\intvcl}[3][1=normal]{\braces{[}{]}{#1}{#2, #3}} 
\newcommandx{\intvop}[3][1=normal]{\braces{(}{)}{#1}{#2, #3}} 
\newcommandx{\intvclop}[3][1=normal]{\braces{[}{)}{#1}{#2, #3}} 
\newcommandx{\intvopcl}[3][1=normal]{\braces{(}{]}{#1}{#2, #3}} 
\DeclareMathOperator*{\argmin}{argmin} 
\DeclareMathOperator{\sign}{sign}

\newcommandx{\abs}[2][1=auto]{\braces{\lvert}{\rvert}{#1}{#2}} 
\newcommandx{\ceil}[2][1=normal]{\braces{\lceil}{\rceil}{#1}{#2}} 
\newcommandx{\floor}[2][1=normal]{\braces{\lfloor}{\rfloor}{#1}{#2}} 
\newcommandx{\round}[2][1=normal]{\braces{[}{]}{#1}{#2}} 
\newcommandx{\der}[1]{D^{#1}} 
\newcommandx{\gradient}{\nabla} 
\newcommandx{\partder}[4][1={},4={}]{\frac{\partial^{#4} #2}{\partial #3^{#4}}\ifargdef{#1}{\Big|_{#1}}} 
\newcommandx{\integ}[4][1={},2={}]{\int_{#1}^{#2} #3 \, #4} 
\newcommandx{\asympffaster}[2][1=normal]{o\braces{(}{)}{#1}{#2}} 
\newcommandx{\asympfaster}[2][1=normal]{O\braces{(}{)}{#1}{#2}} 
\newcommandx{\asympeq}[2][1=normal]{\Theta\braces{(}{)}{#1}{#2}} 
\newcommandx{\asympsslower}[2][1=normal]{\omega\braces{(}{)}{#1}{#2}} 
\newcommandx{\asympslower}[2][1=normal]{\Omega\braces{(}{)}{#1}{#2}} 
\newcommand{\convhull}[1]{\operatorname{conv}(#1)} 

\DeclareMathOperator{\Id}{\textbf{Id}} 
\newcommand{\matr}[1]{\begin{bmatrix} #1 \end{bmatrix}} 
\newcommandx{\norm}[2][1=auto]{\braces{\|}{\|}{#1}{#2}} 
\renewcommandx{\sp}[3][1=auto]{\braces{\langle}{\rangle}{#1}{#2, #3}} 
\newcommandx{\End}[2][2={}]{\mathcal{L}\opleft( #1 \ifargdef{#2}{, #2} \opright)} 
\newcommand{\orthcompl}[1]{{#1}^\perp} 
\DeclareMathOperator{\ran}{ran} 
\DeclareMathOperator{\spann}{\operatorname{span}} 
\renewcommand{\vec}[1]{\boldsymbol{#1}} 

\newcommandx{\measure}[2][1=normal]{\operatorname{vol}\braces{(}{)}{#1}{#2}} 
\DeclareMathOperator{\supp}{supp} 
\newcommandx{\Leb}[3][1={},3=normal]{L^{#2}\ifargdef{#1}{\braces{(}{)}{#3}{#1}}{}} 
\newcommandx{\Lebnorm}[4][1=normal,3={2},4={}]{\norm[#1]{#2}_{#3}} 
\renewcommandx{\l}[3][1={},3=normal]{\ell^{#2}\ifargdef{#1}{\braces{(}{)}{#3}{#1}}} 
\newcommandx{\lnorm}[4][1=normal,3={2},4={}]{\norm[#1]{#2}_{#3}} 
\newcommandx{\Smooth}[4][1={},3={},4=normal]{C_{#3}^{#2}\ifargdef{#1}{\braces{(}{)}{#4}{#1}}} 
\newcommandx{\Schwartz}[2][1={},2=normal]{\mathscr{S}\ifargdef{#1}{\braces{(}{)}{#2}{#1}}} 
\newcommandx{\Schwartzpoly}[2][1=normal]{\braces{\langle}{\rangle}{#1}{\abs[#1]{#2}} } 
\newcommandx{\Tempdistr}[2][1={},2=normal]{\mathscr{S}'\ifargdef{#1}{\braces{(}{)}{#2}{#1}}} 
\newcommandx{\distrinp}[3][1=normal]{\braces{\langle}{\rangle}{#1}{#2, #3}} 
\newcommandx{\ft}[3][1=default,2=auto]{
\ifstrequal{#1}{default}{\widehat{#3}}{
\ifstrequal{#1}{long}{{\braces{(}{)}{#2}{#3}}^{\wedge}}{}}} 
\newcommandx{\ift}[3][1=default,2=auto]{
\ifstrequal{#1}{default}{\check{#3}}{
\ifstrequal{#1}{long}{{\braces{(}{)}{#2}{#3}}^{\vee}}{}}} 

\newcommand{\I}{\vec{Id}}
\newcommand{\e}{\text{e}}

\newcommand{\sigNoisefree}{(\hyperref[eq:sig]{$\text{BP}_{\smash{\noiseparam=0}}^{\text{\smash{sig}}}$})}

\newcommand{\coefNoisefree}{(\hyperref[eq:coef]{$\text{BP}_{\smash{\noiseparam=0}}^{\text{\smash{coef}}}$})}

\newcommand{\x}{\vec{x}} 
\newcommand{\cc}{\vec{c}} 
\newcommand{\Xb}{\vec{X}} 
\newcommand{\eb}{\vec{e}} 
\newcommand{\xx}{\vec{\tilde{x}}} 
\newcommand{\gt}{\vec{x}_0} 
\newcommand{\gtn}{\vec{\tilde{x}_0}} 
\newcommand{\g}{\vec{g}} 
\newcommand{\f}{\vec{f}}
\newcommand{\ub}{\vec{u}} 
\newcommand{\vb}{\vec{v}} 
\newcommand{\sol}{\hat{\vec{x}}}
\newcommand{\ax}{\vec{\theta}}
\newcommand{\kk}{\vec{k}}

\newcommand{\y}{\vec{y}}
\newcommand{\noiseparam}{\eta}
\newcommand{\noise}{\vec{e}}
\newcommand{\solx}{\hat{\vec{x}}}

\renewcommand{\r}{\vec{r}}
\newcommand{\z}{\vec{z}}
\newcommand{\za}{\vec{z}^\ast}
\newcommand{\xa}{\vec{x}^\ast}
\newcommand{\zo}{\vec{z}^2}
\newcommand{\zt}{\vec{z}^1}
\newcommand{\zi}{\vec{z}^i}
\newcommand{\zz}{\vec{\bar{z}}}
\newcommand{\gtz}{\vec{z}_0}
\newcommand{\h}{\vec{h}}
\newcommand{\solz}{\hat{\vec{z}}}
\newcommand{\zl}{\vec{z}_{\ell^{\smash{1}}}}
\newcommand{\Supp}{\mathcal{S}}
\newcommand{\Suppc}{\mathcal{S}^c}
\newcommand{\vvec}{\vec{v}} 
 
\newcommand{\meas}{\vec{a}}
\newcommand{\Meas}{\vec{A}}
\newcommand{\MeasDict}{\vec{\Phi}}

\newcommand{\Dict}{\vec{D}}
\newcommand{\dict}{\vec{d}}

\newcommand{\TV}{\nabla}

\newcommand{\Dictw}{\vec{D}_{\texttt{Haar}}}

\newcommand{\Zset}{\hat{Z}}
\newcommand{\Zlset}{Z_{\ell^{\smash{1}}}}
\newcommand{\X}{\hat{X}}

\newcommandx{\cmw}[2][1={}]{w^{#1}_{\wedge}(#2)} 
\newcommand{\cone}[1]{\operatorname{cone}(#1)} 

\newcommand{\ds}[1]{\mathcal{D}(#1)} 
\newcommand{\dc}[1]{\mathcal{D}_{\wedge}(#1)} 

\newcommandx{\lmin}[3][1={}]{\lambda^{#1}_{\text{min}}\braces{(}{)}{auto}{#2; #3}}
\newcommandx{\lmax}[3][1={}]{\lambda^{#1}_{\text{max}}\braces{(}{)}{auto}{#2; #3}}
\newcommandx{\condi}[3][1={}]{\kappa^{#1}_{{#2, #3}}}

\newcommand{\Sn}[1]{\mathcal{S}^{\smash{#1}}}
\newcommandx{\Bn}[2][1={}]{\text{B}_{#1}^{\smash{#2}}}
\newcommand{\Onenorm}{\norm{\cdot}_1}

\newcommand{\proj}{\vec{P}} 

\newcommand{\gaussparam}{\gamma}
\newcommand{\constant}{c}

\newcommand{\one}{\mathbbm{1}}

\begin{document}

\pagestyle{scrheadings}

\begin{center}
{\bfseries\larger[2]{Sampling Rates for $\ell^1$-Synthesis}}
\end{center}

\vspace{1.5\baselineskip}
\begin{addmargin}[2em]{2em}
 
\noindent{\normalsize\bfseries{Maximilian März$^\S$, Claire Boyer$^*$, Jonas Kahn$^\dagger$, Pierre Weiss$^{\ddagger}$}}

{\smaller\noindent Affiliations: $^\S$Technische Universität Berlin; $^*$LPSM, Sorbonne Universit\'{e}, ENS Paris; $^\dagger$Universit\'{e} de Toulouse;  $^{\ddagger}$ITAV, CNRS, Universit\'{e} de Toulouse

\noindent E-Mail of corresponding author: $^\S$maerz@math.tu-berlin.de
}

\vspace{1\baselineskip}
{\smaller
\noindent\textbf{Abstract.}
This work investigates the problem of signal recovery from undersampled noisy sub-Gaussian measurements under the assumption of a synthesis-based sparsity model. Solving the $\ell^1$-synthesis basis pursuit allows for a simultaneous estimation of a coefficient representation as well as the sought-for signal. However, due to linear dependencies within redundant dictionary atoms it might be  impossible to identify a specific representation vector, although the actual signal is still successfully recovered. The present manuscript studies both estimation problems from a non-uniform, signal-dependent perspective. By utilizing recent results on the convex geometry of linear inverse problems, the sampling rates describing the phase transitions of each formulation are identified. In both cases, they are given by the conic Gaussian mean width of an $\ell^1$-descent cone that is linearly transformed by the dictionary. 
In general, this expression does not allow a simple calculation by following the polarity-based approach commonly found in the literature. Hence, two upper bounds involving the sparsity of coefficient representations are provided: The first one is based on a local condition number and the second one on a geometric analysis that makes use of the thinness of high-dimensional polyhedral cones with not too many generators. It is furthermore revealed that both recovery problems can differ dramatically with respect to robustness to measurement noise -- a fact that seems to have gone unnoticed in most of the related literature. All insights are carefully undermined by numerical simulations.   

\vspace{.5\baselineskip}
}

\end{addmargin}
\newcommand{\shortauthor}{März, Boyer, Kahn \& Weiss: $\ell^1$-Synthesis}

\thispagestyle{plain}



\section{Introduction}
 
In the last two decades, the methodology of \emph{compressive sensing} promoted the use of sparsity based methods for many signal processing tasks. Following the seminal works of Cand\`{e}s, Donoho, Romberg and Tao  \cite{candes2004robust,candes2004near,donoho2006cs}, a vast amount of research has extended the understanding, how additional structure can be exploited for solving ill-posed inverse problems. The classical setup in this area considers a \emph{non-adaptive, linear measurement model,} which reads as follows: 

\begin{model}[Linear Noisy Measurements]{mod:noisy_meas}
 Let $\gt \in \R^n$ be a fixed vector, which is typically referred to as the \emph{signal}. Assume that we are given $m$ \emph{measurements} $\y \in \R^m$ of $\gt$ via the linear acquisition model
 \begin{equation}
  \y = \Meas \gt + \noise,
 \end{equation}
 where $\Meas \in \R^{m \times n}$ is the so-called \emph{measurement matrix} and $\noise \in \R^m$ models \emph{measurement noise} with $\norm{\noise}_2 \leq \noiseparam$ for some $\noiseparam \geq 0$. 
\end{model}

\enlargethispage{1\baselineskip}
The goal of compressive sensing is to solve this inverse problem by reconstructing an approximation of the signal $\gt$ from its indirect measurements $\y$. 
Remarkably, even if $m\ll n$, this task can be achieved by incorporating additional information during the reconstruction process. Most classical compressive sensing works directly assume that $\gt$ is \emph{$s$-sparse}, i.e., that at most $s \ll n$ entries of $\gt$ are nonzero or in symbols $\norm{\gt}_0 = \# \supp (\gt) \leq s$. However, this assumption is hardly satisfied in any real-world application. Nevertheless, many signals allow for sparse representations using specific transforms, such as Gabor dictionaries, wavelet systems or data-adaptive representations, which are inferred from a given set of training samples. Such a model is referred to as \emph{synthesis formulation}, since it assumes that there exists a matrix $\Dict \in \R^{n \times d}$ and a low-complexity representation $\gtz\in \R^d$ such that $\gt$ can be ``synthesized'' as
\begin{equation}
  \label{eq:sparse_rep}
 \gt = \Dict \cdot \gtz.
\end{equation}

Following the standard terminology of the field, the matrix $\Dict = [\dict_1,\dots,\dict_d]$ will be henceforth refered to as \emph{dictionary} and its columns as \emph{dictionary atoms}. It can be expected that the coefficient vector $\gtz$ is dominated by just a few large entries, provided that $\Dict$ allows to capture the signal's inherent structure reasonably well.  

The synthesis formulation of compressive sensing exploits such a representation model, for instance, by employing greedy-based reconstruction algorithms or by utilizing the sparsity-promoting effect of the $\ell^{\smash{1}}$-norm. In this work, we will consider the following convex program, which we refer to as \emph{synthesis basis pursuit for coefficient recovery}:
\begin{equation}
 \label{eq:coef}
 \tag{$\text{BP}_{\noiseparam}^{\text{\smash{coef}}}$}
 \Zset \coloneqq \argmin_{\z \in \R^d} \norm{\z}_1 \quad \mbox{ s.t. } \quad \norm{\y - \Meas  \Dict  \z}_2 \leq \noiseparam. 
\end{equation}%
Under suitable assumptions, one might hope that solutions $\solz$ of this minimization program approximate $\gtz$ reasonably well. Indeed, if $\Dict = \Id$, the formulation \eqref{eq:coef} turns into the classical basis pursuit. It allows to recover any $s$-sparse vector $\gtz$ with overwhelming probability, if $\Meas$ additionally follows a suitable random distribution and $m  \gtrsim s \cdot \log (2n/s)$~\cite{foucart2013cs}.

In many practical and theoretical situations, it turns out that using redundant dictionaries, i.e., choosing $d \gg n$, is beneficial. For instance, the stationary wavelet transform overcomes the lack of translation invariance and learned dictionaries typically infer a larger set of convolutional filters, which are adapted to a particular data distribution. If $\Dict$ does not form a basis, representations as in~\eqref{eq:sparse_rep} are not necessarily unique anymore. Hence, it is not to be expected that a specific representation can be identified by solving~\eqref{eq:coef}. However, in many situations of interest, the representation vector itself is irrelevant and a recovery of the actual signal $\gt$ is of primary interest. Thus, one rather cares about the \emph{synthesis basis pursuit for signal recovery}, which amounts to solving         
\begin{equation}
 \label{eq:sig}
 \tag{$\text{BP}_{\noiseparam}^{\text{\smash{sig}}}$}
 \X \coloneqq  \Dict \cdot \left( \argmin_{\z \in \R^d} \norm{\z}_1 \quad \mbox{ s.t. } \quad \norm{\y - \Meas  \Dict  \z}_2 \leq \noiseparam \right). 
\end{equation}%
In the noiseless case (i.e., when $\noise = \vec{0}$ and $\noiseparam = 0$), it might be the case that $\Zset \neq \left\{ \gtz \right\}$, but there is still hope that $\X = \Dict \cdot \Zset = \left\{ \gt \right\}$. In other words, although solving~\eqref{eq:coef} might fail in identifying a specific coefficient representation, it is still possible that the actual signal is successfully recovered by a subsequent synthesis with $\Dict$. 

\begin{figure}
	\centering
	\begin{subfigure}[t]{0.47\textwidth}
		\centering
		\includegraphics[width=\textwidth]{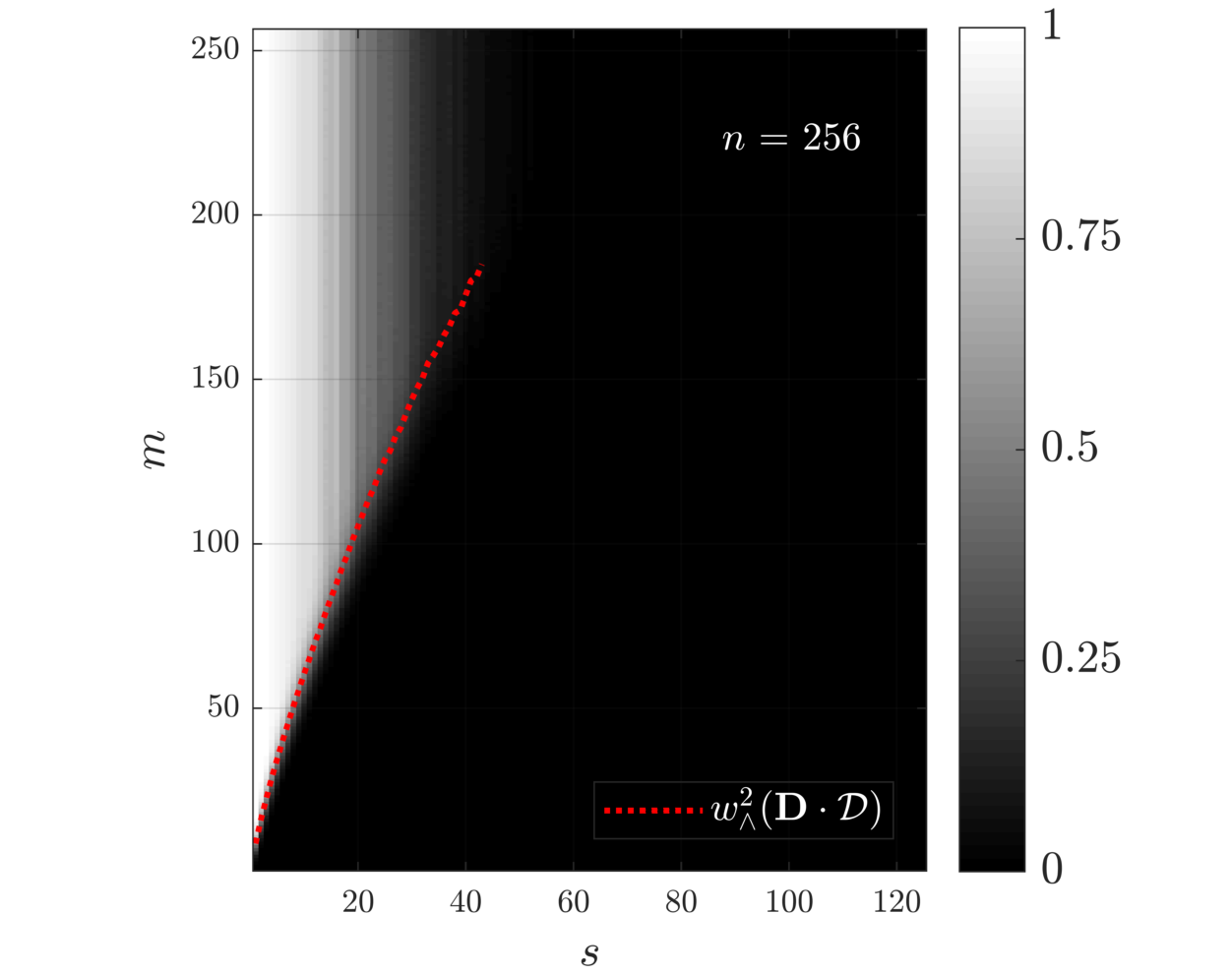}
		\caption{}
		\label{fig:sig:pt_2:1}
	\end{subfigure}%
	\qquad
	\begin{subfigure}[t]{0.47\textwidth}
		\centering
		\includegraphics[width=\textwidth]{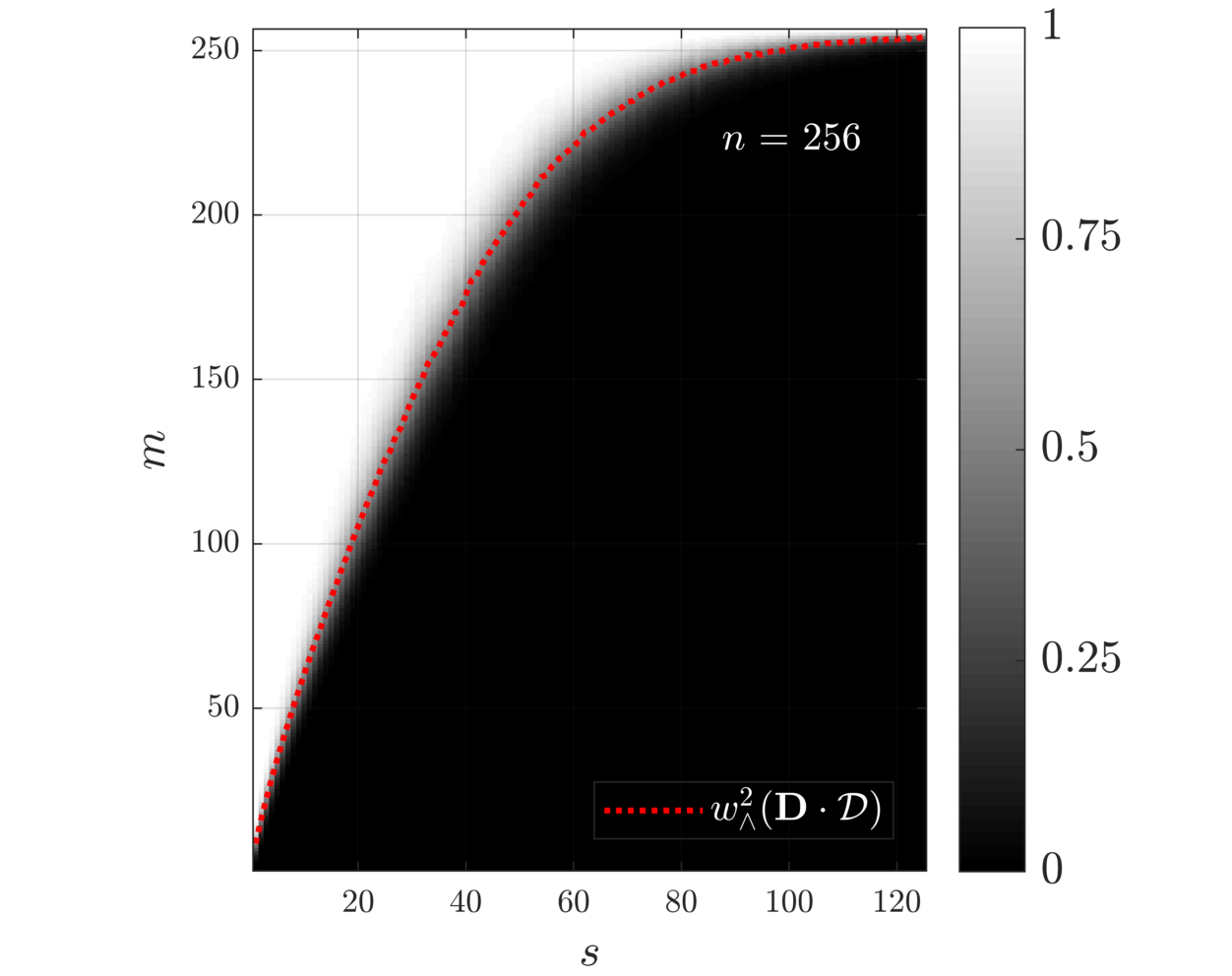}
		\caption{}
		\label{fig:sig:pt_2:2}
	\end{subfigure}%
	\caption{\textbf{Phase transitions of coefficient and signal recovery by $\ell^1$-synthesis.} Subfigure~\ref{fig:sig:pt_2:1} shows the empirical probability that atomic coefficient representations are successfully recovered via solving \coefNoisefree{}, whereas Subfigure~\ref{fig:sig:pt_2:2} shows the empirical probability for the associated signal reconstruction by \sigNoisefree{}. The underlying dictionary is a redundant Haar wavelet frame with three decomposition levels and the defining $s$-sparse coefficients are chosen at random; see Section~\ref{sec:2d_pt} for a precise documentation of the experiment. The brightness of each pixel reflects the observed probability of success, reaching from certain failure (black) to certain success (white). The dotted line shows our predictions for the location of the phase transitions, see Theorem~\ref{thm:coeff} and Theorem~\ref{thm:sig}, respectively. }
	\label{fig:sig:pt_2}
\end{figure}

\subsection{What This Paper Is About}

The goal of this work is to broaden the understanding of the conditions that guarantee \emph{coefficient} and \emph{signal recovery} by solving \eqref{eq:coef} and \eqref{eq:sig}, respectively. To that end, we believe that addressing the following, non-exhaustive list of questions will be of particular importance:
\begin{enumerate}[(Q1)]
 \item Under which circumstances does coefficient and signal recovery differ, i.e., when is it impossible to reconstruct a specific coefficient representation although the signal itself might still be identified?
  \item If possible, how many measurements are required to reconstruct a specific coefficient representation? Analogously, how many measurements are required to recover the associated signal?
 \item In case that coefficients and signals can both be identified, are there still differences between the two formulations, for instance with respect to robustness to measurement noise?
\end{enumerate}

Set out to find answers to these questions, we first restrict ourselves to the following sub-Gaussian measurement model, which will be considered in this work unless  stated o\-ther\-wise: 

\begin{model}[Sub-Gaussian Measurement Model]{mod:meas_mod} 
Let $\meas \in \R^n$ be an isotropic ($\E [\meas \meas^T] =\Id$), zero mean, sub-Gaussian\footnote{A random variable $a$ is \emph{sub-Gaussian} if $\norm{a}_{\psi_2} \coloneqq \sup_{q\geq 1} q^{-1/2} (\E [|a|^q])^{1/q} < \infty$, with $\norm{\cdot}_{\psi_2}$ being the \emph{sub-Gaussian norm of $a$.} For a random vector $\meas \in \R^n$ the sub-Gaussian norm is then given by $\norm{\meas}_{\psi_2} \coloneqq \sup_{\vec{v} \in \Sn{n-1}} \norm{\sp{\meas}{\vec{v}}}_{\psi_2}$ and $\meas$ is called sub-Gaussian if $\norm{\meas}_{\psi_2} < \infty$; see for instance \cite{vershynin2012random} for further details.} random vector with $\norm{\meas}_{\psi_2} \leq \gaussparam$. The sampling matrix $\Meas$ is formed by drawing $m$ independent copies $\meas_1, \dots, \meas_m$ of $\meas$ and setting

\begin{equation}
 \Meas =  \matr{- \meas_1^T -\\  \vdots \\ - \meas_m^T -}.
\end{equation}%
\end{model}
This model has been established as a somewhat classical benchmark setup in the context of compressive sensing. It allows us to follow the methodology initiated in \cite{rudelson2007,mendelson2007reconstruction} and extended in~\cite{stojnic2009gordon,chandrasekaran2012geometry,tropp2014convex,amelunxen2014edge}. In a nutshell, the aim is to determine the \emph{sampling rate} of a convex program (i.e., the number of required measurements for successful recovery) by calculating the so-called \emph{Gaussian mean width}. 

\enlargethispage{-2\baselineskip}
We now briefly outline our work and summarize its main contributions:
\begin{itemize}
\item[(C1)] A cornerstone of our analysis is formed by the set of \emph{minimal $\ell^1$-representers} of $\gt$:
\begin{equation}
\label{eq:minl11}
 \tag{$\text{BP}_{\ell^1}$}
 \Zlset \coloneqq \argmin_{\z \in \R^d} \norm{\z}_1 \quad \mbox{ s.t. } \quad \gt = \Dict  \z. 
\end{equation}
Independently of Model~\ref{mod:meas_mod}, Section~\ref{sec:min_l1} reveals that if $\Zlset = \left\{ \gtz \right\}$, exact recovery of $\gtz$ via \coefNoisefree{} is equivalent to perfect recovery of $\gt$ by solving \sigNoisefree{}. Furthermore, exact recovery of a coefficient vector $\gtz$ by~\coefNoisefree{} is only possible, if $\gtz$ is the unique minimal $\ell^{\smash{1}}$-representer of $\gt = \Dict \gtz$, i.e., if $\Zlset = \left\{\gtz\right\}$.

\item[(C2)] In Section~\ref{sec:sig_recov_cvx_gauge} and Section~\ref{sec:sig_and_coef_rates}, it will be shown that the sampling rate of both formulations can be expressed by the squared conic mean width $\cmw[2]{\Dict \cdot \mathcal{D}}$, where $\mathcal{D}$ denotes the descent cone of the $\ell^{\smash{1}}$-norm at any $\zl \in \Zlset$ (see Section~\ref{sec:primer} for a brief summary of the general recovery framework and definitions of these notions). This observation holds unconditionally true in the case of signal recovery by~\eqref{eq:sig}. For coefficient recovery the additional assumption that $\Zlset$ is a singleton needs to be satisfied.

\item[(C2')] While $\cmw[2]{\Dict \cdot \mathcal{D}}$ forms a precise description of the sampling rate, it is a quantity that is hard to analyze and compute, in general. Therefore, an important goal of our work is to derive more informative upper bounds for this expression. First, under the assumption that $\Zlset = \left\{ \zl \right\}$ is a singleton, we show a condition number bound that relates $\cmw[2]{\Dict \cdot \mathcal{D}}$ to the classical complexity $\cmw[2]{\mathcal{D}} \lesssim  s \cdot \log(2n / s)$, where $s=\norm{\zl}_0$ (see Section~\ref{sec:condbound}). 

The second upper bound of Section~\ref{sec:circumangle} is central to our work. It is based on a geometric analysis that makes use of generic arguments from high-dimensional convex geometry. In comparison to the first bound, it is more general since we do not assume that $\Zlset$ is a singleton. Hence, it particularly addresses the recovery of signals, without requiring the identification of a coefficient representation. The resulting upper bound on the conic mean width relies on the thinness of high-dimensional polyhedral cones with not exponentially many generators. We believe that such an argument might be of general interest beyond its application to the synthesis formulation of compressive sensing. Again, $\cmw[2]{\Dict \cdot \mathcal{D}}$ is related to the sparsity of a minimal $\ell^1$-representation and a further geometrical parameter (referred to as \emph{circumangle}) that measures the narrowness of the associated cone.        

\item[(C3)] Lastly, our recovery statements reveal that recovery of signals by $\eqref{eq:sig}$ is robust to measurement noise without any further restrictions. In contrast, the robustness of coefficient recovery via solving~\eqref{eq:coef} is influenced by an additional factor that is related to the convex program~\eqref{eq:minl11}.   
\end{itemize}
All our findings are underpinned by extensive numerical experiments; see Section~\ref{sec:num_exp}. As a first ``teaser'' we refer the reader to Figure~\ref{fig:sig:pt_2}, which displays two \emph{phase transition plots} and our sampling rates for a redundant Haar wavelet system $\Dict$.

\subsection{Related Literature}

In the following, we first briefly discuss some historical references that are of general interest for $\ell^{\smash{1}}$-norm minimization, sparse representations in redundant dictionaries and compressive sensing. Subsequently, we focus on the existing literature on the synthesis formulation in more depth.  

\subsubsection{Some Historical Landmarks} 

The idea of promoting sparsity in discrete or continuous dictionaries by $\ell^1$-norm minimization can be traced back to the works of Beurling~\cite{Beurling1938} and Krein \cite{Krein1938}. Motivated by questions in Banach space geometry, first theorems establishing sparsity of solutions of related minimization problems can be found in the 1940's~\cite{Zuhovickii1948}. In his PhD-thesis of 1965, Logan utilized $\ell^{\smash{1}}$-minimization for sparse frequency estimation~\cite{Logan1965} and in the 1970's it was employed for solving deconvolution problems in geophysics~\cite{Claerbout1973,Taylor1979}. Of particular importance became the so-called \emph{Rudin-Osher-Fatemi-model}~\cite{rudin1992}, which pioneered the use of total variation minimization for image processing tasks.  

The field of sparse representations arose with the development of (greedy) algorithms for finding expansions in redundant dictionaries such as time-frequency systems~\cite{mallat93,omp1993}. Subsequently, the work~\cite{chen98} triggered notable interest in achieving this task by solving the basis pursuit~\eqref{eq:minl1}; see for instance~\cite{elad2002generalized,donoho2001,Donoho2003}. A special emphasis was thereby given to unions of orthogonal bases~\cite{elad2002generalized,1255564,candes2006quantitative}. Next to the classical concepts of coherence and spark, which are \emph{uniform} across all $s$-sparse signals, also the \emph{non-uniform} notions of dual certificates and exact recovery conditions were progressively introduced~\cite{Fuchs2004,tropp2004,fuchs2005recovery,tropp2006just}.

Under the notion of \emph{compressive sensing}, Cand\`{e}s, Romberg and Tao \cite{candes2004robust,candes2006stable} and Donoho \cite{donoho2006cs} first proposed to capitalize on randomized models in the basis pursuit. In these works, the structured dictionary $\Dict$ is replaced by a random matrix $\Meas$, which follows for instance Model~\ref{mod:meas_mod}. Such a design allows to overcome severe shortcomings of previous results, in particular the \emph{quadratic/square root-bottleneck}; see next subsection or~\cite[Chapter~5.4]{foucart2013cs}. Indeed, under such a randomness assumption, it can be shown that any $s$-sparse vector can be recovered with overwhelming probability if the number of measurements obeys $m \gtrsim s \cdot \log (2n/s)$. These seminal works can furthermore be acknowledged for highlighting the remarkable potential of sparsity-based methods for many signal recovery tasks. 

\subsubsection{Results on the Synthesis Formulation of Compressed Sensing}

An important insight on solving the inverse problem of Model~\ref{mod:noisy_meas} by means of redundant dictionaries was provided by Elad, Milanfar and Rubinstein~\cite{Elad2006}. Therein, the authors compare two different formulations: The synthesis basis pursuit~\eqref{eq:sig} and an alternative formulation, which is referred to as \emph{$\ell^{\smash{1}}$-analysis basis pursuit}:
\begin{equation}
\label{eq:analysis}
 \min\nolimits_{\x \in \R^n} \norm{\vec{\Psi} \x }_1 \quad \mbox{ s.t. } \quad \norm{\y - \Meas \x}_2 \leq \noiseparam.
\end{equation}
The analysis operator $\vec{\Psi}\in \R^{d\times n}$ is thereby chosen in such a way that the coefficient vector $\vec{\Psi}\gt$ is of low-complexity. 
It turns out that the latter formulation and the program~\eqref{eq:sig} are only equivalent if $\vec{\Psi}$ (or $\Dict$) forms a basis. In particular for redundant choices of $\vec{\Psi}$ and $\Dict$, the geometry of both formulations departs significantly from each other. While the synthesis variant appears to be  more natural from a historical perspective, its analysis-based counterpart gained considerable attention in the past years~\cite{candes2011csdict,nam2013,giryes2014,krahmer2015,kabanava2015,rkz2015}. Recently, the non-uniform approach of \cite{genzel2017} revealed that the measure of ``low-complexity'' in the analysis model goes beyond pure sparsity of $\vec{\Psi}\gt$. Instead, a novel sampling-rate bound was proposed that is based on a generalized notion of sparsity, taking the support and the coherence structure of the underlying analysis operator into account.

The earliest reference that deals with the synthesis formulation for the recovery of coefficient vectors appears to be by Rauhut, Schnass and Vandergheynst \cite{rauhut2008}. Therein, the formulation \eqref{eq:coef} is studied under a randomized measurement model. The main result roughly reads as follows: Assume that the dictionary $\Dict$ satisfies a \emph{restricted isometry property (RIP)} with sparsity level $s$. If the random matrix $\Meas\in \R^{m \times n}$ follows Model~\ref{mod:meas_mod} and $m  \gtrsim s\cdot \log (n/s)$, then the composition $\Meas\Dict$ will also satisfy an RIP with sparsity level $s$ with high probability. This property then implies stable and robust recovery of all $s$-sparse coefficient vectors by solving~\eqref{eq:coef}. The assumption that $\Dict$ satisfies an RIP is crucial for the previous result. It can be for instance achieved if the dictionary is sufficiently incoherent, i.e., if it satisfies
\begin{equation}
 \label{eq:coh}
 \mu (\Dict) \coloneqq \max\nolimits_{i\neq j} \left|\sp{\dict_i}{\dict_j}\right| / ({\norm{\dict_i}_2\cdot\norm{\dict_j}_2} ) \leq 1/(16\cdot(s-1)). 
\end{equation}
However, as the authors of~\cite{rauhut2008} point out, such a coherence-based estimate is rather crude and suffers from the so-called square-root bottleneck: The Welch bound~\cite[Theorem 5.7]{foucart2013cs} reveals that condition~\eqref{eq:coh} can only be satisfied for mild sparsity values $s \lesssim\sqrt{n}$. 

In~\cite{chen2014}, Chen, Wang and Wang study conditions for signal recovery via a dictionary-based \emph{null space property (NSP)}: For a given dictionary $\Dict$, a matrix $\Meas$ is said to satisfy the $\Dict$-NSP of order $s$, if for any index set $S \subseteq [d]$ with $\# S \leq s$ and any $\h \in \Dict^{-1}(\ker{\Meas} \setminus \{\vec{0}\})$, there exists $\z \in \ker{\Dict}$, such that $\norm{\h_S + \z}_1 < \norm{\h_{\setcompl{S}}}_1$. It can be shown that this condition is necessary and sufficient for the uniform recovery of all signals $\gt = \Dict  \gtz$ with  $\norm{\gtz}_0 \leq s$ via \eqref{eq:sig}.
Note that the $\Dict$-NSP is in general weaker than requiring that $\Meas\Dict$ satisfies the standard NSP. This means that the previous result is addressing signal recovery without necessarily requiring coefficient recovery. However, the authors then show that under the additional assumption that $\Dict$ is of \emph{full spark} (i.e., every $n$ columns of $\Dict$ are linearly independent), both conditions are in fact equivalent. Hence, in this case, signal and coefficient recovery are also equivalent.  
In the recent work~\cite{chen2017}, this serves as a motivation to study coefficient recovery by analyzing how many measurements are required in order to guarantee that $\Meas\Dict$ has an NSP. To that end, a result is provided that is conceptually similar to~\cite{rauhut2008}, however, it reduces the assumptions on $\Dict$. Instead of requiring that $\Dict$ satisfies an RIP, the authors operate under the weaker assumption that $\Dict$ satisfies an NSP. The main result essentially reads as follows: Under a sub-Gaussian measurement setup similar to Model~\ref{mod:meas_mod} and under the assumption that $\Dict$ satisfies an NSP of order $s$, a number of $m \gtrsim  s\cdot \log (n/s)$ measurements guarantees that also $\Meas\Dict$ satisfies an NSP. This condition then allows for robust recovery of all $s$-sparse coefficient vectors by solving \eqref{eq:coef}. 

To the best of our knowledge, the only work that provides a bound on the required number of measurement for signal recovery (without necessarily requiring coefficient recovery) is the tutorial~\cite[Theorem 7.1]{Vershynin2015}: Assume that $\norm{\dict_i}_2 \leq 1$, $i\in [d]$ and that $\gt = \Dict \gtz$ for an $s$-sparse representation $\gtz \in \R^{\smash{d}}$. For a Gaussian measurement matrix $\Meas \in \R^{m \times n}$, Vershynin establishes the following recovery bound in expectation:
\begin{equation}
 \E \norm{\sol - \gt}_2 \leq c \cdot \sqrt{\tfrac{s\cdot \log(d)}{m}} \cdot \norm{\gtz}_2 + \sqrt{2\pi}\cdot\tfrac{\noiseparam}{\sqrt{m}},
\end{equation}
where $c$ is a constant and $\sol\in\X$ is a solution of \eqref{eq:sig}. Note that we have slightly adapted the statement of~\cite[Theorem 7.1]{Vershynin2015} for a better match with our setup. 
Due to the first summand on the right hand side, the previous error bound is suboptimal, cf.~Theorem~\ref{thm:sig}. In particular, it does not guarantee exact recovery from noiseless measurements. We emphasize that parts of our work are inspired by Vershynin, who also studies the gauge of the set $K = \Dict \cdot \Bn[1]{d}$ in \cite{Vershynin2015}.

We conclude by mentioning a few more works in the literature on synthesis based compressed sensing that appear to be of less relevance for this work. The influential paper~\cite{chandrasekaran2012geometry} studies signal recovery via atomic minimization, however, it does not provide specific insights when redundant dictionaries are used. In \cite{davenport2013}, a (theoretical) CoSaMP algorithm is adapted to the recovery of signals with sparse representations in  redundant dictionaries. Based on the $\Dict$-RIP~\cite{candes2011csdict} and on a connection to the analyis formulation with so-called optimal dual frames, \cite{liu2012b} derives a theorem concerning signal recoery. Finally, \cite{Figueiredo2009} provides  numerical experiments, which empirically compare the analysis and the synthesis formulation.

\subsubsection{The Gap that We Intend to Fill}

In order to obtain statements that are \emph{uniform} across all $s$-sparse signals, most existing results assume that the dictionary $\Dict$ satisfies strong assumptions, e.g., incoherent atoms, an NSP or an RIP~\cite{rauhut2008,chen2014,chen2017}. 
Such notions are well established and allow for appealing results that often resemble known principles of compressive sensing. However, in many situations of interest, these assumptions are too restrictive.
In particular, redundant representation systems (such as Gabor systems, wavelets, curvelets, \dots) or data-adaptive dictionaries do not satisfy any such property. Their atoms are typically highly coherent and share many linear dependencies. We aim to address this issue by following a local, non-uniform approach, which avoids strong assumptions on the dictionary. We believe that such a signal-dependent refinement is crucial for redundant representation systems, cf.~\cite{genzel2017}. 

Similarly, it is occasionally argued that distinguishing signal and coefficient recovery is of minor importance, cf.~\cite{chen2017}. This is justified by the observation that exact recovery of coefficients and signals is equivalent if $\Dict$ is in general position.  
However, due to the linear dependencies in many structured representation systems, such an argumentation is often not valid. Indeed, simple numerical experiments with popular dictionaries reveal that signal recovery can  be frequently observed without reconstructing a specific coefficient representation, see Figure~\ref{fig:sig:pt_2} and Section~\ref{sec:num_exp}. Hence, we believe that it is important to study both formulations and to identify under which conditions coefficient recovery might be expected, see (C1) above.

To the best of our knowledge, this is the first work that provides a precise description of the phase transition behavior of both formulations, see~(C2). While the identified conic mean width of a linearly transformed set $\cmw[2]{\Dict \cdot \mathcal{D}}$ is a rather implicit quantity, it constitutes an important step towards the understanding of $\ell^{\smash{1}}$-synthesis. By deriving more explicit upper bounds on the sampling rate, coefficient sparsity is identified as an important factor. However, additional properties that account for the local geometry are also taken into account, see~(C2').

Last but not least, we establish that both formulations behave differently with respect to robustness to measurement noise, see~(C3). To the best of our knowledge, this aspect has gone unnoticed in the literature so far, although it might have dramatic implications on the reconstruction quality of coefficient representations. 
\enlargethispage{1\baselineskip}

\subsection{Notation}

For the convenience of the reader, we have collected the most important and frequently used objects in Table~\ref{tab:notation}. 

\renewcommand{\arraystretch}{1.2}
\begin{table}
  \centering
    \begin{tabular}{l|l}
    \hline \hline
    \textbf{Notation} & \textbf{Term} \\
    \hline
    $\gt \in \R^n$  & (ground truth) signal vector \\
    \hdashline
    $\Meas \in \R^{m \times n}$ & measurement matrix  \\
    $\noise \in \R^m$, with $\norm{\noise}_2 \leq \noiseparam$ & (adversarial) noise \\
    $\y = \Meas\gt + \noise \in \R^m$ & linear, noisy measurements of $\gt$ \\
    \hdashline
    $\dict_1,\dots,\dict_d \in \R^n$ & dictionary atoms \\
    $\Dict = [\dict_1,\dots,\dict_d] \in \R^{n\times d}$ & dictionary \\
    \hdashline
    $\solx \in \R^n$ & a solution of \eqref{eq:sig} \\
    $\X = \Dict \cdot \Zset \subseteq \R^n$ & solution set of \eqref{eq:sig} \\
    \hdashline
    $\zl \in \R^d$ & a minimal $\ell^1$-decomposition of $\gt$ in $\Dict$, i.e., a solution of \eqref{eq:minl1}\\
    $\Zlset \subseteq \R^d$ & solution set of \eqref{eq:minl1} \\
    $\gtz \in \R^d$ & a sparse representation of $\gt$ in $\Dict$, without $\gtz \in \Zlset$, in general  \\
    $\solz \in \R^d$ & a solution of \eqref{eq:coef} \\
    $\Zset \subseteq \R^d$ & solution set of \eqref{eq:coef} \\
    \hline \hline
    \end{tabular}
    \caption{A summary of the central notations used in this work. }
        \label{tab:notation}
\end{table}

Throughout this manuscript we will use the following notation and conventions: for an integer $n \in \N$ we set $[n] \coloneqq \{1,2,\dots,n\}$. 
If $\mathcal{I} \subseteq [n]$, we let $\setcompl{\mathcal{I}} \coloneqq [n] \setminus \mathcal{I}$ denote the complement of $\mathcal{I}$ in $[n]$.  
Vectors and matrices are symbolized by lower- and uppercase bold letters, respectively. Let $\x = (x_1, \dots, x_n) \in \R^n$. 
For an index set $\mathcal{I} \subseteq [n]$, we let the vector $\x_{\mathcal{I}} \in \R^{\#\mathcal{I}}$ denote the restriction to the components indexed by $\mathcal{I}$. The \emph{support} of $\x$ is defined by the set of its non-zero entries $\supp(\x) \coloneqq \{ k \in [n] \suchthat x_k \neq 0 \}$ and the \emph{sparsity} of $\x$ is $\lnorm{\x}[0] \coloneqq \#\supp(\x)$. 
For $1 \leq p \leq \infty$, $\lnorm{\cdot}[p]$ denotes the \emph{$\l{p}$-norm} on $\R^n$. The associated \emph{unit ball} is given by $\Bn[p]{n} \coloneqq \{ \x \in \R^n \suchthat \lnorm{\x}[p] \leq 1 \}$ and the \emph{Euclidean unit sphere} is $S^{n-1} \coloneqq \{ \x \in \R^n \suchthat \lnorm{\x} = 1 \}$.  
The $i$-th standard basis vector of $\R^n$ is refered to as $\vec{e}_i$ and $\Id \in \R^{n\times n}$ denotes the identity matrix. 
Furthermore, let $\cone{K}$ denote the \emph{conic hull} of a set $K \subset \R^n$. If $L \subset \R^n$ is a linear subspace, the associated \emph{orthogonal projection onto $L$} is given by $\proj_{L} \in \R^{n \times n}$.
Then, we have $\proj_{\orthcompl{L}} = \Id - \proj_{L}$, where $\orthcompl{L} \subset \R^n$ is the orthogonal complement of $L$. The letter $\constant$ is usually reserved for a (generic) constant, whose value could change from time to time.	We refer to $\constant$ as a \emph{numerical constant} if its value does not depend on any other involved parameter. If an \mbox{(in-)equality} holds true up to a numerical constant $\constant$, we sometimes write $a \lesssim b$ instead of $a \leq \constant \cdot b$. 
For a matrix $\vec{A} \in \R^{m \times n}$ we let $\norm{\vec{A}}_2$ denote its \emph{spectral norm}. For a set $K\subseteq \R^n$, $\lambda \in \R$ and $\Meas \in \R^{m \times n}$ we set $\lambda \cdot K \coloneqq \{\lambda \vec{k} : \vec{k} \in K \}$ and  $\Meas \cdot K \coloneqq \{\Meas \vec{k} : \vec{k} \in K \}$. Lastly, the term \emph{orthonormal basis} is abbreviated by ONB. 


\section{A Primer on the Convex Geometry of Linear Inverse Problems}
\label{sec:primer}

In this section, we give a brief introduction to a well-established methodology that addresses the recovery of structured signals from independent linear random measurements. This summary mainly serves the purpose of introducing the required technical notions for our subsequent analysis of the $\ell^{\smash{1}}$-synthesis formulation. It is inspired by \cite{chandrasekaran2012geometry,tropp2014convex,amelunxen2014edge} and we refer the interested reader to these works for a more detailed discussion of the presented material.

\subsection{Minimum Conic Singular Value}
\label{sec:primer_conic_sing_value} 

Assume that Model \ref{mod:noisy_meas} is satisfied. For a robust recovery of $\gt$ from its linear, noisy measurements $\y$, we consider the \emph{generalized basis pursuit}
\begin{equation}
\label{eq:gen_bp}
\tag{$\text{BP}_{\noiseparam}^{\smash{f}}$}
\min_{\x \in \R^n} f(\x) \quad \mbox{ s.t. } \quad \norm{ \y - \Meas\x }_2 \leq \noiseparam,
\end{equation}
\newcommand{\genNoisefree}{(\hyperref[eq:gen_bp]{$\text{BP}_{\noiseparam=0}^{\smash{f}}$})}%
where $f \colon \R^n \to \R$ is a convex function that is supposed to reflect the ``low complexity'' of the signal $\gt$. Hence, the previous minimization problem searches for the most structured signal that is still consistent with the given measurements $\y$. 

The recovery performance of \eqref{eq:gen_bp} can be understood by a fairly standard geometric analysis. It seeks to understand the geometric interplay of the structure-promoting functional $f$ and the measurement matrix $\Meas$ in a neighborhood of the signal vector $\gt$. To that end, we first introduce the following notions of descent cones and minimum conic singular values.
%

\begin{definition}[Descent cone]
 Let $f \colon \R^n \to \R$ be a convex function and let $\gt \in \R^n$. The \emph{descent set} of $f$ at $\gt$ is given by
 \begin{equation}
  \ds{f,\gt} \coloneqq \left\{ \vec{h} \in \R^n  : f ( \gt + \vec{h} ) \leq f (\gt)  \right\},
  \end{equation}
  and the corresponding \emph{descent cone} is defined by $\dc{f,\gt} \coloneqq \cone{\ds{f,\gt}}$.
\end{definition}
  
The notion of minimum conic singular values describes the behavior of a matrix $\Meas$ when it is restricted to a cone $C \subseteq \R^n$.
  
\begin{definition}[Minimum conic singular value]
  Consider a matrix $\Meas \in \R^{m \times n}$ and a cone $C \subseteq \R^n$. The minimum conic singular value of $\Meas$ with respect to the cone $C$ is defined by 
  \begin{equation}
   \lmin{\Meas}{C} \coloneqq \inf_{\x \in C \cap \Sn{n-1}} \norm{\Meas \x}_2.
  \end{equation}
\end{definition}

The following result characterizes exact recoverability of the signal $\gt$ and provides a deterministic error bound for the solutions to \eqref{eq:gen_bp}. The statement is an adapted version of Proposition 2.1 and Proposition 2.2 in \cite{chandrasekaran2012geometry}; see also Proposition 2.6 in \cite{tropp2014convex}.

\begin{proposition}[A deterministic error bound for \eqref{eq:gen_bp}]
\label{prop:recover}
Assume that $\gt,\Meas,\y,\noise$ and $\noiseparam$ follow Model \ref{mod:noisy_meas} and let $f \colon \R^n \to \R$ be a convex function. Then the following holds true:

\begin{enumerate}
 \item[(a)] If $\eta=0$, exact recovery of $\gt$ by solving \genNoisefree{} is equivalent to $\lmin{\Meas}{\dc{f,\gt}}>0$. 
 \item[(b)] In addition, any solution $\sol$ of \eqref{eq:gen_bp} satisfies
\begin{equation}
\label{eq:gen_rec}
 \norm{\gt - \sol}_2 \leq \frac{2\noiseparam}{\lmin{\Meas}{\dc{f,\gt}}}.
\end{equation}
\end{enumerate}
\end{proposition}

\subsection{Conic Mean Width}
\label{sec:primer_gaussian_width}

While Proposition \ref{prop:recover} provides an elegant analysis of the solutions to the optimization problem \eqref{eq:gen_bp}, it can be difficult to apply. The notion of a minimum conic singular value is related to the concept of co-positivity \cite{hiriart2010variational} and its computation is known to be an NP-hard task for general matrices and cones \cite{murty1987some,hiriart2010variational}. 

However, when $\Meas$ is chosen at random, sharp estimates can be obtained by exploring a connection to the \emph{statistical dimension} or \emph{Gaussian mean width}. These geometric parameters stem from geometric functional analysis and convex geometry (e.g., see \cite{gordon1985gaussian,gordon1988escape,giannopoulos2004asymptotic,milman1984mstar}), but they also show up in \emph{Talagrand’s $\gamma_2$-functional} in stochastic processes \cite{talagrand2014chaining}, or under the name of \emph{Gaussian complexity} in statistical learning theory \cite{bartlett2003complexity}.
Their benefits for compressive sensing have first been exploited in \cite{rudelson2007,mendelson2007reconstruction}. More important for our work is their use in the more recent line of research \cite{stojnic2009gordon,chandrasekaran2012geometry,amelunxen2014edge,tropp2014convex}, which aims for non-uniform signal recovery statements.


\begin{definition}
 Let $K\subseteq \R^n$ be a set. 
 \begin{enumerate}
  \item[(a)] The \emph{(global) mean width} of $K$ is defined as
  \begin{equation}
   w (K) \coloneqq \E \left[ \sup_{\vec{h} \in K} \langle \g, \vec{h} \rangle \right],
  \end{equation}
  where $\g \sim \mathcal{N} (\vec{0}, \I)$ is a standard Gaussian random vector.
  \item[(b)] The \emph{conic mean width} of $K$ is given by 
  \begin{equation}
   \cmw{K} \coloneqq w (\cone{K} \cap \mathcal{S}^{n-1} ).
  \end{equation}
 \end{enumerate}
We refer to $\cmw{\ds{f,\gt}}$ as the \emph{conic mean width} of $f$ at $x_0$.
 \end{definition}

The next theorem is known as \emph{Gordon's Escape Through a Mesh} and dates back to \cite{gordon1988escape}. The version presented here follows from \cite{liaw2016randommat}.

\begin{theorem}[Theorem 3 in \cite{liaw2016randommat}]
\label{thm:versh}
 Assume that $\Meas$ satisfies the assumption in Model \ref{mod:meas_mod} and let $K\subseteq \Sn{n-1}$ be a set. Then there exists a numerical constant $\constant>0$ such that, for every $u>0$, we have
 \begin{equation}
 \label{eq:gordon}
  \inf_{\x \in K} \norm{\Meas \x}_2 > \sqrt{m-1} - \constant \cdot\gaussparam^2 \cdot (w (K) + u),
 \end{equation}
 with probability at least $1-e^{-u^2/2}$. If $\meas \sim \mathcal{N}(\vec{0},\I)$, we have $\constant=\gaussparam=1$. 
\end{theorem}

Thus, a straightforward combination the error bound in \eqref{eq:gen_rec} and the estimate in \eqref{eq:gordon} for the set $K = \dc{f,\gt} \cap \Sn{n-1}$ reveals that robust recovery via \eqref{eq:gen_bp} is possible if the number of sub-Gaussian measurements obeys
\begin{equation}
 m  \geq  \constant^2 \cdot  \gaussparam^4 \cdot \cmw[2]{\ds{f,\gt}} + 1.
\end{equation}
In the case of Gaussian measurements, it is known that this bound yields a tight description of the so-called \emph{phase transition} of \genNoisefree. Indeed, for a convex cone $C \subseteq \R^n$ it can be shown that $\lmin{\Meas}{C} = 0$ with high probability when $m \leq \cmw[2]{C} - c\cdot\cmw{C}$, where $c>0$ denotes a numerical constant. Applying this statement to the descent cone $\dc{f,\gt}$ reveals that exact recovery of $\gt$ by solving \genNoisefree{} fails with high probability when
\begin{equation}
 m \leq \cmw[2]{\ds{f,\gt}} - c\cdot \cmw{\ds{f,\gt}}.
\end{equation}
Hence, exact signal recovery by solving \genNoisefree{} obeys a sharp phase transition at $m \approx \cmw[2]{\ds{f,\gt}}$ Gaussian measurements. We refer to \cite{amelunxen2014edge} and \cite[Remark 3.4]{tropp2014convex} for more details on this matter and conclude our discussion by the following summary:


\begin{highlight}
 Robust signal recovery via the generalized basis pursuit \eqref{eq:gen_bp} is characterized by the minimum conic singular value $\lmin{\Meas}{\dc{f,\gt}}$. The required number of sub-Gaussian random measurements can be determined by the conic mean width of $f$ at $\gt$, in symbols  $\cmw[2]{\ds{f,\gt}}$. 
\end{highlight}

\section{Coefficient and Signal Recovery}
\label{sec:coef_sig}

Our study of the synthesis formulation in this section is based on the differentiation between coefficient and signal recovery. First, we introduce the set of minimal $\ell^{\smash{1}}$-representers in Section~\ref{sec:min_l1} and discuss its importance for the relationship between both formulations. Section~\ref{sec:sig_recov_cvx_gauge} is then dedicated to the fact that signal recovery via \eqref{eq:sig} can be cast as an instance of atomic norm minimization, in which the gauge of the synthesis defining polytope is minimized. Finally, in Section~\ref{sec:sig_and_coef_rates}, we derive two non-uniform recovery theorems that determine the sampling rates of robust coefficient and signal recovery, respectively. 

\subsection{Recovery and Minimal $\ell^1$-Representers}
\label{sec:min_l1}

In this section, we discuss how the uniqueness of a minimal $\ell^1$-representer impacts coefficient and signal recovery. 

\begin{definition}[Minimal $\ell^1$-representers]
The set of \emph{minimal $\ell^1$-representers} of a signal $\gt$ with respect to a dictionary $\Dict$ is defined by
\begin{equation}
 \label{eq:minl1}
 \tag{$\text{BP}_{\ell^1}$}
 \Zlset \coloneqq \argmin_{\z \in \R^d} \norm{\z}_1 \quad \mbox{ s.t. } \quad \gt = \Dict  \z. 
\end{equation} 
\end{definition}

In general, $\Zlset$ may not be a singleton. Indeed, a coefficient vector $\zl$ can only be the unique minimal $\ell^{\smash{1}}$-representer of the associated signal $\gt = \Dict \zl$, if the set of atoms $\{ \dict_i : i \in \supp (\zl) \}$ is linearly independent \cite[Theorem 3.1]{foucart2013cs}. However, many dictionaries of practical interest possess linear dependent and coherent atoms.  Hence, typical notions that would certify uniqueness for all signals with sparse representations in $\Dict$ (e.g.,~the NSP~\cite[Theorem 4.5]{foucart2013cs}) are not expected to hold for such dictionaries.  

%

The following simple lemma shows that exact coefficient recovery by solving \coefNoisefree\ requires $\Zlset$ to be a singleton. 
Otherwise, it is impossible to recover a specific coefficient representation, while a retrieval of the signal by \sigNoisefree\ might still be possible.
\begin{lemma}
\label{lem:observation}
Assume that $\gt,\Meas$ and $\y$ follow Model \ref{mod:noisy_meas} with $\noiseparam=0$. 
Let $\Dict \in \R^{n\times d}$ be a dictionary such that $\gt \in \ran (\Dict)$. 

\begin{enumerate}

 \item[(a)] Assume that $\gt=\Dict \gtz$ and that we wish to reconstruct $\gtz$. If $\Zlset \neq \left\{ \gtz \right\}$, then recovering $\gtz$ by solving \coefNoisefree\ is impossible.
 
 \item[(b)]  Signal recovery by solving \sigNoisefree, i.e., having $\X = \left\{ \gt \right\}$, is equivalent to the condition $ \Zlset = \Zset$. 
 
\end{enumerate}
\end{lemma}
A short proof of the previous result is given in Appendix \ref{sec:proof_lem_observation}.
Thus, under the assumption that $\gt$ has a unique minimal $\ell^{\smash{1}}$-representer, exact coefficient recovery by \coefNoisefree\ and signal recovery by \sigNoisefree\ are equivalent. 

\subsection{Signal Recovery and the Convex Gauge}
\label{sec:sig_recov_cvx_gauge}


The literature on compressive sensing predominantly focuses on a recovery of coefficient representations. However, if the goal is to recover the associated signal, this approach may be insufficient for structured dictionaries, as argued previously. In this section, we express the initial optimization problem  over the coefficient domain  \eqref{eq:sig} as a minimization problem over the signal space. The $\ell^{\smash{1}}$-ball $\Bn[1]{d}$ in the coefficient domain is thereby mapped to the convex body $\Dict\cdot\Bn[1]{d}$, which is referred to as \emph{synthesis defining polytope} in \cite{Elad2006}. The formulation \eqref{eq:sig} can be equivalently expressed as a constrained minimization of its corresponding \emph{convex gauge}.

\begin{definition}[Convex gauge]
Let $K \subseteq \R^n$ be a closed convex set that contains the origin. The \emph{gauge} of $K$ (also referred to as \emph{Minkowski functional}) is defined as
\[
p_K (\x) \coloneqq \inf \left\{ \lambda >0 : \x \in \lambda\cdot K \right\}.
\]
For a symmetric set (i.e., $-K = K$) the gauge defines a semi-norm on $\R^n$, which becomes a norm if $K$ is additionally bounded.  
\end{definition}

The following lemma provides an alternative characterization of the solutions $\X$ to \eqref{eq:sig}. 
\begin{lemma}
\label{lem:gauge_formulation}
Assume that $\gt,\Meas,\y,\noise$ and $\noiseparam$ follow Model \ref{mod:noisy_meas} and let $\Dict\in \R^{n\times d}$ be a dictionary. Then we have:
\begin{equation}
 \label{eq:gauge}
\X= \argmin_{\x \in \R^n} p_{\Dict \cdot \Bn[1]{d}} (\x) \quad \mbox{ s.t. } \quad \norm{\y - \Meas \x}_2 \leq \noiseparam.
\end{equation}
\end{lemma}
A short proof for his equivalence is given in Appendix \ref{sec:proof_gauge_formulation}.
Under the heading of \emph{atomic norm minimization}, problems of the form~\eqref{eq:gauge} were previously considered in greater generality in \cite{chandrasekaran2012geometry}: Given a collection of atoms $\mathcal{A} \subseteq \R^n$, Chandrasekaran et al.~study the geometry of signal recovery based on minimizing the associated gauge $p_{\convhull{\mathcal{A}}}$ in \eqref{eq:gauge}. It turns out that many popular methods such as classical $\ell^{\smash{1}}$-, or nuclear norm-minimization can be cast in such a form, e.g.,~by choosing $\mathcal{A}$ as the set of one-sparse unit-norm vectors, or the set of rank-one matrices with unit-Euclidean-norm. Note that in the considered case of signal recovery via \eqref{eq:sig}, one would choose the atoms $\mathcal{A} = \left\{\pm \dict_i : i \in [d] \right\}$ to obtain $\convhull{\mathcal{A}} = \Dict \cdot \Bn[1]{d}$. With this reformulation it is evident that dictionary atoms that are convex combinations of the remaining atoms in $\mathcal{A}$ can be removed without altering $\X$~\cite[Corollary 1]{Elad2006}.

For some specific problem instances novel sampling rate bounds are derived in \cite{chandrasekaran2012geometry}. Although this work plays a key role for the foundation of our work, we wish to emphasize that no explicit insights or bounds are derived in the case of signal recovery with dictionaries. In particular, the connection between \eqref{eq:sig} and \eqref{eq:coef} has not been studied.  


With regard to the recovery framework of Section \ref{sec:primer_conic_sing_value}, it is of interest to determine the descent cone of the functional $p_{\Dict \cdot \Bn[1]{d}}$ at $\gt$. The following lemma shows how this cone is related to descent cones of the $\ell^{\smash{1}}$-norm in the coefficient space.
\begin{lemma}
 \label{lem:dc}
 Let $\Dict \in \R^{n \times d}$ be a dictionary and let $\gt \in \ran (\Dict)$. For any $\zl \in \Zlset$ we have 
 \[
  \dc{p_{\Dict\cdot\Bn[1]{d}},\gt} = \Dict \cdot  \dc{\norm{\cdot}_1,\zl} \quad \mbox{ and } \quad \ds{p_{\Dict\cdot\Bn[1]{d}},\gt} = \Dict \cdot  \ds{\norm{\cdot}_1,\zl}.
 \]
\end{lemma}
The proof is given in Appendix \ref{sec:proof_lem_dc}.

\subsection{Sampling Rates for Signal and Coefficient Recovery}
\label{sec:sig_and_coef_rates}

The purpose of this section is to determine the sampling rates for robust coefficient and signal recovery from sub-Gaussian measurements. 
\paragraph{Coefficient recovery}

With Lemma \ref{lem:observation} in mind, studying coefficient recovery is meaningful only if the signal $\gt$ has a unique minimal $\ell^1$-representer $\zl$ with respect to $\Dict$. Proposition~\ref{prop:recover} implies that this condition can be equivalently expressed by
\begin{equation}
\label{eq:central_coef}
 \lmin{\Dict}{\dc{\Onenorm,\zl}} > 0.
\end{equation}
Equipped with this assumption, we now state our main theorem regarding the recovery of coefficient vectors via~\eqref{eq:coef}. 
\begin{theorem}[Coefficient recovery]
\label{thm:coeff}
 Assume that $\gt, \Meas, \y, \noise$ and $\noiseparam$ follow Model~\ref{mod:noisy_meas}, where $\Meas$  is drawn according to the sub-Gaussian Model~\ref{mod:meas_mod} with sub-Gaussian norm $\gamma$. 
 Let $\Dict \in \R^{n\times d}$ be a dictionary and $\zl \in \R^d$ be a coefficient vector for the signal $\gt = \Dict \zl \in \R^n$, such that 
 \begin{equation*}
\lmin{\Dict}{\dc{\Onenorm,\zl}}> 0.  
 \end{equation*}
Then there exists a numerical constant $\constant>0$ such that for every $u>0$, the following holds true with probability at least $1-\e^{-u^2/2}$: If the number of measurements obeys 
\begin{equation}
\label{eq:m_coef}
 m>m_0 \coloneqq \constant^2 \cdot \gaussparam^4 \cdot \left(  \cmw{\Dict \cdot \ds{\Onenorm;\zl}}  + u \right)^2 + 1,
\end{equation}
then any solution $\solz$ to the program \eqref{eq:coef} satisfies 
\begin{equation}
\label{eq:rec_coef}
 \norm{\zl - \solz}_2 \leq \frac{2\noiseparam}{\lmin{\Dict}{\dc{\Onenorm;\zl}} \cdot (\sqrt{m-1} - \sqrt{m_0 - 1})}.
\end{equation}
If $\meas \sim \mathcal{N} (\vec{0},\I)$, then $\constant = \gaussparam = 1$.
\end{theorem}
A proof is given in Appendix \ref{sec:proof_coeff}.
Before turning towards signal recovery, let us highlight a few observations regarding the previous theorem.
\begin{remark}
\label{rem:coef}
 \begin{enumerate}
  \item[(a)] Note that Theorem~\ref{thm:coeff} does not assume anything on the dictionary $\Dict$ and the coefficient representation $\zl$, except for $\lmin{\Dict}{\dc{\Onenorm,\zl}} > 0$, which is a necessary condition for the theorem to hold true.
 As pointed out above, it reflects that $\zl$ is a unique $\ell^{\smash{1}}$-representer of $\gt$ with respect to $\Dict$, i.e., that $\zl$ is the unique solution to \eqref{eq:minl1}. In general, verifying this property is involved (cf.~the discussion in Section~\ref{sec:primer_gaussian_width}) and forms a trail of research on its own, e.g., see \cite[Chapter 12]{mallat09} or \cite[Chapter 9]{casazza2012finite}. In this regard, we think that an important contribution of Theorem~\ref{thm:coeff} is that it allows to isolate the minimum prerequisite of a unique $\ell^{\smash{1}}$-representer in $\Dict$ from the actual task of compressive coefficient recovery.

  \item[(b)] Equation \eqref{eq:m_coef} identifies $\cmw[2]{\Dict \cdot \ds{\Onenorm;\zl}}$ as the essential component of the sampling rate for coefficient recovery by~\eqref{eq:coef}. Indeed, the proof reveals (in combination with the discussion subsequent to Theorem~\ref{thm:versh}) that $m_0$ is a tight description of the required number of noiseless Gaussian measurements for exact recovery. 

  \item[(c)] Lastly, the error bound \eqref{eq:rec_coef} shows that coefficient recovery is robust to measurement noise, provided that $\lmin{\Dict}{\dc{\Onenorm,\zl}} \gg 0$; cf.~the numerical experiments in Section \ref{sec:num_exp}, which confirm this observation. However, we note that this bound might not be tight, in general (cf.~the intermediate inequality \eqref{eq:esti} in the proof, which is not necessarily sharp). 
  
 \end{enumerate}
\end{remark}

\paragraph{Signal recovery}

Considering signal recovery by \eqref{eq:sig}, 
a combination of the gauge formulation \eqref{eq:gauge}, its description of the descent cone in Lemma \ref{lem:dc}, and Theorem \ref{thm:versh} directly yields the next result.

\begin{theorem}[Signal recovery]
 \label{thm:sig}
 Assume that $\gt, \Meas, \y, \noise$ and $\noiseparam$ follow the measurement Model~\ref{mod:noisy_meas}, where $\Meas$  is drawn according to the sub-Gaussian Model~\ref{mod:meas_mod} with sub-Gaussian norm $\gamma$. Let $\Dict \in \R^{n\times d}$ be a dictionary with $\gt \in \ran (\Dict)$ and pick any $\zl \in \Zlset$. 
 
 \noindent Then there exists a numerical constant $\constant>0$ such that for every $u>0$, the following holds true with probability at least $1-\e^{-u^2/2}:$ If the number of measurements obeys 
\begin{equation}
\label{eq:m_sig1}
 m>m_0 \coloneqq \constant^2 \cdot \gaussparam^4 \cdot \left(  \cmw{\Dict \cdot \ds{\Onenorm;\zl}}  + u \right)^2 + 1,
\end{equation}
then any solution $\sol$ to the program \eqref{eq:sig} satisfies 
\begin{equation}
\label{eq:rec_sig1}
 \norm{\gt - \sol}_2 \leq \frac{2\noiseparam}{\sqrt{m-1} - \sqrt{m_0 - 1}}.
\end{equation}
If $\meas \sim \mathcal{N} (\vec{0},\I)$, then $\constant = \gaussparam = 1$.
\end{theorem}

Let us discuss the previous result in view of its counterpart for coefficient recovery, Theorem~\ref{thm:coeff}.

\begin{remark}
\begin{enumerate}
 \item[(a)] Similarly as for coefficient recovery, \eqref{eq:m_sig1} identifies $\cmw[2]{\Dict \cdot \ds{\Onenorm;\zl}}$ as the main quantity of the sampling rate for signal recovery by~\eqref{eq:sig}. An important difference is that the set minimal of $\ell^{\smash{1}}$-representers is not required to be a singleton: The descent cone in the signal space may be evaluated at any possible $\zl \in \Zlset$ and the resulting sampling rate for signal recovery does not depend on this choice.
 
 \item[(b)] In the case of noiseless Gaussian measurements, the number $m_0$ is a tight description of the phase transition of signal recovery, cf. the discussion subsequent to Theorem \ref{thm:versh}. 
 
 
 \item[(c)] While the sampling rates for coefficient and signal recovery are similar, the error bounds of the two theorems differ. The inequality~\eqref{eq:rec_sig1} does not involve the minimal conic singular value as in Theorem~\ref{thm:coeff}. This suggests the following noteworthy consequence: In the case of simultaneous coefficient and signal recovery, the robustness to noise of \eqref{eq:coef} and \eqref{eq:sig} might still be different. Indeed, while a reconstruction of $\gt$ is independent of the value of $\lmin{\Dict}{\dc{\Onenorm,\zl}}$ -- in fact, even 0 is allowed--, the error with respect to $\zl$ is directly influenced by it. We emphasize that the bound~\eqref{eq:rec_sig1} cannot be retrieved from the analysis conducted for coefficient recovery. Indeed, the estimate~\eqref{eq:rec_coef} of Theorem~\ref{thm:coeff} only implies that
\begin{equation}
 \norm{\gt -\sol}_2 = \norm{\Dict (\zl - \solz)}_2 \leq \frac{\|\Dict\|_2}{\lmin{\Dict}{\dc{\Onenorm,\zl}}} \cdot \frac{2\noiseparam}{\sqrt{m-1} - \sqrt{m_0-1}},
\end{equation}
which is worse than \eqref{eq:rec_sig1}, in general. 
\end{enumerate}
\end{remark}

While the bound \eqref{eq:m_sig1} is accurate for an exact recovery from noiseless measurements, it can be improved when an approximate recovery of $\gt$ is already sufficient. This is reflected by the following proposition on \emph{stable recovery}, which is an adaptation of a result in \cite{genzel2017}; see Appendix~\ref{sec:proof_stable} for a proof. Note that such an argumentation does not allow for a similar statement about stable coefficient recovery, due to the product $\Meas\Dict$ in \eqref{eq:coef}.  

\begin{proposition}[Stable signal recovery\label{prop:epsilongaussian}]
   Assume that $\gt, \Meas, \y, \noise$ and $\noiseparam$ follow the measurement Model~\ref{mod:noisy_meas}, where $\Meas$  is drawn according to the sub-Gaussian Model~\ref{mod:meas_mod} with sub-Gaussian norm $\gamma$. Let $\Dict \in \R^{n\times d}$ be a dictionary with $\gt = \Dict \gtz$. For a desired precision $\varepsilon>0$ let 
   \begin{equation}
   \label{eq:suro}
    \za \in \argmin_{\substack{\z: \norm{\gt - \Dict\z}_2\leq \varepsilon \\ \norm{\gtz}_1 = \norm{\z}_1}} \cmw{\Dict \cdot \ds{\norm{\cdot}_1;\z}}.
   \end{equation}

\noindent Then there exists a numerical constant $c>0$ such that for every $r>0$ and $u>0$ the following holds true with probability at least $1-\e^{-u^2/2}$:  
If the number of measurements obeys 
 \begin{equation}
 \label{eq:improvedsamplingrate}
 m > \tilde m_0 \coloneqq \constant^2 \cdot \gaussparam^4 \cdot \left( \frac{r+1}{r} \cdot [\cmw{\Dict \cdot \ds{\Onenorm;\za}} +1] +u \right)^2 +1,
 \end{equation}
 then any solution $\sol$ to \eqref{eq:sig} satisfies
 \begin{equation}
 \label{eq:maxerror}
  \norm{\gt - \sol}_2 \leq \max\left(r\varepsilon,\frac{2\noiseparam}{\sqrt{m-1} - \sqrt{\tilde m_0 - 1}}\right).
 \end{equation}
 If $\meas \sim \mathcal{N} (\vec{0},\I)$, then $\constant = \gaussparam = 1$.
\end{proposition}

The previous result extends Theorem~\ref{thm:sig} by an intuitive trade-off regarding stable signal recovery: By allowing for a lower recovery precision $\varepsilon>0$, the number of required measurements $\tilde{m}_0$ can be significantly lowered in comparison to $m_0$ in~\eqref{eq:m_sig1}. Indeed, \eqref{eq:suro} searches for surrogate representations $\za$ of $\gt$ in $\Dict$ that yield a minimal sampling rate. Note that the original coefficient vector $\gtz$ is not required to be a minimal $\ell^1$-representer of $\gt$ with respect to $\Dict$. Thus, Proposition~\ref{prop:epsilongaussian} enables to trade off the required number of measurements against the desired recovery accuracy. The factor $r>0$ is an additional oversampling parameter that may assist in balancing out this trade-off.   

We emphasize that this approach to stability is centered around a Euclidean approximation in the signal domain $\R^n$. This is in stark contrast to a stability theory in the coefficient domain, which is typically based on an approximation of compressible vectors by ordinary \emph{best $s$-term approximations}. We refer to Section 2.4 and 6.1 in~\cite{genzel2017} as well as Section 2.4 in \cite{Genz2020} for more details on the presented approach to stable recovery and related results in the literature. 

\begin{remark}
The normalization condition $\norm{\gtz}_1 = \norm{\z}_1$ in \eqref{eq:suro} can be discarded at the expense of a slightly worse error bound. Under the same conditions as in Proposition \ref{prop:epsilongaussian}, we can also derive the following result. Set $\varepsilon>0$ and let 

$$
\xa \in \mathop{\mathrm{argmin}}_{\|\x-\gt\|_2\leq \varepsilon} \cmw{\ds{p_{\Dict \cdot \Bn[1]{d}};\x}} \quad \mbox{ and } \quad K=\ds{p_{\Dict \cdot \Bn[1]{d}},\xa}.
$$
For every $u>0$ and
$$
m > \tilde m_0 := c^2\cdot \gamma^4 \cdot (\cmw{K} + u)^2 +1,
$$
the inequality  
\begin{equation}\label{eq:second_epsilon_width}
\| \gt - \sol\|_2 \leq \frac{4(\eta + \varepsilon (\sqrt{m} +C \gamma^2 + u))}{\sqrt{m-1} - \sqrt{\tilde m_0-1}} + \varepsilon
\end{equation}
holds true with probability larger than $1-2\e^{-u^2}$ for some constant $C$ depending only on the distribution of $\Meas$. For instance, letting $m=r \tilde m_0$ with $r>1$ as an oversampling factor, the latter bound essentially becomes:
	$$
	\|\gt -\sol \|_2 \lesssim \frac{4\eta}{(\sqrt{r}-1) \tilde m_0 } + \frac{4\varepsilon \sqrt{r}}{(\sqrt{r}-1)} + \varepsilon.
	$$
\end{remark}

We conclude this section by an illustration of stable recovery in two simple examples.

\begin{example}
\begin{enumerate}[(a)]
 
 \item Assume that $\Dict=\Id$ and let $\gt \in \R^n$ denote a fully populated vector, which is without loss of generality assumed to be positive and nonincreasing. Standard results on the computation of the conic mean width (see for instance~\cite[Example 4.3]{tropp2014convex}) stipulate that $\cmw[2]{\ds{\Onenorm;\gt}}=n$. Hence, it is impossible to exactly recover $\gt$ from noiseless compressive measurements. However, if we are satisfied with an approximate recovery of $\gt$, we can set the precision for instance to $\varepsilon=3\cdot\sigma_s(\gt)_1/\sqrt{s}$, where $\sigma_s(\gt)_p$ denotes the $\ell^p$-error of the best $s$-term approximation to $\gt$. Then, the surrogate vector $\xa \in \R^n$ defined as $x^\ast_{i} \coloneqq x_{0,i} + \sigma_s(\gt)_1 / s$ for $i=1,\dots,s$ and $x^\ast_i \coloneqq 0$ for $i=s+1,\dots,n$, satisfies $\norm{\xa}_1 = \norm{\gt}_1$.  A straightforward calculation shows that $\norm{\xa  - \gt}_2^2 \leq \sigma_s(\gt)_1^2/s + \sigma_s(\gt)_2^2$, which eventually leads to $\norm{\xa - \gt}_2 \leq 3\cdot\sigma_s(\gt)_1/\sqrt{s}$. Furthermore, a computation of the conic mean width yields that $ \cmw[2]{\ds{\Onenorm;\xa}} \lesssim 2s\log(n/s)$.
 Hence, Proposition~\ref{prop:epsilongaussian} shows that \sigNoisefree{} allows for the reconstruction of an approximation $\sol$ from $m \gtrsim 2s\log(n/s)$ noiseless sub-Gaussian measurements that satisfies $\norm{\gt - \sol}_2 \lesssim 3 \cdot\sigma_s(\gt)_1/\sqrt{s}$. A comparison with \cite[Theorem 4.22 and Section 11.1]{foucart2013cs} shows that such a stability result is essentially optimal.  
 
 \item Let us provide another simple result highlighting the important difference between signal and coefficient recovery. Consider a dictionary consisting of a convolution with a low pass filter $\h\in \R^n$, i.e., for any $\z\in \R^d$, $\Dict \z = \h \star \z$. The problem \eqref{eq:sig} then becomes a \emph{deconvolution} problem and it could be turned to a super-resolution problem by considering non integer shifts of the kernel. In this setting, it is well known \cite{candes2012} that coefficients of the form $\z_0=[1,-1,0,\hdots,0]$ are hard to recover by solving \eqref{eq:coef} since $\h\star \z_0\simeq 0$. There is a minimum separation distance to respect to guarantee the recovery of sparse spikes with arbitrary signs. We can however use the result \eqref{eq:second_epsilon_width} by setting $\varepsilon=\|\Dict \z_0\|_2$. In that case, we obtain $\tilde{m}_0=\constant^2\gamma^4u^2+1$ by picking $\x^\star=0$. Hence, we can recover an $\varepsilon$-approximation of $\gt=\Dict \z_0$ with a few measurements.  

\end{enumerate}
\end{example}

\section{Upper Bounds on the Conic Gaussian Width}
\label{sec:bounds}

The previous results identify the conic mean width $\cmw[2]{\Dict \cdot \ds{\Onenorm;\zl}}$ as the key quantity that controls coefficient and signal recovery by $\ell^{\smash{1}}$-synthesis. However, this expression does not convey an immediate understanding without further simplification. While tight and informative upper bounds are available for simple dictionaries such as orthogonal matrices, the situation becomes significantly more involved for general, possibly redundant transforms. 
Indeed, note that the polar cone of $\Dict \cdot \ds{\Onenorm;\zl}$ is given by $(\Dict \cdot \ds{\Onenorm;\zl})^\circ = (\Dict^T)^{-1} (\ds{\Onenorm;\zl}^\circ)$. The appearance of the preimage $(\Dict^T)^{-1}$ hinders the application of the standard approach based on polarity; see for instance \cite[Recipe 4.1]{amelunxen2014edge}.

Hence, the goal of this section is to provide two upper bounds for $\cmw[2]{\Dict \cdot \ds{\Onenorm;\zl}}$ that are more accessible and intuitive: Section~\ref{sec:condbound} is based on a local conditioning argument, and addresses recovery when a unique minimal $\ell^{\smash{1}}$-representer exists. The second bound of Section~\ref{sec:circumangle} follows a geometric analysis that explores the thinness of high-dimensional polyhedral cones with not too many generators. This approach possesses a broader scope and plays a central role in our work.

\subsection{A Condition Number Bound}
\label{sec:condbound}
In this section, we aim at ``pulling'' the dictionary $\Dict$ ``out of'' the expression $\cmw[2]{\Dict \cdot \ds{\Onenorm;\zl}}$, in order to make use of the fact that $\cmw[2]{\ds{\Onenorm;\zl}}$ is well understood. 
We begin by introducing the following notation of a local condition number.
\begin{definition}[Local condition number]
Let $\Dict \in \R^{n \times d}$ be a dictionary and let $C\subseteq\R^d$ be closed convex cone. Then, we define the \emph{local condition number of $\Dict$ with respect to $C$} by
\begin{equation}
 \condi{\Dict}{C} \coloneqq \frac{\norm{\Dict}_2}{\lmin{\Dict}{C}},
\end{equation}
with the convention $\condi{\Dict}{C}=+\infty$ if $\lmin{\Dict}{C}=0$. We also use the notation $\condi{\Dict}{\gtz} \coloneqq \condi{\Dict}{\dc{\Onenorm;\gtz}}$, which we refer to as \emph{local condition number of $\Dict$ at $\gtz$ with respect to the $\ell^{\smash{1}}$-norm}.
\end{definition}
Before the previous quantity will be used to simplify $\cmw[2]{\Dict \cdot \ds{\Onenorm;\zl}}$, we first comment on the origin of its name and give an intuitive interpretation of its meaning in the following remark.
\begin{remark}
\label{rem:cond}
\begin{enumerate}[(a)]
 \item First, recall that the classical, generalized condition number of a matrix is defined as the ratio of the largest and the smallest nonzero singular value. Hence, referring to $\condi{\Dict}{C}$ as a local condition number is motivated by the fact that it can also be written as
\begin{equation}
 \condi{\Dict}{C} = \frac{\norm{\Dict}_2}{\lmin{\Dict}{C}} = \frac{\lmax{\Dict}{\R^d}}{\lmin{\Dict}{C}},
\end{equation}
where $\lmax{\Dict}{\R^d} \coloneqq \max_{ \z \in \R^d \cap \Sn{d-1}} \norm{\Dict \z}_2 = \norm{\Dict}_2$ is the largest singular value of $\Dict$.   

\item Furthermore, note that $\condi{\Dict}{\gtz}$ acts as a local measure for the conditioning of $\Dict$ at $\gtz$ with respect to the $\ell^{\smash{1}}$-norm. It quantifies how robustly $\gtz$ can be recovered as the minimal $\ell^{\smash{1}}$-representer of $\gt = \Dict \gtz$: Consider the perturbation $\widehat{\gtz} = \gtz + \hat{\noise}$, where $\hat{\noise} \in \R^d$ with $\norm{\hat{\noise}}_2 \leq \hat{\noiseparam}$. Thus, in the signal domain we obtain $\norm{\Dict \gtz - \Dict \cdot \widehat{\gtz}}_2 = \norm{\Dict \hat{\noise}}_2 \leq \norm{\Dict}_2 \cdot \hat{\noiseparam}$. Proposition~\ref{prop:recover} then yields that any solution $\solz$ of the program
\begin{equation}
 \min_{z \in \R^d} \norm{\z}_1 \quad \mbox{ s.t. } \quad \norm{\Dict \widehat{\gtz} - \Dict \z}_2 \leq \norm{\Dict}_2 \cdot \hat{\noiseparam}
\end{equation}
satisfies 
\begin{equation}
\norm{\gtz - \solz}_2 \leq \frac{2 \cdot \norm{\Dict}_2 \cdot \hat{\noiseparam}}{\lmin{\Dict}{\dc{\Onenorm;\gtz}}} \lesssim \condi{\Dict}{\gtz} \cdot \hat{\noiseparam},
\end{equation}
which shows that $\condi{\Dict}{\gtz}$ can be seen as a measure for the stability of $\gtz$ with respect to $\ell^1$-minimization with $\Dict$.
\end{enumerate}
\end{remark}

The following proposition provides a generic upper bound for the conic mean width of a linearly transformed cone.
\begin{proposition}\label{prop:conditiongaussian}
Let $C \subseteq \R^d$ denote a closed convex cone. For any dictionary $\Dict \in \R^{n\times d}$, we have 
\begin{equation}
\label{eq:condbound}
 \cmw[2]{\Dict \cdot C} \leq \condi[2]{\Dict}{C} \cdot \left( \cmw[2]{C} + 1 \right).
\end{equation} 
\end{proposition}
\begin{proof}
See Appendix \ref{sec:proof_conditiongaussian}.
\end{proof}

Note that for sparse coefficient vectors $\gtz$ the quantity $\cmw[2]{\ds{\Onenorm;\gtz}}$ is well understood and has been frequently calculated in the literature, see for instance \cite[Example 4.3]{tropp2014convex}. It turns out that it can be bounded from above by
 \begin{equation}
  \cmw[2]{\ds{\Onenorm;\gtz}} \leq 2s \log (d/s) + 2s,
 \end{equation}
where $s = \# \supp (\gtz)$. Hence, we directly obtain the following corollary.

\begin{corollary}
\label{cor:coef}
If $\zl$ is the unique minimal $\ell^1$-representer of the associated signal $\gt = \Dict \zl$, the critical number of measurements $m_0$ in \eqref{eq:m_coef} and \eqref{eq:m_sig1} satisfies
\begin{equation}
 m_0 \leq c^2 \cdot \gamma^4 \cdot \left(  \condi{\Dict}{\zl} \cdot \left( \cmw[2]{\ds{\Onenorm;\zl}} + 1\right)^{1/2} +u \right)^2 + 1  \lesssim \condi[2]{\Dict}{\zl} \cdot s\log (d/s),
\end{equation}
where $s = \# \supp (\zl)$.
\end{corollary}

We have assumed $\zl$ to be a unique minimal $\ell^1$-representer since otherwise the previous statement becomes meaningless due to $\condi{\Dict}{\zl}=+\infty$. Thus, the condition number bound of Corollary~\ref{cor:coef} foremost addresses coefficient recovery via~\eqref{eq:coef}, as well as a reconstruction of signals with unique minimal $\ell^{\smash{1}}$-representers in $\Dict$ by \eqref{eq:sig}.
In both cases,  $m \gtrsim \condi[2]{\Dict}{\zl} \cdot s \log (d/s)$ sub-Gaussian measurements are sufficient (recall that the two formulations might nevertheless differ with respect to robustness to measurement noise). 
Hence, the results of this section identify the following three decisive factors for successful recovery: 
\begin{enumerate}[(i)]
 \item The uniqueness of $\zl$ as the minimal $\ell^1$-representer of $\gt = \Dict \zl$;
 \item The complexity of $\zl$ with respect to $\ell^1$-norm, which is measured by $\cmw[2]{\ds{\Onenorm;\zl}}$, or by its sparsity $s = \# \supp (\zl)$;
 \item The quantity $\condi{\Dict}{\zl}$, which resembles a local measure for the conditioning of $\Dict$ at $\zl$.
\end{enumerate}

We demonstrate in the numerical experiments of Section~\ref{sec:num_coef} that such a condition number approach might be accurate for some specific problems, however, it is overly pessimistic in general. Indeed, it is possible that $\cmw[2]{\Dict \cdot \ds{\Onenorm;\zl}} \leq \cmw[2]{\ds{\Onenorm;\zl}}$, but $\condi{\Dict}{\zl} \gg 1$. We suspect that a more accurate description might require a detailed analysis of random conic spectra \cite{seeger2003eigenvalues}.
 
\begin{remark}
 \begin{enumerate}[(a)]
  \item In the case $\Dict = \Id$, observe that Corollary~\ref{cor:coef} is consistent with standard compressed sensing results. Indeed, in this situation, it holds true that 
  \begin{equation}
   \condi{\Id}{\zl}  = 1 = \norm{\Id}_2  = \lmin{\Id}{\dc{\Onenorm,\zl}},
  \end{equation}
  implying that $m  \gtrsim s \log (n/s)$ measurements are sufficient for robust recovery of $s$-sparse signals.  
  
  \item During completion of this work, we discovered that similar bounds as~\eqref{eq:condbound} were recently derived in~\cite{amelunxen2018}. Amelunxen et al.~do not address the synthesis formulation of compressed sensing, but they study the statistical dimension of linearly transformed cones in a general setting. Their results are based on a notion of \emph{Renegar's condition number}, which can be defined as
\begin{equation}
\label{eq:ren}
\mathcal{R}_C (\Dict) = \min \left\{ \frac{\norm{\Dict}_2}{\lmin{\Dict}{C}}, \frac{\norm{\Dict}_2}{\sigma_{\R^n\to C}(-\Dict^T)}\right\}, 
\end{equation}
where $C\subseteq \R^d$ is a closed, convex cone, $\sigma_{\R^n\to C}(-\Dict^T) \coloneqq \min_{\x \in \Sn{n-1} } \norm{\vec{\Pi}_C (-\Dict^T \x)}_2$ and $\vec{\Pi}_C$ denotes the orthogonal projection on $C$. \cite[Theorem A]{amelunxen2018} then establishes the bound $\delta (\Dict \cdot C) \leq \mathcal{R}^2_C(\Dict) \cdot \delta (C)$, where $\delta$ denotes the \emph{statistical dimension}, which is essentially equivalent to the conic mean width; see proof of Proposition~\ref{prop:conditiongaussian} in Appendix~\ref{sec:proof_conditiongaussian} for details.

Additionally, the authors of \cite{amelunxen2018} provide a ``preconditioned'', probabilistic version of the latter bound: For $m \leq n$ let $\vec{P}_m$ denote the projection onto the first $m$ coordinates and define the quantity $\mathcal{R}_{C,m}^2 (\Dict) \coloneqq \E_{\vec{Q}} [\mathcal{R}_C (\vec{P}_m\vec{Q} \Dict)^2 ] $, where the expectation is with respect to a random orthogonal matrix $\vec{Q}$, distributed according to the normalized Haar measure on the orthogonal group.  \cite[Theorem B]{amelunxen2018} then states that for any parameter $\nu \in (0,1)$ and $m \geq \delta (C)  + 2\sqrt{\log (2 / \nu)m}$, we have that $\delta (\Dict \cdot C) \leq  \mathcal{R}_{C,m}^2 (\Dict) \cdot \delta(C) + (n-m) \cdot \nu$.
Due to the second term in~\eqref{eq:ren}, both versions of Renegar's condition number will be not greater than $\condi{\Dict}{C}$, in general. Hence, ignoring the dependence on $\nu$ and the condition on $m$ for simplicity, the bound on the required samples of Corollary~\ref{cor:coef} could also be formulated with $\mathcal{R}_{\dc{\Onenorm,\zl}}^2 (\Dict)$ or  $\mathcal{R}_{\dc{\Onenorm,\zl},m}^2 (\Dict)$ instead of $\condi{\Dict}{\zl}$. 
 \end{enumerate}
\end{remark}

\subsection{A Geometric Bound}
\label{sec:circumangle}

In this section, we derive an upper bound for $\cmw[2]{\Dict \cdot \ds{\Onenorm;\zl}}$ that is based on generic arguments from high-dimensional convex geometry. We exploit the fact that the cone $\Dict \cdot \dc{\Onenorm;\zl}$  is finitely generated by at most $2d$ vectors (see the proof of Proposition~\ref{prop:lindc} in Appendix~\ref{sec:proof_geom}) -- a number that is typically significantly smaller than exponential in the ambient dimension $n$. The resulting upper bound depends on the maximal sparsity of elements in $\Zlset$ and on a single geometric parameter that we refer to as \emph{circumangle}, whereas the number of generators only has a logarithmic influence. This is comparable to the mean width of a convex polytope, which is mainly determined by its diameter (cf.~Lemma~\ref{lem:convexpolywidth}) and by the logarithm of its number of vertices. 

In Section~\ref{sec:circumstart}, we first introduce the required notation and show an upper bound on the conic mean width of pointed polyhedral cones. We then focus on the geometry of the descent cone $\dc{p_{\Dict\cdot\Bn[1]{d}},\gt}$ (see Section~\ref{sec:geom_desc}) in order to derive the desired upper bound on the expression $\cmw[2]{\Dict \cdot \ds{\Onenorm;\zl}}$ in Section~\ref{sec:sr}. Finally, we show how this bound can be used in practical examples; see Section~\ref{sec:examples}. 

\subsubsection{The Circumangle}
\label{sec:circumstart}

The goal of this section is to relate the conic mean width of a pointed polyhedral cone to its \emph{circumangle}, which describes the angle of an enclosing circular cone. To that end, recall that a \emph{circular cone} (also referred to as \emph{revolution cone}) with \emph{axis} $\ax\in \Sn{n-1}$ and \emph{(half-aperture) angle} $\alpha \in [0,\pi/2]$ is defined as 
    \begin{equation}
     C(\alpha,\ax)\coloneqq \{\x\in \R^n, \langle \x,\ax\rangle \geq \norm{\x}_2 \cdot \cos(\alpha)\}.\label{eq:def_circular_cone}
    \end{equation}
The conic mean width of a circular cone depends linearly on the ambient dimension $n$, i.e., $\cmw[2]{C(\alpha,\ax)} = n\cdot \sin^2(\alpha)  + O(1)$, see for instance~\cite[Proposition 3.4]{amelunxen2014edge}. Although not directly related, it will be insightful to compare this result with the subsequent upper bound of Proposition~\ref{prop:circ}.

The following definition introduces the so-called circumangle of a nontrivial (different from $\left\{\vec{0}\right\}$ and $\R^n$) closed convex cone $C$. It describes the angle of the smallest circular cone that contains $C$. 
\begin{definition}[Circumangle]
    Let $C\subset \R^n$ denote a nontrivial closed convex cone. Its \emph{circumangle} $\alpha$ is defined by
    \begin{equation}
     \alpha \coloneqq \inf \left\{\hat{\alpha} \in [0,\pi/2] : \exists \ax \in \Sn{n-1}, C \subseteq C(\hat{\alpha},\ax)\right\}.\label{eq:def_circumangle}
    \end{equation}
\end{definition}

The previous notion can be found under various names in the literature, see for instance~\cite{Freund99,Renegar1993,Iusem2008, Henrion2010a}. In particular, the previous quantity arises in the definition of an outer center of a cone~\cite{Henrion10}. It turns out that the circumangle satisfies
\begin{equation}
 \cos(\alpha) = \sup_{\ax \in \Sn{n-1}} \inf_{\x\in C \cap \Sn{n-1}} \langle \ax , \x\rangle,
\end{equation}
where a vector $\ax$ that maximizes the right hand side is referred to as \emph{circumcenter} (or \emph{outer center}\footnote{Note that the notions of circumcenter and outer centers generally differ, however, in the Euclidean setting of this work they are equivalent \cite[Section 5]{Henrion10}.}) of $C$~\cite{Henrion10,Iusem2008}. Furthermore, if $C$ is pointed (i.e., if it does not contain a line), the circumcenter is unique and $\alpha \in [0,\pi/2)$~\cite{Henrion10}. 

Note that the function $\ax\mapsto\inf_{\x\in C \cap \Sn{n-1}} \langle \ax,\x\rangle$ is concave as a minimum of concave functions. Hence, if $C$ is pointed, it is easy to see that determining the circumcenter and the circumangle amounts to solving the following convex optimization problem:  
\begin{equation}
 \cos(\alpha) = \sup_{\ax \in \Bn[2]{n}} \inf_{\x\in C \cap \Sn{n-1}} \langle \ax , \x\rangle.
\end{equation}
We now show that this characterization can be further simplified for pointed polyhedral cones. The simple characterization of the following proposition makes it possible to numerically compute the circumangle of such cones. We emphasize that this stands in contrast to previously discussed notions such as the minimum conic singular value, which is intractable in general. A short proof is included in Appendix~\ref{sec:simple_circumcenter}.
\begin{proposition}[Circumangle and circumcenter of polyhedral cones\label{prop:simple_circumcenter}]
Let $\x_i \in \Sn{n-1}$ for $i\in[k]$ and let $C = \cone{\x_1,\hdots,\x_k}$ be a nontrivial pointed polyhedral cone.
Finding the circumcenter and circumangle of $C$ amounts to solving the convex problem:
\begin{equation}
    \cos(\alpha)= \sup_{\ax \in \Bn[2]{n}} \inf_{i \in [k]} \langle \ax,\x_i\rangle\label{eq:kpolyhedral}.
\end{equation}
\end{proposition}

The goal of this section is to upper bound the conic mean width of all polyhedral cones $C \subset \R^n$ with $k$ generators that are contained in a circular cone of angle $\alpha$. To that end, we first introduce the following notation:
\begin{definition}
    A \emph{$k$-polyhedral $\alpha$-cone} $C \subset\R^n$ is a nontrivial pointed polyhedral cone generated by $k$ vectors that is included in a circular cone with angle $\alpha\in [0,\pi/2)$.
    Furthermore, we let $\mathcal{C}_{k}^{\alpha}$ denote the set of all $k$-polyhedral $\alpha$-cones.
\end{definition}
Note that $C$ being a nontrivial pointed polyhedral cone implies that such an encompassing circular cone with angle $\alpha\in [0,\pi/2)$ exists.
The next result provides a simple upper bound on the quantity
\begin{equation}
    W(\alpha, k, n) := \sup_{C \in \mathcal{C}_k^{\alpha}, C\subset \R^n} \cmw{C}.    
\end{equation}
The underlying geometric idea is explained in Figure \ref{fig:alphacone} and its proof is detailed in Appendix \ref{sec:proof_circ}. Note that the bound does not depend on the ambient dimension $n$, which is in contrast to the conic width of a circular cone.

\begin{figure}
    \begin{center}
    \begin{subfigure}[t]{0.45\textwidth}
        \includegraphics{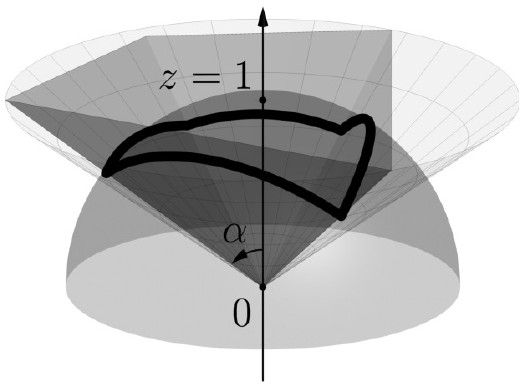}
  \end{subfigure}%
  \qquad
  \begin{subfigure}[t]{0.45\textwidth}
        \includegraphics{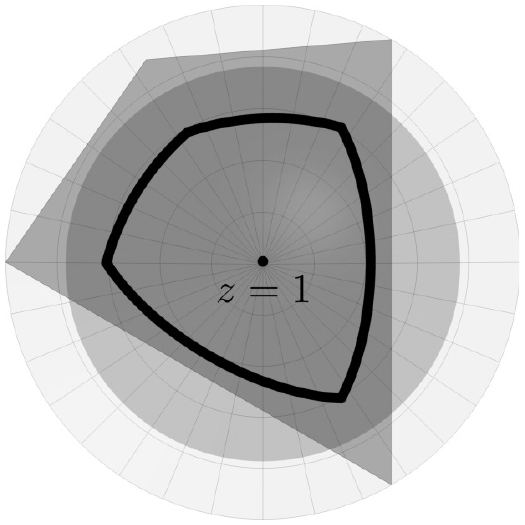}
    \end{subfigure}%
        \caption{\textbf{Geometry of Proposition~\ref{prop:circ}.} The figure shows a polyhedral cone (transparent gray) truncated at $z=1$ and included in a circular cone with angle $\alpha$ (wire-frame). The standard basis vector $\vec{e}_z$ corresponds to the circumcenter of the polyhedral cone and $\alpha$ is chosen as its circumangle. The thick line is the intersection of the unit sphere with the faces of the polyhedral cone. Right view is from above, or equivalently, the projection on the plane $z=1$.         
        The conic mean width of the polyhedral cone can be bounded by evaluating the mean width of any set containing the thick line plus 1 (as a subset of the plane). The proposed bound is based on using the intersection of the polyhedral cone and the plane $z=1$. Notice that this convex body is included in the disk with radius $\tan \alpha$. In high dimensions, it will be thin if the polyhedral cone does not have overwhelmingly many extremal rays.\label{fig:alphacone}}
    \end{center}
\end{figure}

\begin{proposition}
\label{prop:circ} 
   For $k\geq 5$, the conic mean width of a $k$-polyhedral $\alpha$-cone $C$ in $\R^n$ is bounded by
 \begin{align}
 \label{eq:polybound}
     W(\alpha, k, n)  &\leq  \tan \alpha \cdot \left( \sqrt{2\log \left(k/\sqrt{2\pi}\right)} + \frac{1}{\sqrt{2\log \left( k/\sqrt{2\pi}\right)}}  \right) +\frac{1}{\sqrt{2\pi}}. \\
 \end{align}
\end{proposition}

We conclude this section with the following remark:
\begin{remark}
\label{rem:poly}
    The previous upper bound is based on Lemma~\ref{lem:convexpolywidth}, which provides a basic bound on the Gaussian mean width of a convex polytope; see also \cite[Ex.~7.5.10 \& Prop.~7.5.2]{vershynin2018}. Using a tighter estimate there (possibly an implicit description as in~\cite[Proposition 4.5]{amelunxen2014edge}) would in turn also improve Proposition~\ref{prop:circ}. 
\end{remark}

\subsubsection{Geometry of the Descent Cone}
\label{sec:geom_desc}

In order to derive an upper bound for the quantity $\cmw[2]{\Dict \cdot \ds{\Onenorm;\zl}}$ based on the previous result, we first need a more geometrical description of the descent cones of the $\ell^{\smash{1}}$-norm and of the gauge $p_{\Dict\cdot\Bn[1]{d}}$. 
\begin{lemma}
\label{lem:dc_lone}
 Let $\z \in \R^d$ with support $\supp{\z} = \Supp$ and $\# \Supp = s$. Then,
 \begin{equation}
  \dc{\norm{\cdot}_1,\z} = \cone{\pm s \cdot \vec{e}_i - \vec{v} : i \in [d]}.
 \end{equation}
 where $\vec{v}$ is any vector such that $\|\vec{v}\|_1=s$ and $\sign \vec{v}=\sign \z$, e.g., $\vec{v}=\sign \z$ or $\vec{v}=s\cdot \z/\|\z\|_1$.
\end{lemma}

A proof of the previous lemma can be found in Appendix~\ref{sec:proof_geom}. The statement is illustrated in Figure \ref{fig:descentconel1} for dimension $d=3$.  Observe that sliding $\z$ along the edge linking $\eb_2$ with $\eb_3$ leaves the descent cone unchanged.

\begin{figure}
    \begin{center}
    \begin{subfigure}[t]{0.33\textwidth}
        \includegraphics[height=4.5cm,trim={1.5cm 60 50 55},clip]{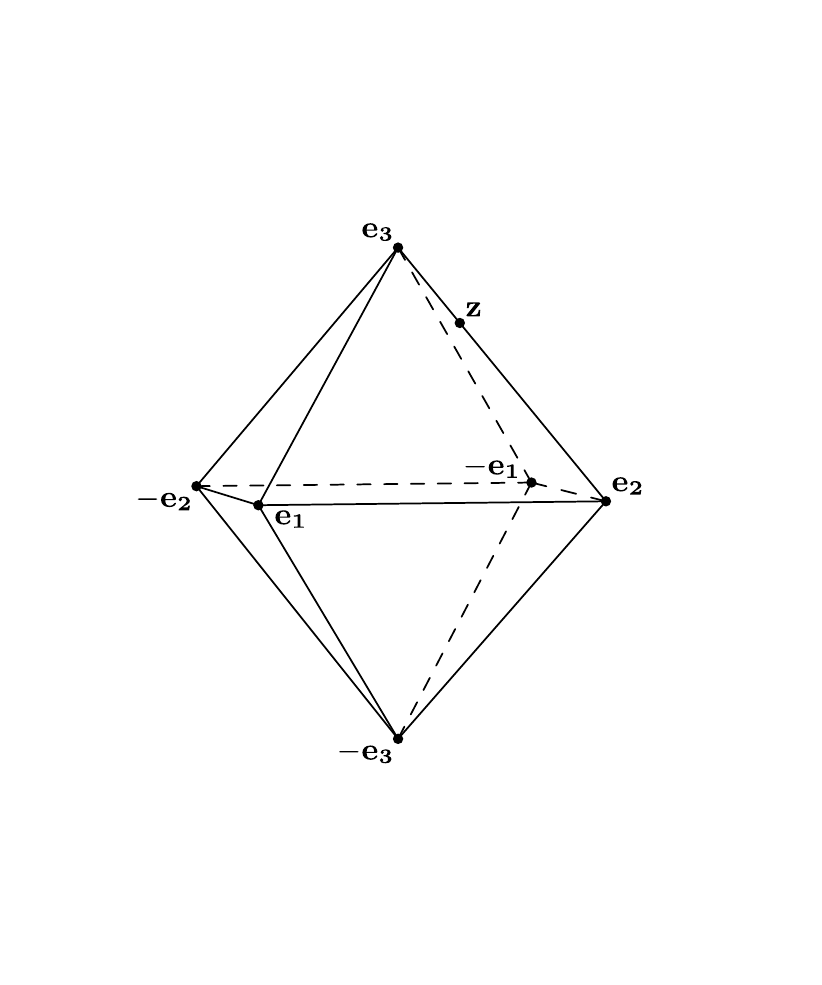}
  \end{subfigure}%
  \quad
  \begin{subfigure}[t]{0.33\textwidth}
        \includegraphics[height=4.5cm,trim={1.5cm 60 50 55},clip]{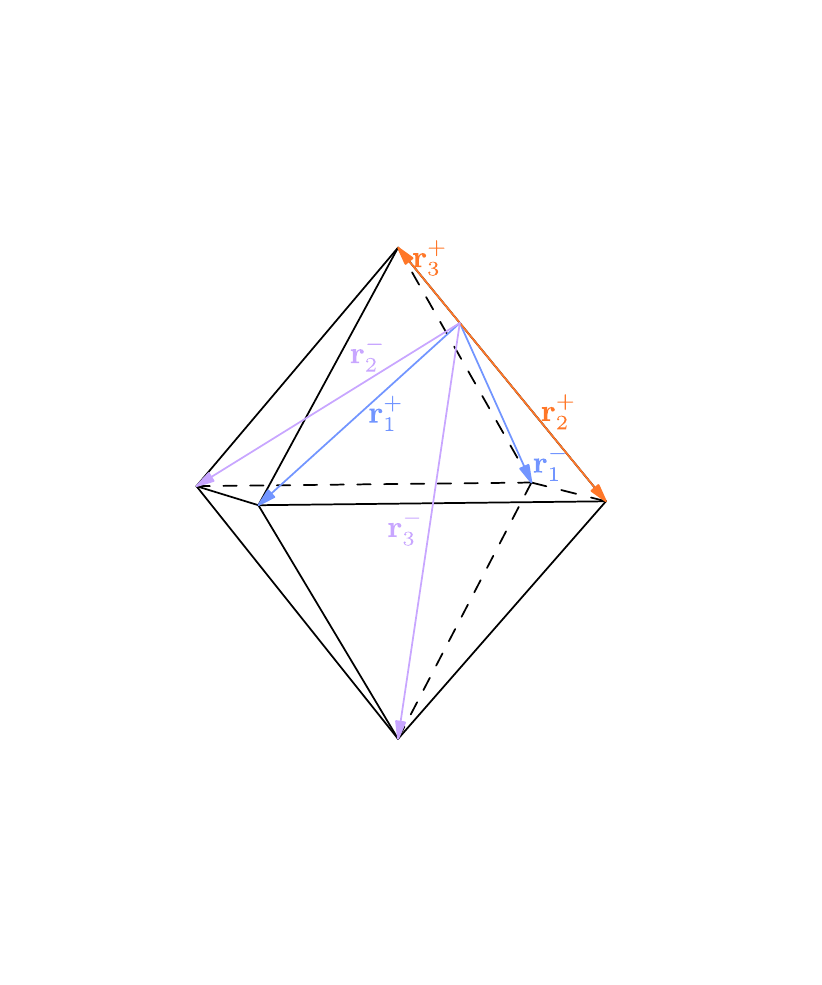}
    \end{subfigure}%
  \hspace{-0.9cm}
  \begin{subfigure}[t]{0.33\textwidth}
        \includegraphics[height=4.5cm,trim={0cm 60 50 55},clip]{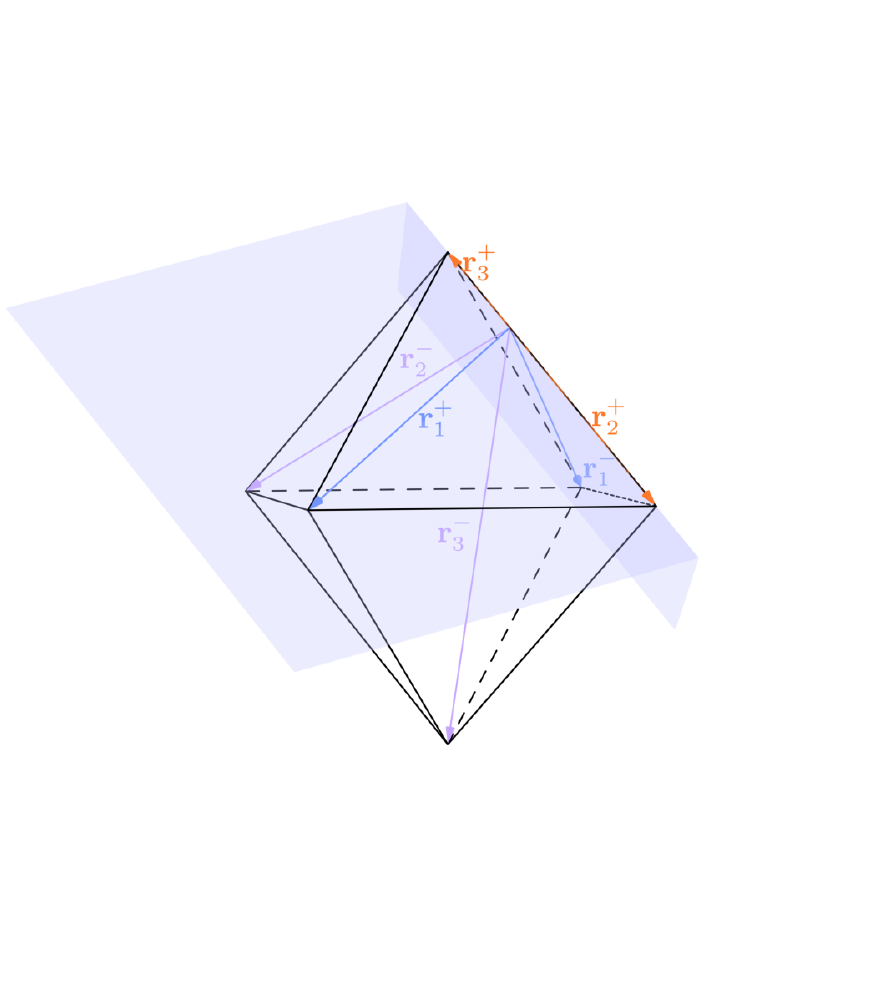}
    \end{subfigure}%
        \caption{\textbf{Illustration of Lemma~\ref{lem:dc_lone}.} Left: $\ell^1$-ball in $\R^3$ and a $2$-sparse vector $\z$. Center: The rays of the descent cone are supported by the vectors $\r_i^{\pm}=\pm \eb_i - \z$, which corresponds to the generators of Lemma~\ref{lem:dc_lone} with $\vec{v} = s\cdot\z /\norm{\z}_1$, multiplied by $1/2$. Right: The resulting descent cone (shifted by $\z$). Note that it contains a linear subspace spanned by  $\r_2^+$ and $\r_3^+$. \label{fig:descentconel1}}
    \end{center}
\end{figure}


Whenever a convex cone contains a subspace, its circumangle is $\pi/2$ and the bound of Proposition \ref{prop:circ} is not applicable. 
As can be seen in Figure \ref{fig:descentconel1}, the descent cone of the $\ell^1$-norm at $\vec{z}$ contains the subspace spanned by the face of minimal dimension containing $\vec{z}$.
To avoid this pitfall, let us recall the notion of lineality.

\begin{definition}[Lineality \cite{rockafellar2015convex}]
For a non-empty convex set $C \subseteq \R^n$, the \emph{lineality space $C_L$ of $C$} is defined as 
\begin{equation}
 C_L \coloneqq \left\{ \x \in \R^n :  \forall \xx \in C: \left\{ \xx + \alpha \cdot \x : \alpha \in \R \right\} \subseteq C \right\}.
\end{equation}
It defines a subspace of $R^n$ and its dimension is referred to as the \emph{lineality} of $C$.
\end{definition}

Any non-empty convex set $C$ can be expressed as the direct sum 
\begin{equation}\label{eq:decompotisionconvexset}
 C = C_L \oplus C_R \mbox{ with } C_R\coloneqq \proj_{C_L^\perp} (C),
\end{equation}
where $\proj_{\smash{C_L^\perp}}$ is the orthogonal projection onto $C_L^\perp$. The notation $C_R$ is used in analogy to the \emph{range} in linear algebra.  The set $C_R$ is ``line-free'' (i.e., its lineality is $\left\{\vec{0}\right\}$) and its dimension is called the \emph{rank} of $C$.  If $C$ is a convex cone, the lineality space is the largest subspace contained in $C$ and the circumangle of the range $C_R$ is less than $\pi/2$. In the following, we will therefore apply Proposition \ref{prop:circ} to the latter set only. Figure \ref{fig:decompositiondescentcone} illustrates the orthogonal decomposition of~\eqref{eq:decompotisionconvexset} for the descent cone of Figure~\ref{fig:descentconel1}. 

The following lemma characterizes the lineality space and the range for descent cones of the $\ell^1$-norm. A proof is given in Appendix~\ref{sec:proof_geom}.
\begin{lemma}
\label{lem:lineality}
 Let $\z =(z_1,\hdots,z_d)$ be a vector with support $\supp{\z} = \Supp$ and $\# \Supp = s \geq 1$. The lineality space of $C=\dc{\norm{\cdot}_1,\z}$ is then given by
 \begin{equation}
  C_L = \spann \left( s \cdot \sign (z_i)   \cdot \vec{e}_i  -\sign( \z): i \in \Supp \right),
 \end{equation}
 with lineality $\dim ( C_L ) = s-1$.
\end{lemma}
Notice that the lineality space is nothing but the face of the $\ell^{\smash{1}}$-ball of smallest dimension containing $\z$.

\begin{figure}
    \begin{center}
      \begin{subfigure}[t]{0.3\textwidth}
        \includegraphics[height=4cm,trim={0.5cm 60 50 50},clip]{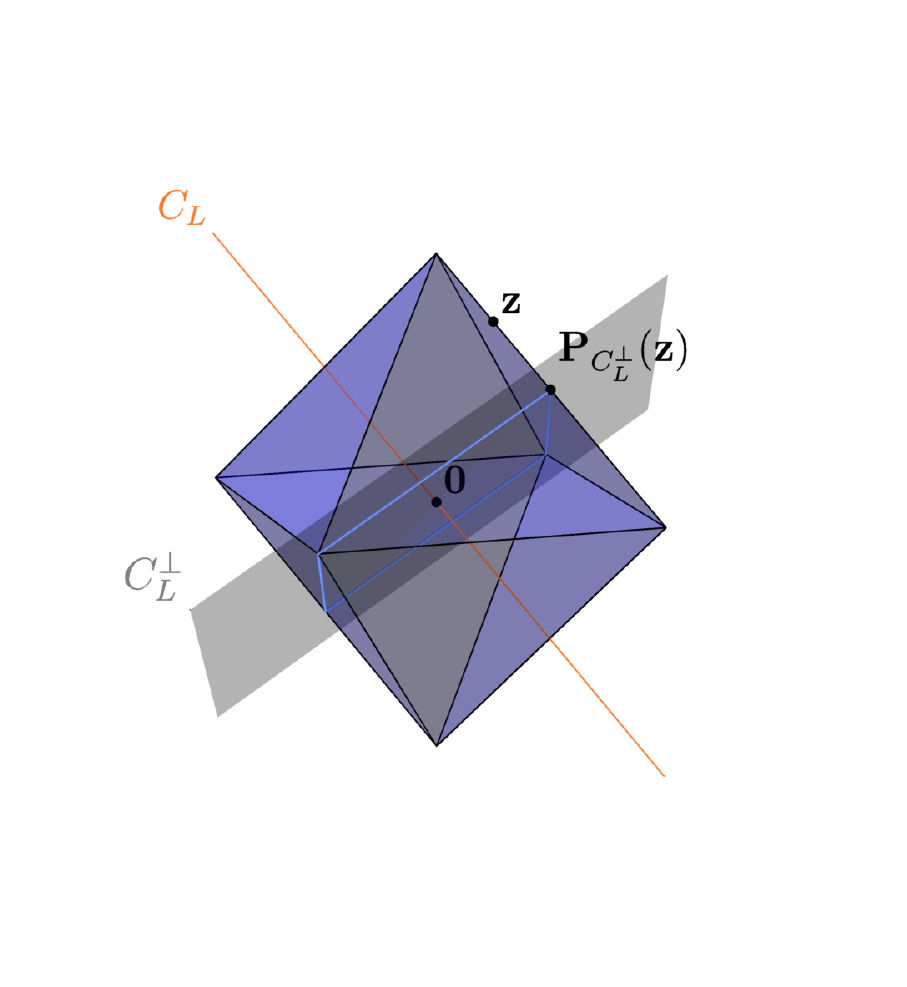}
    \end{subfigure}%
    \qquad
    \begin{subfigure}[t]{0.3\textwidth}
        \includegraphics[height=4cm,trim={0cm 0 50 100},clip]{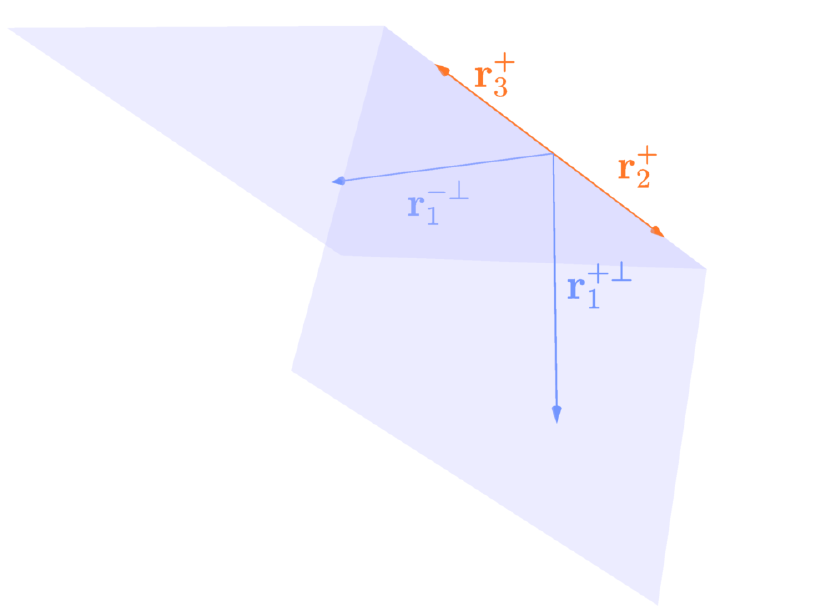}
  \end{subfigure}%
    \qquad
  \begin{subfigure}[t]{0.3\textwidth}
        \includegraphics[height=4cm,trim={0.0cm 0 50 100},clip]{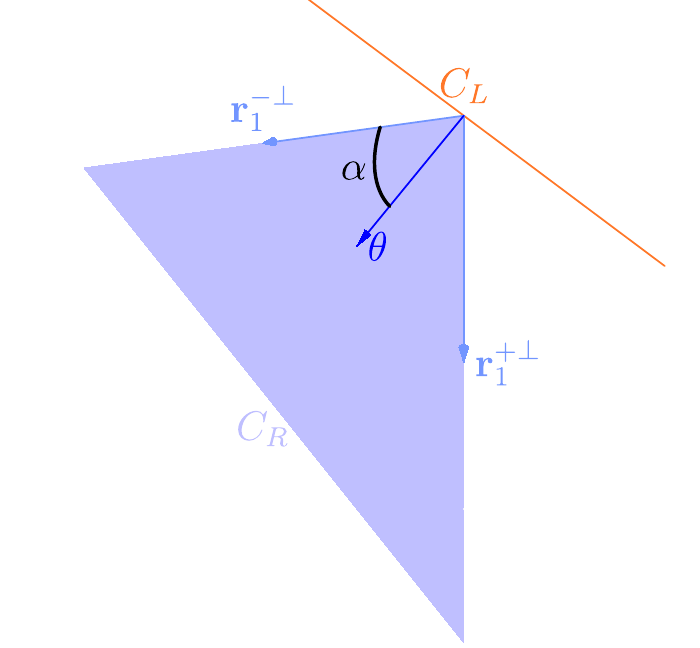}
    \end{subfigure}%
        \caption{\textbf{Decomposition into the direct sum of the lineality space and range}. Left: Decomposition of $\R^3$ into the lineality space $C_L$ and its orthogonal complement $C_L^\perp$, where $C$ is the descent cone of Figure~\ref{fig:descentconel1}. Middle: Different view of $C$, where $\r_i^{\pm \perp} \coloneqq \proj_{\smash{C_L^\perp}} (\r_i^\pm)$. Right: Visualization of the orthogonal decomposition of $C$ into its lineality space and range $C_R=\proj_{\smash{C_L^\perp}}(C)$. The angle $\alpha$ corresponds to the circumangle of $C_R$ and $\ax$ denotes its circumcenter. \label{fig:decompositiondescentcone}}
    \end{center}
\end{figure}

We now turn to the decomposition of the descent cone of the gauge $p_{\Dict\cdot\Bn[1]{d}}$ into lineality space and range. 
To that end, let us first make the following simple observation.
\begin{lemma}[Sign pattern of $\ell^1$-representers]
All minimal $\ell^1$-representers of $\gt$ with respect to $\Dict$ share the same sign pattern, in the sense that for all $\zl^1$, $\zl^2\in \Zlset$, the coordinate-wise product $\zl^1\cdot \zl^2$ is nonnegative.
\end{lemma}
This lemma allows us to define the maximal $\ell^1$-support.
\begin{definition}[Maximal $\ell^1$-support]
Let $\gt \in \R^n$ and $\Dict\in\R^{n\times d}$ be a dictionary. 
The maximal $\ell^1$-support $\bar \Supp$ of $\gt$ in $\Dict$ (or simply maximal support) is defined as $\bar \Supp \coloneqq \cup_{\zl \in \Zlset} \mathrm{supp}(\zl)$. 
In what follows, we let $\bar s=\#\bar\Supp$ denote its cardinality. 
\end{definition}
Since all solutions $\zl \in \Zlset$ have the same sign pattern, any point $\zl$ in the relative interior of $\Zlset$ has maximal support. The next decomposition forms the main result of this section; see Appendix~\ref{sec:proof_geom} for a proof.
\begin{proposition}
\label{prop:lindc}
Let $\Dict\in\R^{n\times d}$ be a dictionary and let $\gt \in \ran (\Dict)\setminus \{\vec{0}\}$.
Let $C=\dc{p_{\Dict\cdot\Bn[1]{d}},\gt}$ denote the descent cone of the gauge at $\gt$. 
Let $\zl \in \mathrm{ri}(\Zlset)$ be any minimal $\ell^{\smash{1}}$-representer of $\gt$ in $\Dict$ with maximal support and set $\bar \Supp = \supp (\zl)$ as well as $\bar s=\#\bar\Supp$. Assume $\bar{s} < d$. Then we have: 
\begin{enumerate}
 \item[(a)] The lineality space of $C$ has a dimension less than $\bar s-1$ and is given by
\begin{equation}
 C_L=\spann(\bar s \cdot \sign(z_{\ell^1,i}) \cdot \dict_i - \Dict \cdot \sign(\zl) : i \in \bar \Supp).
\end{equation}
\item[(b)] The range of $C$ is a $2(d-\bar{s})$-polyhedral $\alpha$-cone given by: 
\begin{equation}
  C_R=\mathrm{cone}(\r_j^{\pm \perp} : j\in \bar \Supp^c) \mbox{ with } \r_j^{\pm\perp}\coloneqq \proj_{C_L^\perp}\left(\pm\bar s \cdot \dict_j - \Dict \cdot \sign(\zl)\right).
\end{equation}
\end{enumerate}
\end{proposition}

\subsubsection{Consequence for the Sampling Rates}
\label{sec:sr}

We now combine the main results of the previous two sections to derive an upper bound on the conic mean width of $\dc{p_{\Dict\cdot\Bn[1]{d}},\gt}$.
\begin{theorem} \label{thm:decomp}
If $\bar{s} \leq d -3$, we obtain that
\begin{equation*}
\cmw[2]{\dc{p_{\Dict\cdot\Bn[1]{d}},\gt}} \leq \bar s + \left( \tan \alpha \cdot \left(\sqrt{2\log \left( \frac{2(d-\bar s)}{\sqrt{2\pi}}\right)} + 1 \right) + \frac{1}{\sqrt{2\pi}} \right)^2,
\end{equation*}
where we have used the same notation and assumptions as in Proposition~\ref{prop:lindc}.
\end{theorem}

A proof of the previous result is given in Appendix~\ref{sec:proof_decomp}. As a direct consequence, we get the following upper bound on the sampling rates for coefficient and signal recovery.
\begin{corollary}\label{cor:criticalnumberalpha}
The critical number of measurements $m_0$ in \eqref{eq:m_coef} and \eqref{eq:m_sig1} satisfies
\begin{equation}
 m_0 \lesssim \bar s + \tan^2 \alpha \cdot \log (2(d-\bar s)/\sqrt{2\pi}).
\end{equation}
\end{corollary}

This result shows that robust coefficient and signal recovery is possible, when the number of measurements obeys $m \gtrsim \bar s + \cdot \tan^2 \alpha \cdot\log (d)$. Hence, the sampling rate is mainly governed by the sparsity of maximal support $\ell^1$-representations of $\gt$ in $\Dict$ and the ``narrowness'' of the remaining cone $C_R$, which is captured by its circumangle $\alpha \in [0,\pi/2)$. The number of dictionary atoms only has a logarithmic influence. The next section is devoted to applying the previous result to various examples.  

\begin{remark}
\label{rmk:sig}
For the sake of clarity, the previous results are given in terms of the maximal sparsity. However, (potentially) more precise bounds can be achieved when replacing $\bar s$ by $\dim (C_L)$. 
  Furthermore note, that the proof of Theorem~\ref{thm:decomp} reveals that $\dim (C_L)$ is a necessary component in the required number of measurements. Indeed, since $\cmw[2]{\dc{p_{\Dict\cdot\Bn[1]{d}},\gt}}$ is a sharp description for the required number of measurements, equation~\eqref{eq:intermed} shows that the number of measurements for successful recovery is lower bounded by $\dim (C_L)$. 
\end{remark}

\subsubsection{Examples}
\label{sec:examples}

In this section, we discuss four applications of the previous upper bound on the conic mean width. First, we show that prediction for the required number of measurements agrees with the standard theory of compressed sensing. We then analytically compute the sampling rate of Corollary~\ref{cor:criticalnumberalpha} for a specific scenario, in which the dictionary is formed by a concatenation of convolutions. The third example focuses on a numerical simulation in the case of 1D total variation.  Lastly, we demonstrate how the circumangle can be controlled by the classical notion of coherence.

\paragraph{The Standard Basis}\label{subsec:standardbasis}
Our first example is dedicated to showing that the result of Corollary~\ref{cor:criticalnumberalpha} is consistent with the standard theory of compressed sensing when $\Dict = \Id$. Hence, assume that we are given a sparse vector $\gt \in \R^n$ with $\Supp = \supp (\gt)$ and $s = \# \Supp \leq n-3$. Trivially, $\gt$ is then its own, unique $\ell^1$-representation with respect to $\Id$. According to Lemma~\ref{lem:lineality}, the $(s-1)$-dimensional lineality space of $C = \dc{\norm{\cdot}_1,\gt}$ is given by
\begin{equation}
 C_L = \spann \left( \r_i^+ : i \in \Supp \right),
\end{equation}
where $\r_i^+ = s \cdot \sign (x_{0,i})  \cdot \vec{e}_i -\sign( \gt)$. For $i \in \Suppc$ a simple calculation shows that $\ax, \r^\pm_i \in C_L^\perp,$ where $\ax = - \sign (\gt) / \sqrt{s} \in \Sn{n-1}$ and $\r^\pm_i = \pm s \cdot \vec{e}_i - \sign (\gt)$. Furthermore, for $i \in \Suppc$ it holds true that
\begin{equation}
 \sp{\ax}{\r^\pm_i} = \sqrt{s} = (1/\sqrt{s+1}) \cdot \norm{\r^\pm_i},
\end{equation}
so that the vectors $\r^\pm_i$ generate a $2(n-s)$-polyhedral $\alpha$-cone $C_R$ with $\tan^2 \alpha = s$.
Hence, Corollary~\ref{cor:criticalnumberalpha} states that robust recovery of $\gt$ is possible for $m \gtrsim 2s \log (2(n-s)/\sqrt{2\pi})$ measurements. This bound is to be compared with the classical compressed sensing result, which prescribes to take $m \gtrsim 2s \log (n/s)$ measurements. Note that the slight difference in the logarithmic factor is due to our simple bound on $ W(\alpha, k, n)$, cf.~Remark~\ref{rem:poly}. 

\paragraph{A Convolutional Dictionary}

Consider a dictionary $\Dict$ defined by the concatenation of two convolution matrices $\mathbf{H}_1$ and $\mathbf{H}_2$ with convolution kernels $\h_1=[1,1]$ and $\h_2=[1,-1]$, respectively. For instance, in dimension $n=4$, this would yield the following matrix:
\begin{equation}
 \Dict=\begin{pmatrix}
    1 & 1 & 0 & 0  & 1  & -1 & 0  & 0 \\
    0 & 1 & 1 & 0  & 0  & 1  & -1 & 0  \\
    0 & 0 & 1 & 1  & 0  & 0  & 1  & -1  \\
    1 & 0 & 0 & 1  & -1 & 0  & 0  & 1  
   \end{pmatrix}.
\end{equation}
In particular for imaging applications, popular signal models are based on sparsity in such concatenations of convolutional matrices, e.g., translation invariant wavelets~\cite{mallat09} or learned filters in the convolutional sparse coding model~\cite{wohlberg2014,bristow2013}.  Note that the resulting dictionary is highly redundant and correlated, so that existing coherence- and RIP-based arguments cannot provide satisfactory recovery guarantees. For the same reason, a recovery of a unique minimal $\ell^1$-representer by~\eqref{eq:coef} is unlikely, cf.~the numerical simulation in Section~\ref{sec:sampling_sig}. However, in the following, we will show how the previous upper bound based on the circumangle can be used in order to analyze signal recovery by~\eqref{eq:sig}.

To that end, we consider the recovery of a simple vector $\gt\in\R^n$ supported on the first and the last component only, i.e., $x_{0,i}=0$ for all $2\leq i\leq n-1$. A generalization to vectors supported on supports made of pairs of contiguous indices separated by pairs of contiguous zeros is doable, but we prefer this simple setting for didactic reasons.
In this case, the set of minimal $\ell^1$-representers can be completely characterized. Assuming additionally that $x_{0,1}>x_{0,n}>0$, one can show that
\begin{align*}
  Z_{\ell^1} = \big\{ & \zl = [\z^{(1)}; \z^{(2)}]\in \R^{2n}, \textrm{ with } \supp(\z^{(1)})=\supp(\z^{(2)})=\{1,2\}, \\
 & z_1^{(1)}=\frac{x_{0,1}+x_{0,n}}{2} -\delta, z_1^{(1)}=\delta,  z_2^{(1)}=\frac{x_{0,1}-x_{0,n}}{2} -\delta, z_2^{(1)}=-\delta, \\
 & 0\leq \delta \leq \frac{x_{0,1}-x_{0,n}}{2}\big\}.
\end{align*}
Let $\zl \in Z_{\ell^1}$ denote any representer with maximal support $S=\supp (\zl) = \left\{1,2,n+1,n+2\right\}$ and set $\vvec=\Dict\cdot\sign(\zl)$. 
According to Proposition~\ref{prop:lindc}, we then decompose the descent cone $C=\dc{p_{\Dict\cdot\Bn[1]{d}},\gt}$ into $C = C_L \oplus C_R$, where $C_L$ is the lineality space given by
\begin{equation*}
 C_L= \spann(4\cdot \sign(z_{\ell^1,i}) \cdot \dict_i - \vvec : i\in S),
\end{equation*}
and the range is given by $C_R=\cone{\proj_{\smash{C_L^\perp}}( \pm 4 \cdot \dict_j - \vvec : j \in S^c)}$.
It is easy to see that $\dim (C_L) = 2$, and that the projection on $C_L^\perp$ can be expressed as 
\begin{equation*}
 (\proj_{C_L^\perp}(\x))_i=\begin{cases}
				   0, & \textrm{ if } i \in \{2,n\}, \\
                       x_i,  & \textrm{ otherwise}.
                      \end{cases}
\end{equation*}
The goal is now to show that $C_R$ is contained in a circular cone with angle $\alpha = \arccos (1/\sqrt{3})$ and axis $\ax = - \proj_{\smash{C_L^\perp}}(\vvec) / \norm{\proj_{\smash{C_L^\perp}}(\vvec)}_2 =  -\vec{e}_1$. Indeed, a straightforward computation shows that for $j \in S^c$ we have
 \[
\left( \proj_{C_L^\perp} (\pm 4 \cdot \dict_j - \vvec) / \norm{\proj_{C_L^\perp} (\pm 4 \cdot \dict_j-\vvec)}_2  \right)_1  \in \left\{ -1/\sqrt{2}, -1/\sqrt{3}\right\}.
 \]
Hence, Corollary~\ref{cor:criticalnumberalpha} implies that robust recovery of $\gt$ is possible for $m\gtrsim 2 + 2\log(4n)$ measurements.

\paragraph{1D Total Variation}

As a third example we consider the problem of total variation minimization in 1D. Assume that $\gt, \Meas, \y, \noise$ and $\noiseparam$ follow Model~\ref{mod:noisy_meas} with $\noiseparam = 0$ and that $\Meas$ obeys Model \ref{mod:meas_mod}. Total variation minimization is based on the assumption that $\gt$ is gradient-sparse, i.e., that $\# \supp (\TV \gt) \leq s \ll n$, where  $\TV \in \R^{n-1\times n}$ denotes a discrete gradient operator, which is for instance based on forward differences with von Neumann boundary conditions. In order to recover $\gt$ from its noiseless, compressed measurements $\y$, one solves the program
\begin{equation}
 \label{eq:tv1}
 \min_{\x\in\R^n} \norm{\TV\x}_1 \quad \mbox{ s.t. } \quad \y = \Meas \x.
\end{equation}
For signals with $\mathbf{1}^T \cdot \gt = 0$ it is easy to see that the previous formulation is equivalent to solving the synthesis basis pursuit~\sigNoisefree{} with $\Dict = \TV^\dagger$, where $\TV^\dagger \in \R^{n\times n-1}$ denotes the Moore-Penrose inverse of $\TV$. 

The research of the past three decades demonstrates that encouraging a small total variation norm often efficaciously reflects the inherent structure of real-world signals. Although not as popular as its counterpart in 2D, total variation methods in one spatial dimension find application in many practical applications, see for instance~\cite{Little2011}.
Somewhat surprisingly, Cai and Xu have shown  that a \emph{uniform} recovery of all $s$-gradient-sparse signals is possible if and only if the number of (Gaussian) measurements obeys $m \gtrsim \sqrt{sn} \cdot \log(n)$; see~\cite{cai_guarantees_2015}. Recently, \cite{Genz2020} has proven that this square-root bottleneck can be broken for signals with well separated jump discontinuities. This result is also based on establishing a non-trivial upper bound on the conic mean width. For such ``natural'' signals, $m\gtrsim s \cdot \log^{\smash{2}} (n)$ measurements are already sufficient for exact recovery. See also ~\cite{guntuboyina2020,dalalyan2017} for closely related results in a denoising context. 

We want to demonstrate that the upper bound on the conic mean width based on the circumangle is capable of breaking the square-root bottleneck of the synthesis-based reformulation above. A theoretical analysis appears to be doable, however, it is beyond the scope of this work. Instead, we restrict ourselves to a simple numerical simulation.  We consider signals that are defined by the pointwise discretization of a function on an interval with a few equidistant discontinuities and zero average. Note that for such a signal $\gt$ the unique minimal $\ell^1$-representer with respect to $\TV^\dagger$ is simply given by $\TV\gt$. Hence, we are only left with numerically computing the circumangle $\alpha$ of the range $C_R$ in Proposition~\ref{prop:lindc}, which is done by means of Proposition~\ref{prop:simple_circumcenter}.  In order to confirm that required number of measurements scales logarithmically in the ambient dimension $n$, we analyze the behavior of $\tan^2\alpha$ when the resolution is increased, i.e., for $n=500,1000, \dots 10000$. The result is shown in Figure~\ref{fig:alpha}. The logarithmic scaling of $\tan^2 \alpha$ (note that the $n$-axis is logarithmic) indeed suggests that the bound of Corollary~\ref{cor:criticalnumberalpha} predicts  that $m \gtrsim s\cdot \log^{\smash{2}} (n)$ measurements suffice for exact recovery. Hence, the presented upper bound based on the circumangle appears to be sharp enough to break the square-root bottleneck of total variation minimization in 1D. 

\begin{figure}
    \begin{center}
        \includegraphics[height=5cm]{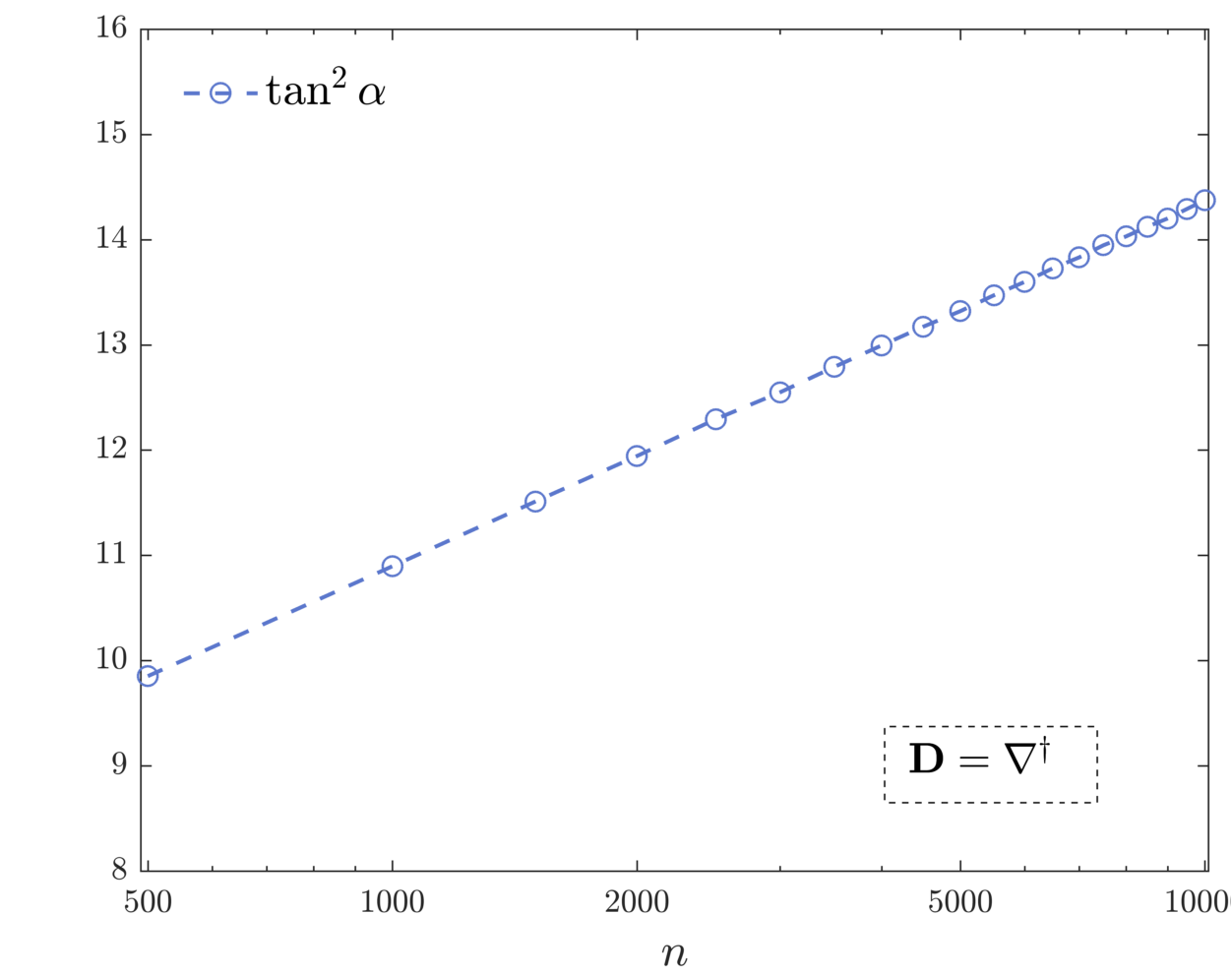}
        \caption{\textbf{Logarithmic scaling of $\tan^2 \alpha$.} The figure displays the behavior of $\tan^2 \alpha$ for increasing ambient dimension $n$, where $\alpha$ denotes the circumangle of the range $C_R$ in Proposition~\ref{prop:circ}. Here, the dictionary is chosen as $\Dict = \TV^\dagger$ and the considered signals $\gt$ have equidistant jump discontinuities and zero average. The plot indicates that the upper bound based on the circumangle is sharp enough to break the square-root bottleneck of~\cite{cai_guarantees_2015}.} 
        \label{fig:alpha}
    \end{center}
\end{figure}

\paragraph{Coherence and Circumangle}
In our last example, we show that the circumangle of $C_R$ of Proposition~\ref{prop:lindc} can be controlled in terms of the \emph{mutual coherence} of the dictionary (see Equation~\eqref{eq:coh}). This notion is a classical concept in the literature on sparse representations, which is frequently used to derive uniform recovery statements; see for instance~\cite[Chapter 5]{foucart2013cs} for an overview. Note that the assumption $s < \tfrac{1}{2}(1+ \mu^{-1})$ of the following result guarantees that every $s$-sparse $\zl$ is the unique minimal $\ell^{\smash{1}}$-representer of its associated signal $\Dict \zl$~\cite{Donoho2003,1255564}. Hence, in this situation, coefficient and signal recovery are equivalent, and both formulations are governed by the conic mean width of the cone $C =  \dc{p_{\Dict\cdot\Bn[1]{d}},\Dict\zl} =   \Dict \cdot \dc{\norm{\cdot}_1,\zl}$.

\begin{proposition}\label{prop:coherence_bound}
Let $\Dict\in\R^{n\times d}$ be a dictionary that spans $\R^n$ with $\|\dict_i\|_2 \leq 1$ for $i\in [d]$ and mutual coherence  $ \mu = \mu(\Dict)$. Let $\zl\in \R^d \setminus \left\{ \vec{0}\right\}$ denote an arbitrary $s$-sparse vector with $s < \frac{1}{2}(1 + \mu^{-1})$.
Then the circumangle $\alpha$ of the range $C_R$ of the descent cone $C = \dc{p_{\Dict\cdot\Bn[1]{d}},\Dict\zl}$ obeys:
\begin{align*}
    \tan^2 \alpha & \leq \frac{s(1 - s\mu)}{(1 - 2 s \mu)^2}.
\end{align*}
\end{proposition}

A proof of the previous result can be found in Appendix~\ref{sec:coherence_bound}. The statement allows us to retrieve a bound of the order $m \gtrsim s \log (d)$ for the needed number of measurements. For example, with $s\mu = 1/10$, the bound of Corollary~\ref{cor:criticalnumberalpha} results in a sampling rate of  $ m_0 \leq 4s \log (d)$. Observe that this is essentially the same result as~\cite[Corollary II.4]{rauhut2008}, however, the constants of Proposition~\ref{prop:coherence_bound} are better controlled. 

Note that the mutual coherence of a dictionary (sometimes also referred to as \emph{worst-case coherence}~\cite[Chapter 9]{casazza2012finite}) is a global quantity that is usually used to derive recovery guarantees that are uniform across all $s$-sparse signals. Approaches based on this notion suffer from the so-called square-root bottleneck: The Welch bound~\cite[Theorem 5.7]{foucart2013cs} reveals that the condition $s < \tfrac{1}{2}(1+ \mu^{-1})$ can only be satisfied for mild sparsity values $s \lesssim\sqrt{n}$. We emphasize that this is in contrast to the strategy of this work, which is tailored for a non-uniform recovery of individual signals. Indeed, the circumangle is a signal-dependent notion that allows to describe the local geometry.

\section{Numerical Experiments}
\label{sec:num_exp}

In this section, we illustrate our main findings regarding the $\ell^1$-synthesis formulation by performing numerical simulations. First, we study the recovery of coefficient representations by \coefNoisefree{}; see Section~\ref{sec:num_coef}. In Section~\ref{sec:sampling_sig}, we then focus on signal recovery by \sigNoisefree{} in situations, where the identification of a coefficient representation is impossible. Section~\ref{sec:2d_pt} is dedicated to the experiment of Figure~\ref{fig:sig:pt_2}, in which both formulations are compared. Lastly, we investigate the differences concerning robustness to measurement noise in Section~\ref{sec:robustness_noise}. For general design principles and more details on our simulations we refer the interested reader to Appendix~\ref{sec:num_details}.  

To the best of our knowledge, all other compressed sensing results on the $\ell^{\smash{1}}$-synthesis formulation with redundant dictionaries describe the sampling rate as an asymptotic order bound.  Hence, these results are not compatible with the experiments in this section and will not be further considered. 

\subsection{Sampling Rates for Coefficient Recovery}
\label{sec:num_coef}

In order to study coefficient recovery by~\coefNoisefree, we create \emph{phase transition plots} by running Experiment~\ref{exp:coef} for different dictionary and signal combinations reported below. 

\begin{experiment}[Phase transition for a fixed coefficient vector]\leavevmode
\label{exp:coef}
\vspace{-.25\baselineskip}\hrule\vspace{.5\baselineskip}

\myhangindent{Input: \ }
\expkwd{Input:} Dictionary $\Dict \in \R^{n \times d}$, coefficient vector $\zl \in \R^d$.

\vspace{.5\baselineskip}
\expkwd{Compute:} Repeat the following procedure $100$ times for every $m =1,2,\dots,n$:
\begin{expstep}
	 \item 
		Draw a standard i.i.d.\ Gaussian random matrix $\Meas \in \R^{m\times n}$ and determine the measurement vector $\y = \Meas\Dict\zl$.
	\item 
		Solve the program \coefNoisefree{} to obtain an estimator $\solz \in \R^d$. 
	\item
		Compute and store the recovery error $\norm{\zl - \solz}_2$. Declare success if $\norm{\zl - \solz}_2 < 10^{-5}$.
\end{expstep}
\end{experiment}

\paragraph{Simulation Settings} Our first two examples are based on a redundant Haar wavelet frame, which can be seen as a typical representation system in the field of applied harmonic analysis, see~\cite{mallat09} for more details on wavelets and Section 3.1 in~\cite{genzel2017} for a short discussion in the context of compressed sensing. As a back-end for defining the wavelet transform, we are using the \texttt{Matlab} software package \texttt{spot}~\cite{spot}, which is in turn based on the \texttt{Rice Wavelet Toolbox}~\cite{rwt}. We set the ambient dimension to $n=256$ and consider a Haar system with $3$ decomposition levels and normalized atoms. The resulting dictionary is denoted by $\Dict = \Dictw \in \R^{256 \times 1024}$. 
The first coefficient vector $\zl^{\smash{1}} \in \R^{1024}$ is obtained by selecting a random support set of cardinality $s=16$, together with random coefficients; see Subfigure~\ref{fig:coef:wave:3} for a visualization of $\zl^{1}$ and Subfigure~\ref{fig:coef:wave:2} for the resulting signal $\vec{x}_1 = \Dictw \cdot \zl^{1}$. The second coefficient vector $\zl^{2}$ is created by defining two contiguous blocks of non-zero coefficients in the low frequency part, again with $s=16$; see Subfigure~\ref{fig:coef:wave:6} for a plot of $\zl^{2}$ and Subfigure~\ref{fig:coef:wave:5} for the resulting signal $\x_2 = \Dictw \cdot \zl^{\smash{2}}$. For each signal we run Experiment~\ref{exp:coef} and report the empirical success rate in the Subfigures \ref{fig:coef:wave:1}, \ref{fig:coef:wave:4}, respectively.

In the third example, the dictionary is chosen as a Gaussian random matrix, which is a typical benchmark system for compressed sensing with redundant frames, see~for instance \cite{genzel2017,kabanava2015,chen2014}. Also in this case, we set $n= 256$, but we choose $d=512$. The resulting dictionary is denoted by $\Dict_{\texttt{rand}} \in \R^{256\times 512}$. The coefficient vector $\zl^{3}$ is  defined in the same manner as $\zl^1$ above (see Subfigure~\ref{fig:coef:wave:9}), where we again have $\|\zl^{3}\|_0=16$. The resulting signal $\x_3$ is shown in Subfigure~\ref{fig:coef:wave:8} and the empirical success rate in Subfigure~\ref{fig:coef:wave:7}.       

Our fourth and last dictionary is inspired by super-resolution; see for instance~\cite{candes2012}. We again set $n=256$ and choose the dictionary $\Dict_{\texttt{super}} \in \R^{256 \times 256}$ as a convolution with a discrete Gaussian function of large variance. This example can therefore be considered as a finely discretized super-resolution problem. The coefficient vector $\zl^{4}$ is then chosen as a sparse vector with $z_{\ell^{\smash{1}},128}^{\smash{4}} = 1$ and $z_{\ell^{\smash{1}},129}^{\smash{4}} = -1$, see Subfigure~\ref{fig:coef:wave:12}. Hence, in the signal $\x_4$, the two neighboring peaks almost cancel out and result in the low amplitude signal shown in Subfigure~\ref{fig:coef:wave:11}. Finally, the empirical success rate is depicted in Subfigure~\ref{fig:coef:wave:10}.  
Note that for each example we have verified the condition $\lmin{\Dict}{\dc{\Onenorm,\zl^{i}}} > 0$ heuristically by verifying that $\Zlset = \left\{ \zl^{\smash{i}} \right\}$, respectively.

\paragraph{Results} Let us now analyze the empirical success rates of Figure~\ref{fig:coef:wave} and compare them with the estimates of $\cmw[2]{\Dict \cdot \ds{\Onenorm;\zl^i}}$ and $\cmw[2]{\ds{\Onenorm;\zl^i}}$. Our findings are summarized in the following:
\begin{enumerate}[(i)]
 \item The convex program \coefNoisefree{} obeys a sharp phase transition in the number of measurements $m$: Recovery of a coefficient vector fails if $m$ is below a certain threshold and succeeds with overwhelming probability once a small transition region is surpassed. This observation could have been anticipated, given for instance the works~\cite{amelunxen2014edge,tropp2014convex}. However, note that the product structure of the matrix $\Meas\Dict$ in \coefNoisefree{} does not allow for a direct application of these results. 
 
 \item The quantity $\cmw[2]{\Dict \cdot \ds{\Onenorm;\zl}}$ accurately describes the sampling rate of \coefNoisefree{}. Indeed, in all four simulation settings of Figure~\ref{fig:coef:wave}, the phase transition occurs near the estimated conic mean width of $\Dict \cdot \dc{\Onenorm;\zl^i}$, as predicted  by~Theorem~\ref{thm:coeff}.
 
  \item In contrast, $\cmw[2]{\ds{\Onenorm;\zl}}$ does not describe the sampling rate of \coefNoisefree{}, in general. Indeed, $\cmw[2]{\Dict \cdot \ds{\Onenorm;\zl^2}} \ll \cmw[2]{\ds{\Onenorm;\zl^2}}$ and $\cmw[2]{\Dict \cdot \ds{\Onenorm;\zl^4}} \gg \cmw[2]{\ds{\Onenorm;\zl^4}}$; see Subfigure~\ref{fig:coef:wave:4} and~\ref{fig:coef:wave:10}, respectively.
 
 \item For two minimal $\ell^1$-representations $\zl^1,\zl^2$ with the same sparsity, but with different supports, the quantities $\cmw[2]{\Dict \cdot \ds{\Onenorm;\zl^1}}$  and $\cmw[2]{\Dict \cdot \ds{\Onenorm;\zl^2}}$ might differ significantly, while $\cmw[2]{\ds{\Onenorm;\zl^1}} = \cmw[2]{\ds{\Onenorm;\zl^2}}$; see Subfigures~\ref{fig:coef:wave:1} and \ref{fig:coef:wave:4}. Hence, sparsity alone does not appear to be a good proxy for the sampling complexity of \coefNoisefree{}. A refined understanding of coefficient recovery requires a theory that is non-uniform across the class of all $s$-sparse signals. 
 
 \item The local condition number $\condi{\Dict}{\zl}$ might explode, which often renders a condition bound as in Proposition~\ref{prop:conditiongaussian} unusable. Indeed, we report upper bounds for $\smash{\lmin{\Dict}{\dc{\Onenorm,\zl^{\smash{i}}}}}$ in the first column of Figure~\ref{fig:coef:wave}. Since the norms of each dictionary are well controlled, this quantity is responsible for the large values of the local condition number. 
 
\end{enumerate}

\begin{figure}[!t]
	\centering
	\begin{subfigure}[t]{0.3\textwidth}
		\centering
		\includegraphics[width=\textwidth]{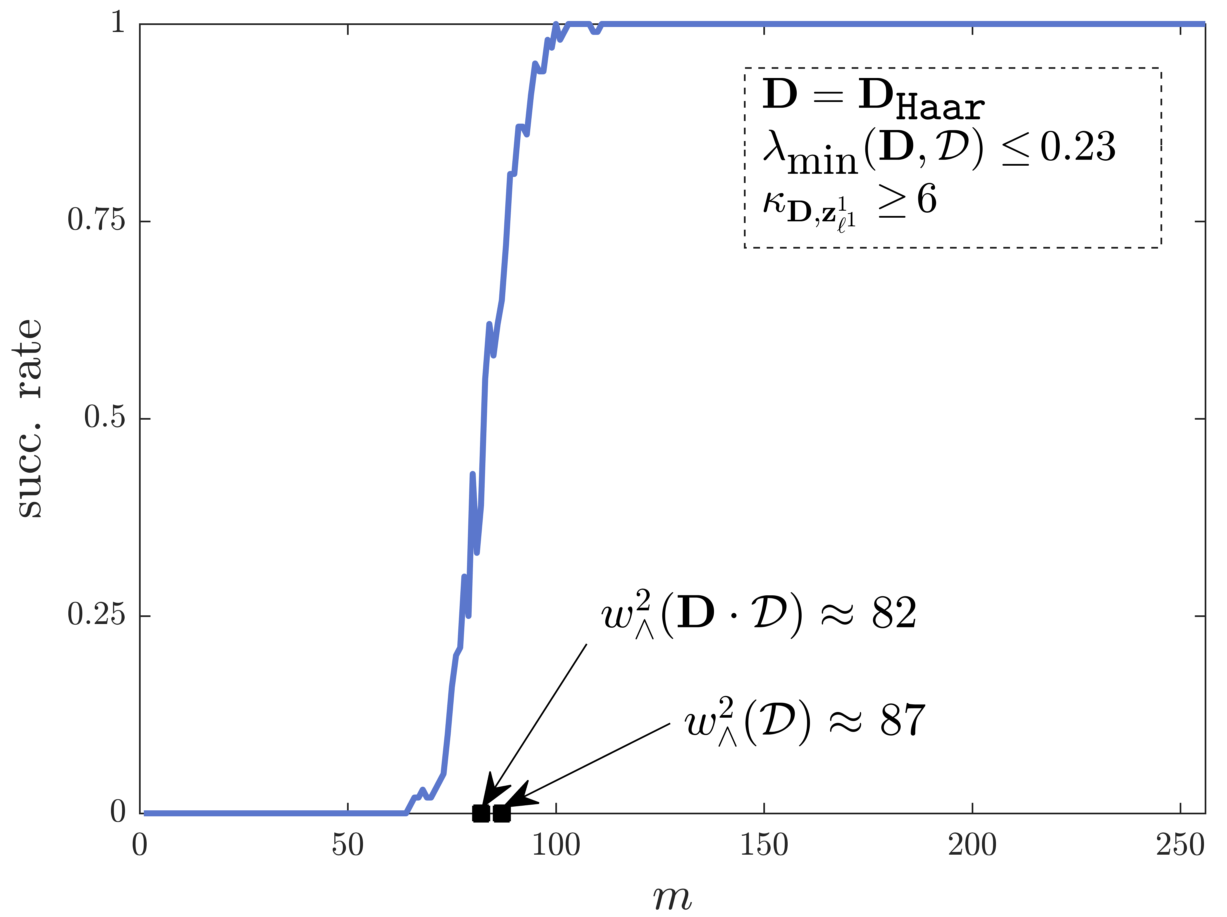}
		\caption{}
		\label{fig:coef:wave:1}
	\end{subfigure}%
	\qquad
	\begin{subfigure}[t]{0.3\textwidth}
		\centering
		\includegraphics[width=\textwidth]{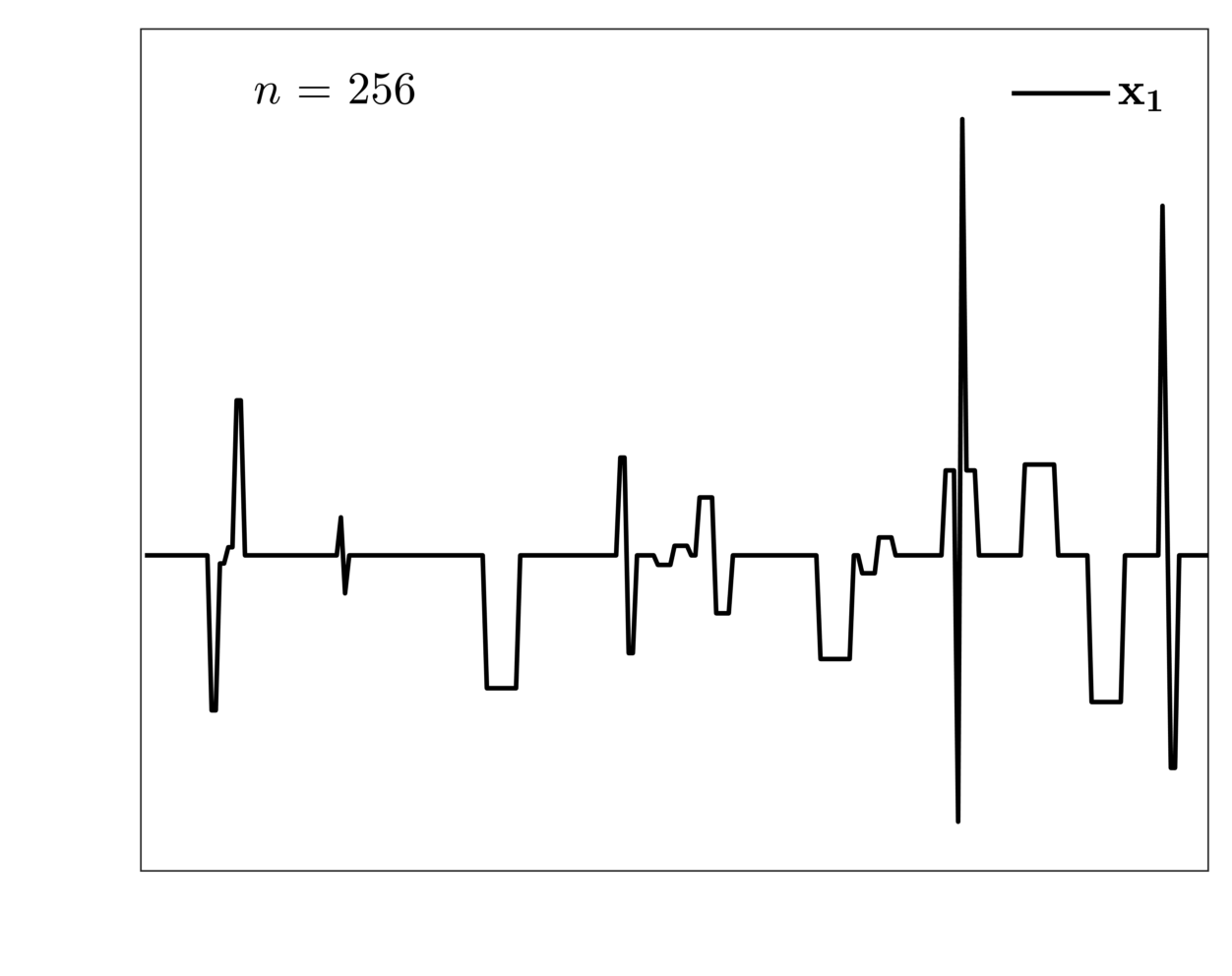}
		\caption{}
		\label{fig:coef:wave:2}
	\end{subfigure}%
	\qquad
	\begin{subfigure}[t]{0.3\textwidth}
		\centering
		\includegraphics[width=\textwidth]{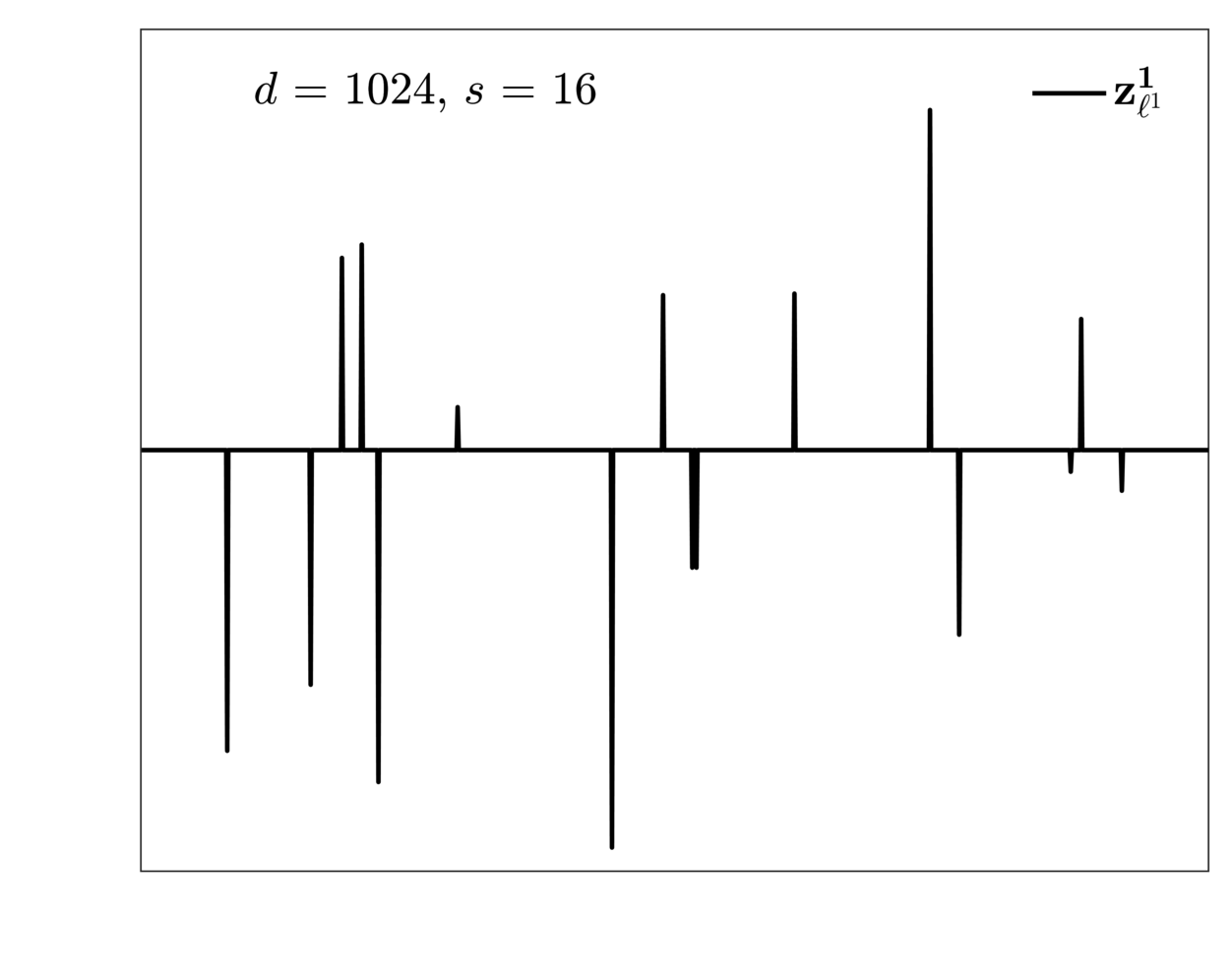}
		\caption{}
		\label{fig:coef:wave:3}
	\end{subfigure}%
	\vspace{.5\baselineskip}
		\begin{subfigure}[t]{0.3\textwidth}
		\centering
		\includegraphics[width=\textwidth]{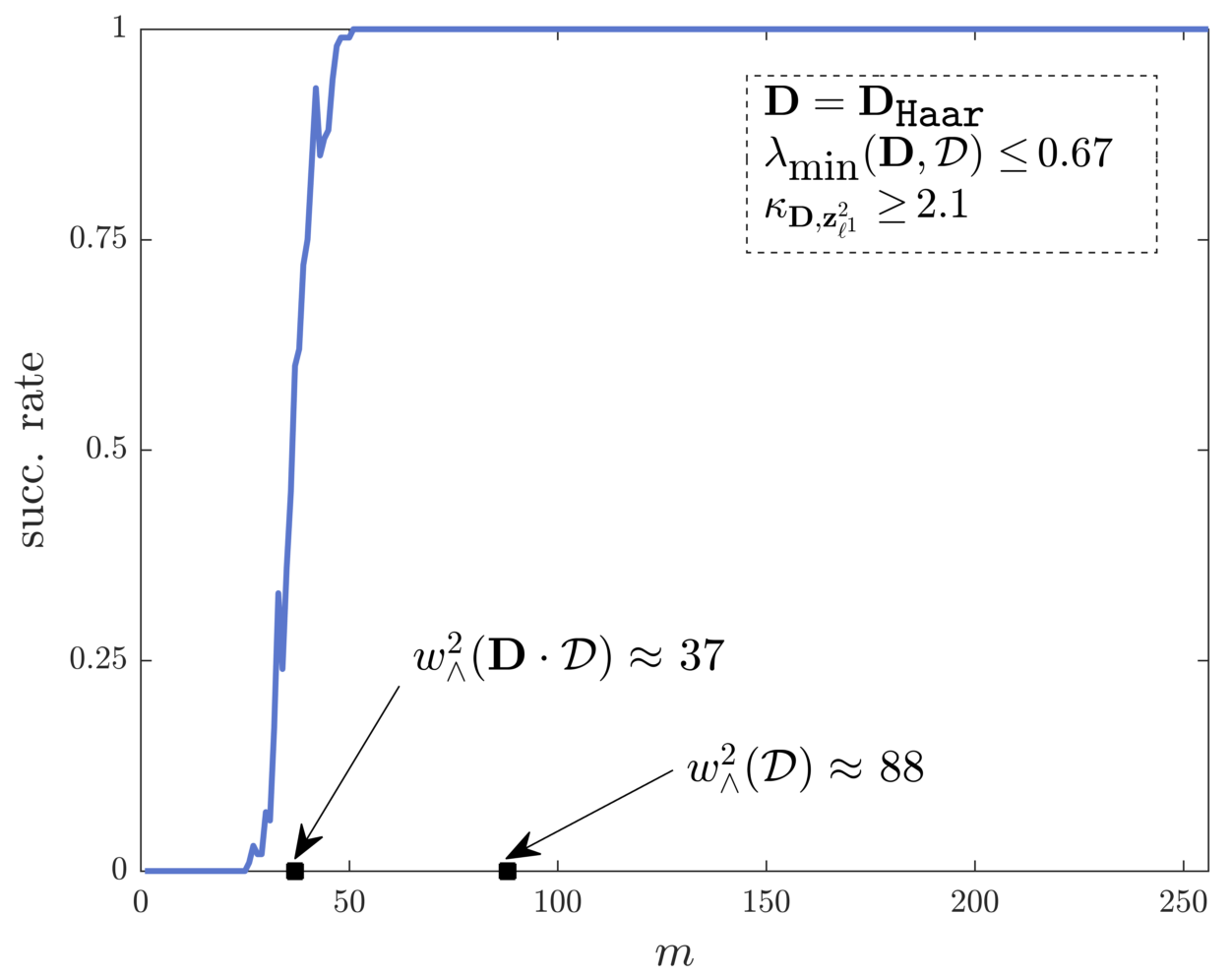}
		\caption{}
		\label{fig:coef:wave:4}
	\end{subfigure}%
	\qquad
	\begin{subfigure}[t]{0.3\textwidth}
		\centering
		\includegraphics[width=\textwidth]{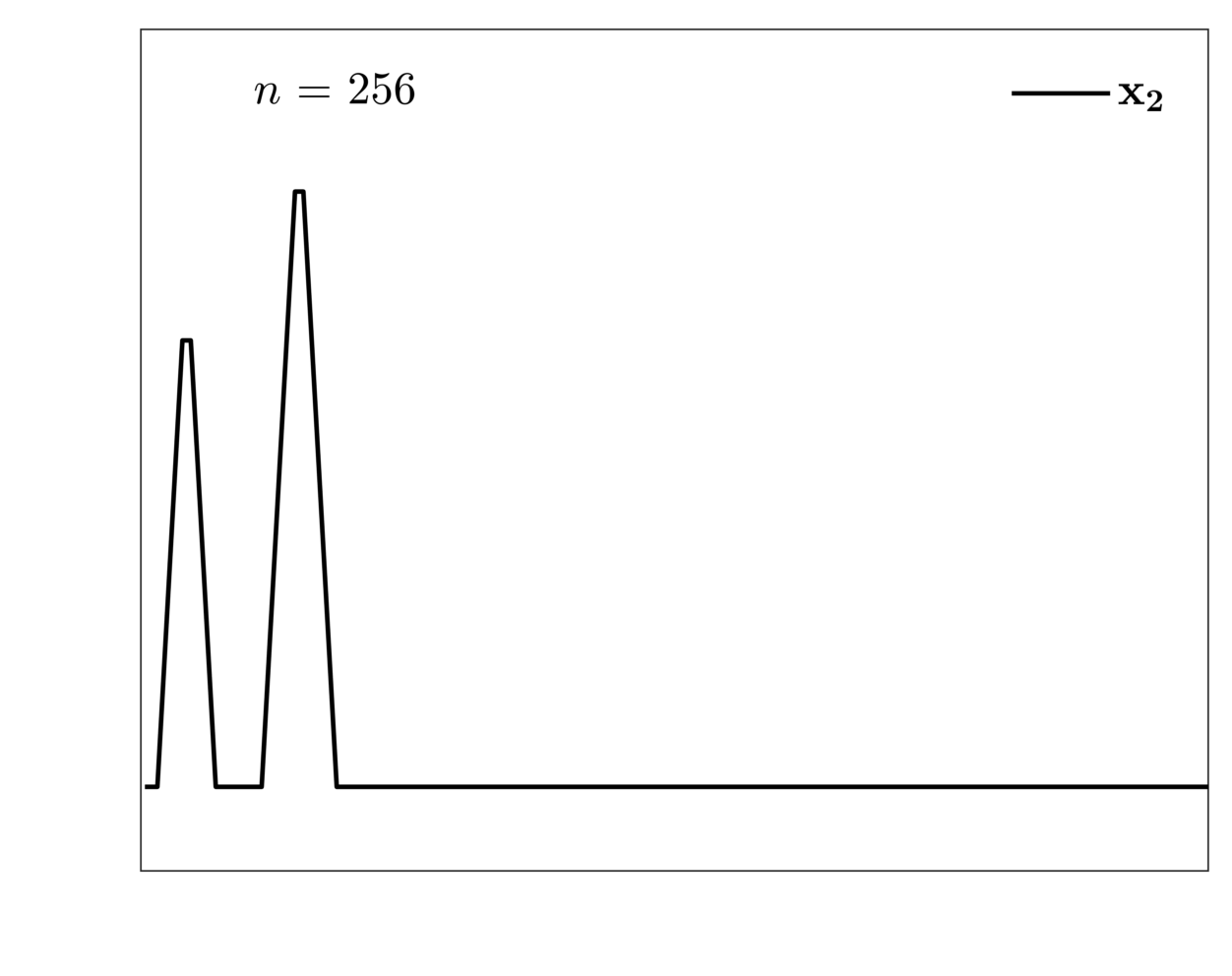}
		\caption{}
		\label{fig:coef:wave:5}
	\end{subfigure}%
	\qquad
	\begin{subfigure}[t]{0.3\textwidth}
		\centering
		\includegraphics[width=\textwidth]{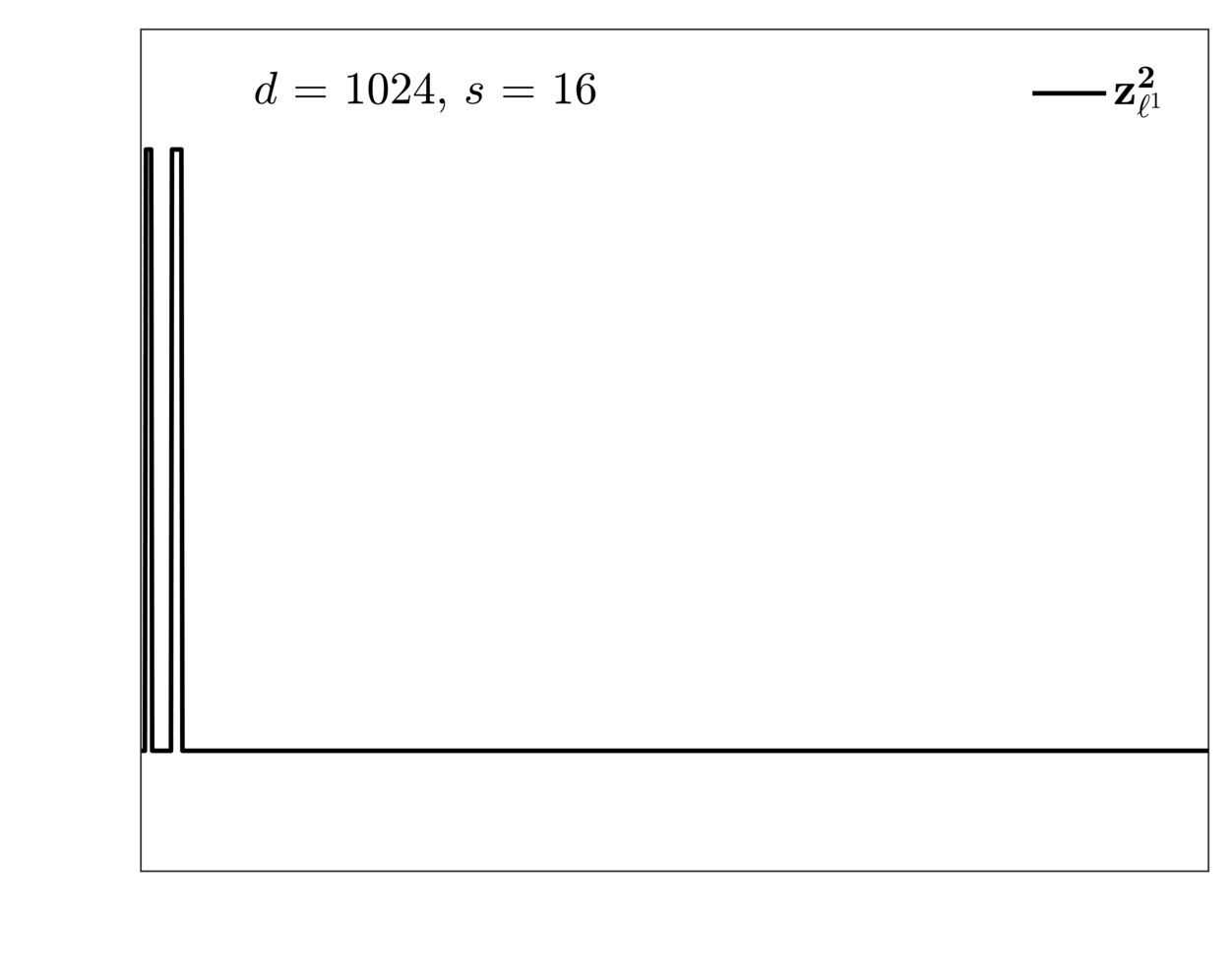}
		\caption{}
		\label{fig:coef:wave:6}
	\end{subfigure}%
	\vspace{.5\baselineskip}
	\begin{subfigure}[t]{0.3\textwidth}
		\centering
		\includegraphics[width=\textwidth]{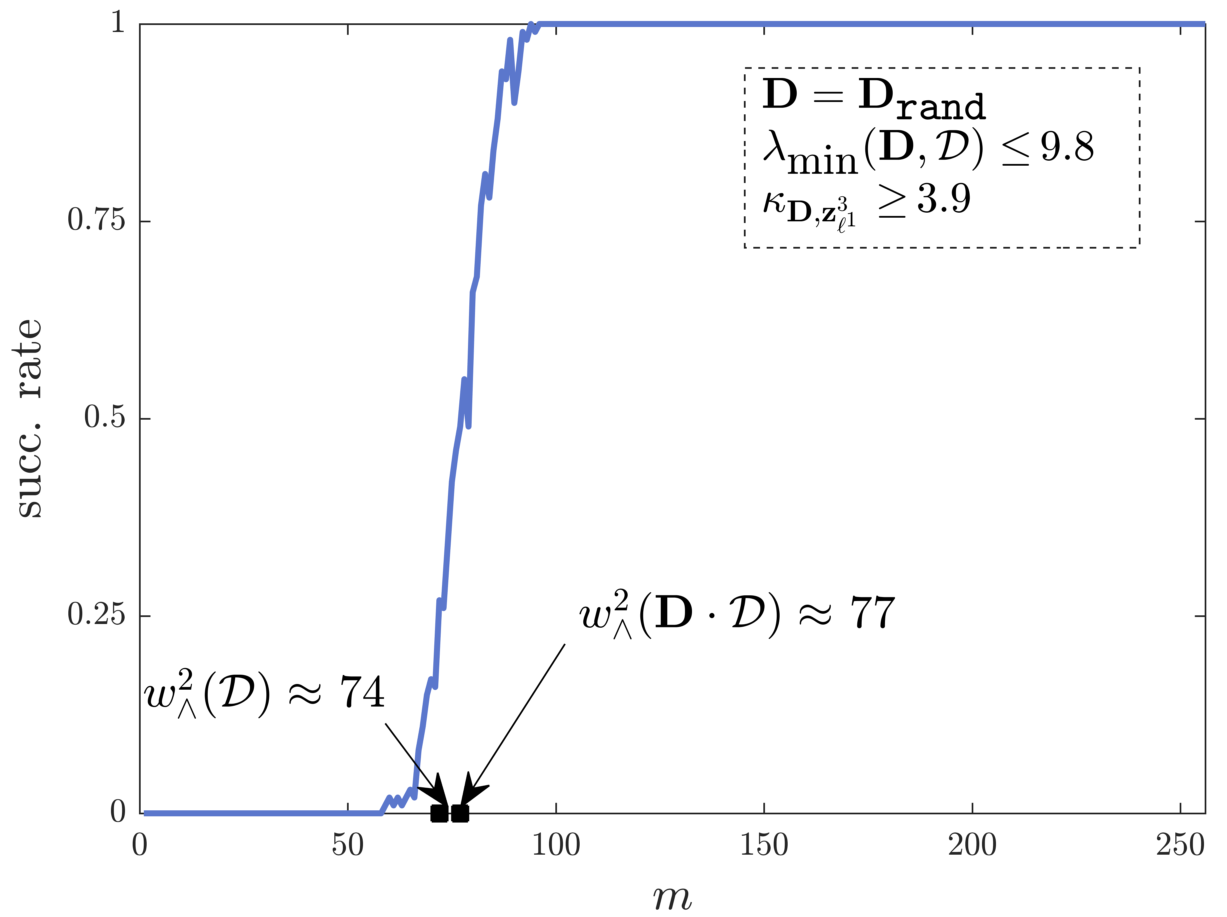}
		\caption{}
		\label{fig:coef:wave:7}
	\end{subfigure}%
	\qquad
	\begin{subfigure}[t]{0.3\textwidth}
		\centering
		\includegraphics[width=\textwidth]{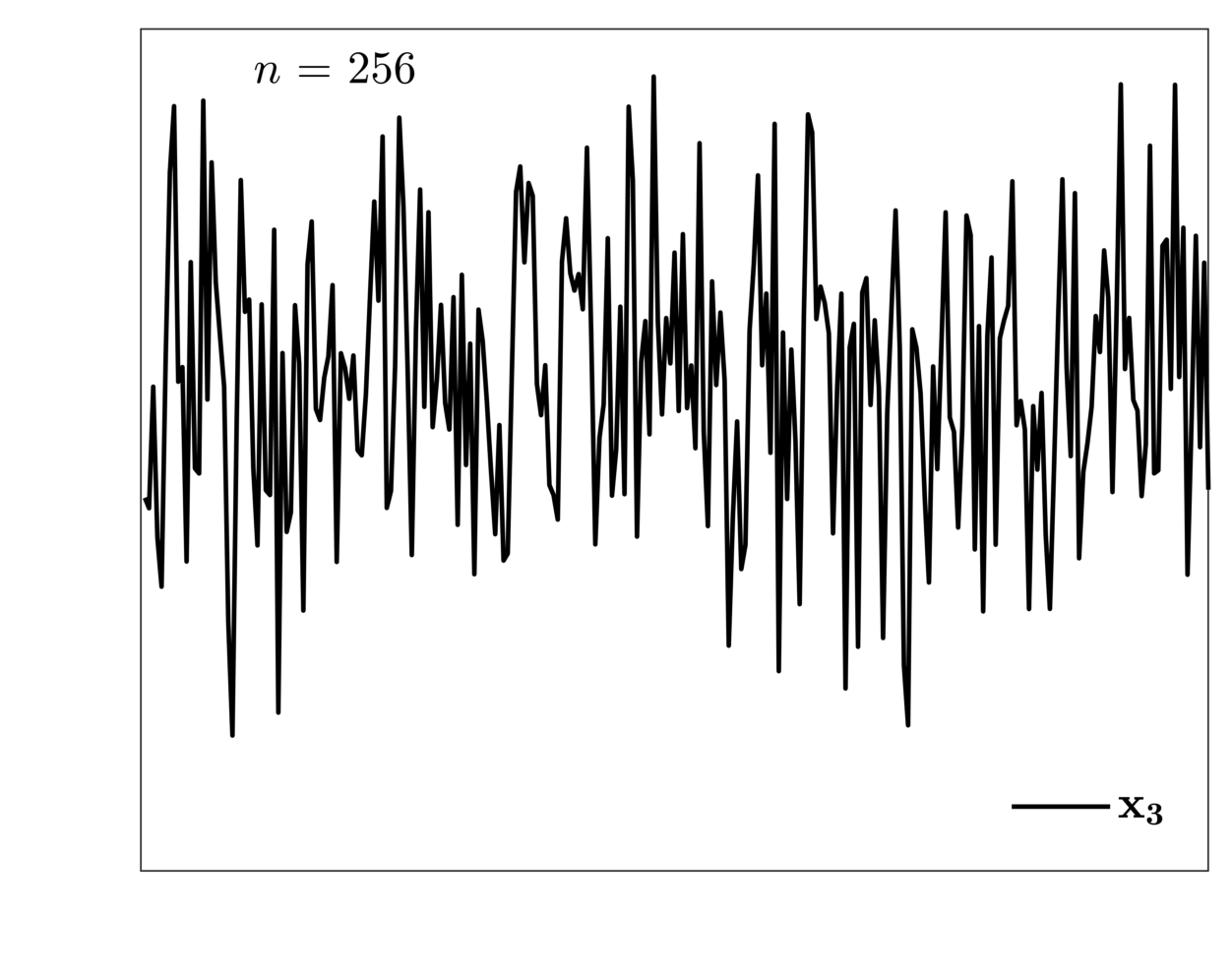}
		\caption{}
		\label{fig:coef:wave:8}
	\end{subfigure}%
	\qquad
	\begin{subfigure}[t]{0.3\textwidth}
		\centering
		\includegraphics[width=\textwidth]{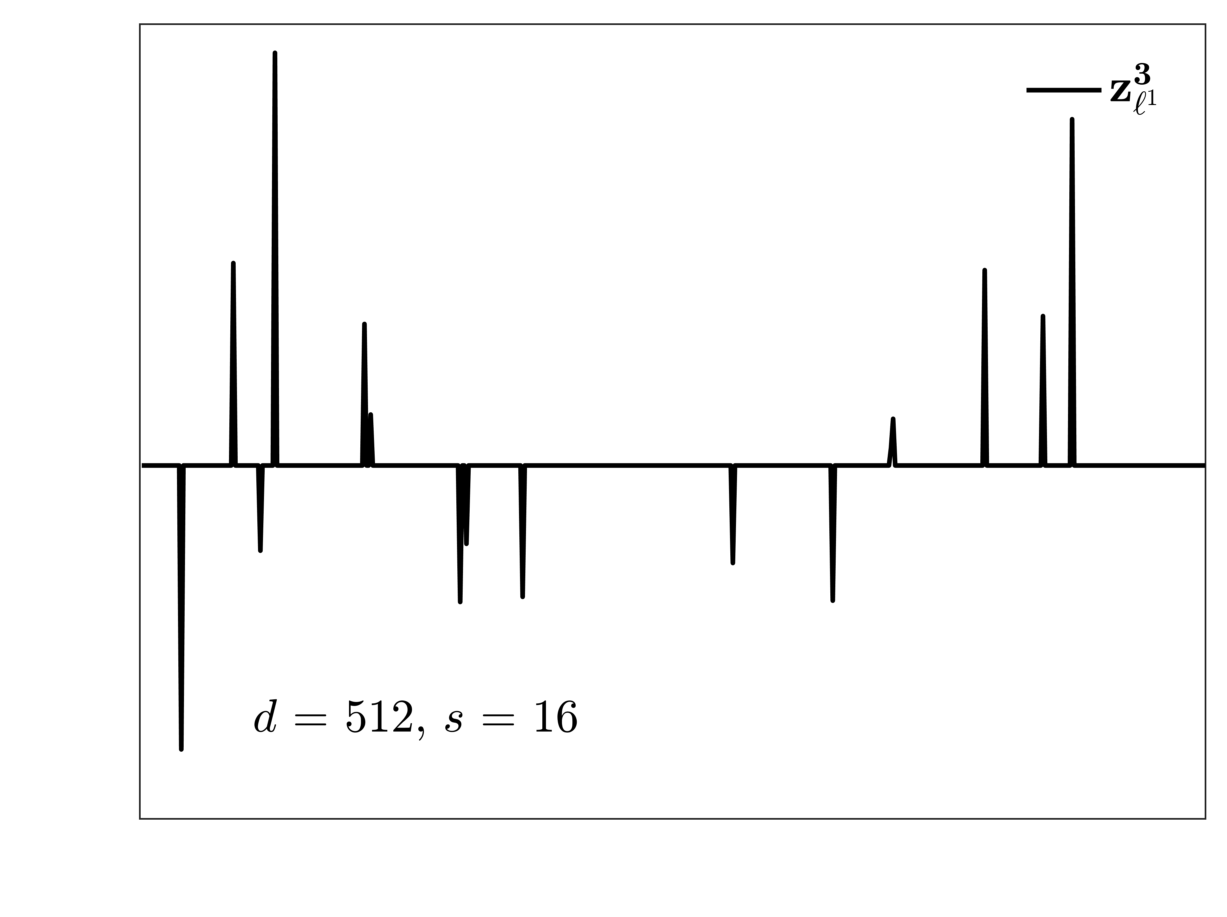}
		\caption{}
		\label{fig:coef:wave:9}
	\end{subfigure}%
	\vspace{.5\baselineskip}
	\begin{subfigure}[t]{0.3\textwidth}
		\centering
		\includegraphics[width=\textwidth]{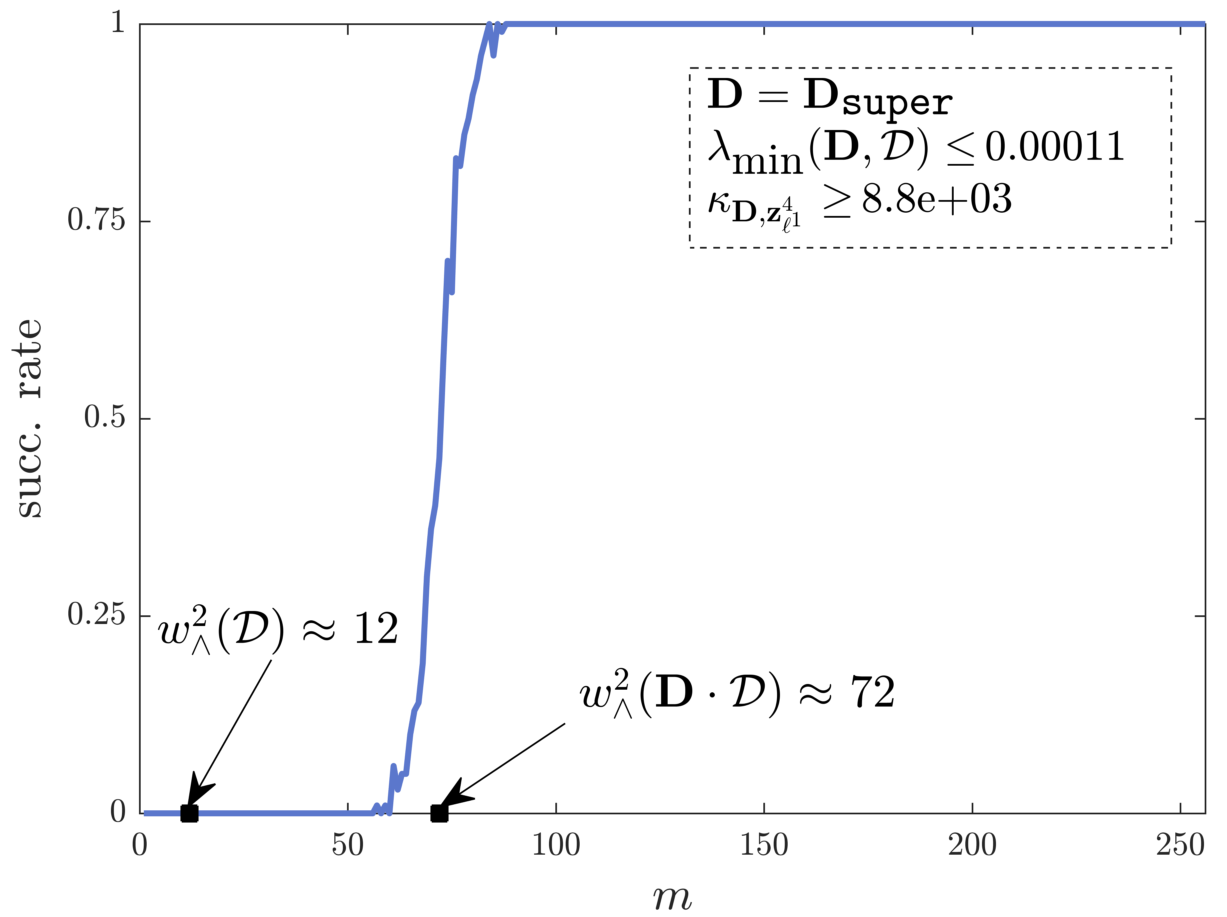}
		\caption{}
		\label{fig:coef:wave:10}
	\end{subfigure}%
	\qquad
	\begin{subfigure}[t]{0.3\textwidth}
		\centering
		\includegraphics[width=\textwidth]{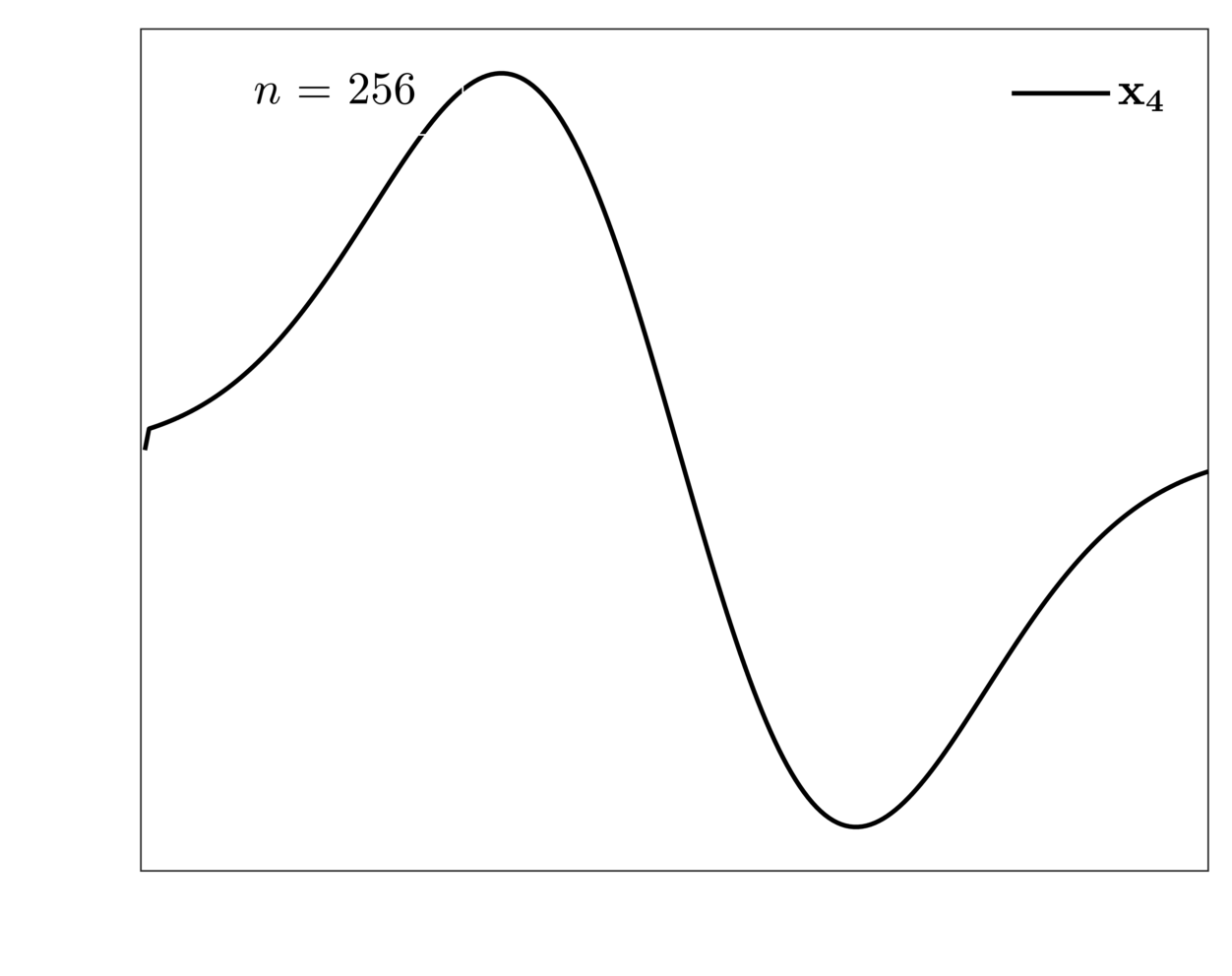}
		\caption{}
		\label{fig:coef:wave:11}
	\end{subfigure}%
	\qquad
	\begin{subfigure}[t]{0.3\textwidth}
		\centering
		\includegraphics[width=\textwidth]{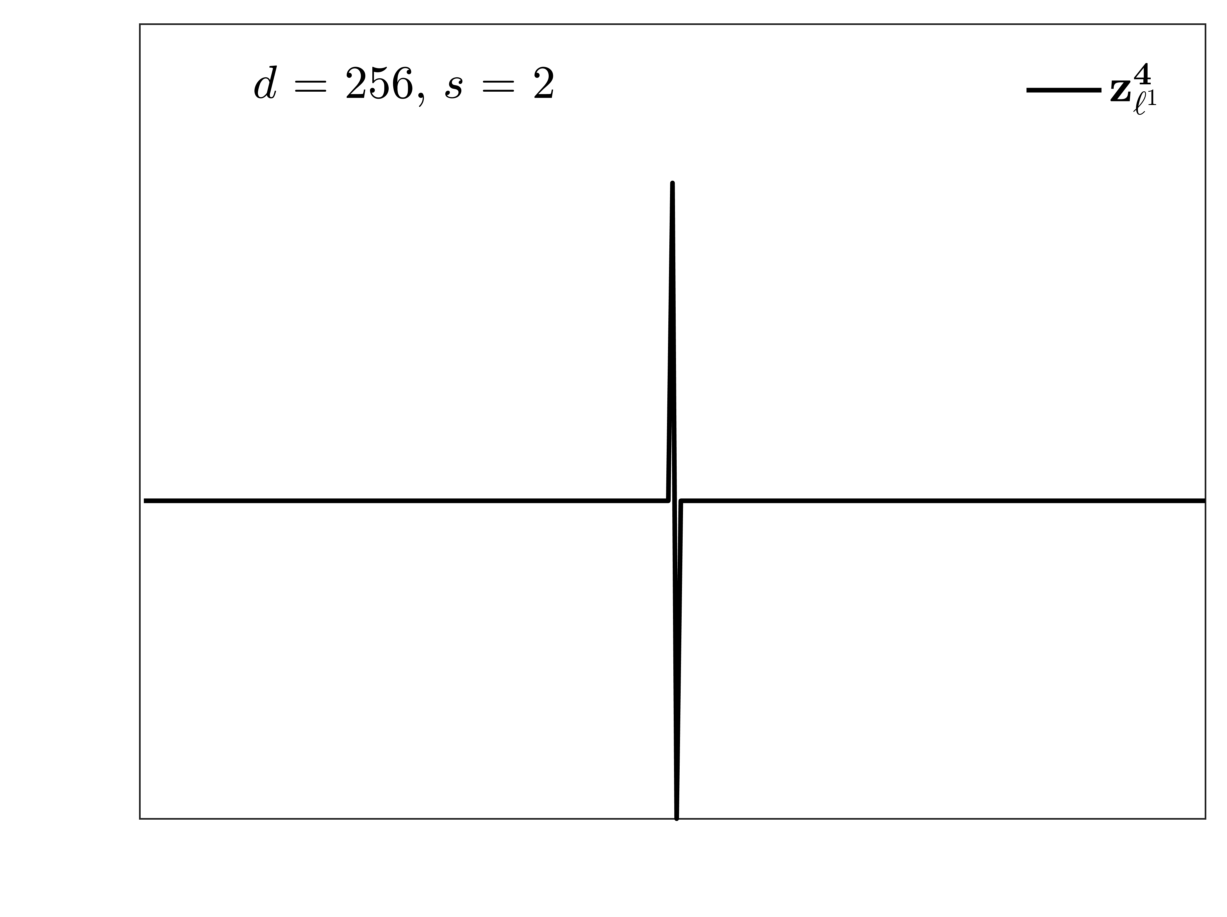}
		\caption{}
		\label{fig:coef:wave:12}
	\end{subfigure}%
	\caption{\textbf{Phase transitions of coefficient recovery by solving \coefNoisefree{}.} Empirical success rates and other key figures are reported in the first column, where we use the notation $\mathcal{D} =\ds{\Onenorm,\zl^{\smash{i}}}$. The coefficient vectors $\zl^i$ that are used in each experiment are shown in the third column. The associated signal vectors $\x_i = \Dict \zl^{\smash{i}}$ are displayed in the second column. The first two rows are relying on a redundant Haar wavelet frame, the third row is based on a Gaussian random matrix and the last row is using a dictionary inspired by super-resolution. }
	\label{fig:coef:wave}
\end{figure}

%

\subsection{Sampling Rates for Signal Recovery}
\label{sec:sampling_sig}

For the investigation of signal recovery via~\sigNoisefree, we also create phase transition plots by running Experiment~\ref{exp:coef} for different dictionary and signal combinations. Note that we also compute and store the signal error $\|\x_i - \sol\|_2 = \|\Dict \zl^{\smash{{i}}} - \Dict\solz\|_2$ in the third step of the experiment (in addition to $\|\zl^{\smash{{i}}} -\solz\|_2$). Recovery is declared successful if $\norm{\x_i - \sol}_2 < 10^{-5}$. 

\paragraph{Simulation Settings} 
Our first two examples are based on the same Haar wavelet system with 3 decomposition levels that is used in Section~\ref{sec:num_coef}.
The first signal is constructed by defining a coefficient vector $\z_1 \in \R^{1024}$ with a random support set of cardinality $s=35$ and random coefficients; see Subfigure~\ref{fig:sig:wave:10} for a visualization of $\z_1$ and Subfigure~\ref{fig:sig:wave:4} for the resulting signal $\x_1 = \Dictw \cdot \z_1$. In order to apply the result of Theorem~\ref{thm:sig}, we compute a minimal $\ell^1$-representer $\zl^{\smash{1}} \in \Zlset$ of $\x_1$; see Subfigure~\ref{fig:sig:wave:7}. 
The second coefficient vector $\z_2$ is created by defining two contiguous blocks of non-zero coefficients in a lower frequency decomposition scale, again with $s=35$; see Subfigure~\ref{fig:sig:wave:11} for a plot of $\z_2$ and Subfigure~\ref{fig:sig:wave:5} for the resulting signal $\x_2 = \Dictw  \cdot \z_2$. A corresponding minimal $\ell^1$-representer $\zl^{\smash{2}} \in \Zlset$ of $\x_2$ is shown in Subfigure~\ref{fig:sig:wave:8}. 

Finally, for the third setup we are choosing a simple example in 2D. In order to keep the computational burden manageable, we restrict ourselves to a $28 \times 28$-dimensional digit from the MNIST data set~\cite{mnist}, i.e., the vectorized image is of size $n=28^2 = 784$. As a sparsifying system we utilize a dictionary that is based on the 2D discrete cosine transform (dct-2). It makes use of \texttt{Matlab}'s standard dct-2 transform as convolution filters on $3 \times 3$ patches. The resulting operator is denoted by $\Dict = \Dict_{\text{dct-2}} \in \R^{n \times 9 n}$, i.e., $d = 9n$. Note that such a convolution sparsity model is frequently used in the literature, in particular also with learned filters, e.g., see \emph{convolutional sparse coding} in \cite{bristow2013}. Although the dct-2 filters might not be a perfect match for MNIST digits, we consider them as a classical representative that is well suited to demonstrate the predictive power of our results. In order to construct a suitable sparse representation $\z_3\in\R^d$ of an arbitrarily picked digit in the database, we make use of the orthogonal matching pursuit algorithm~\cite{omp1993}; see Subfigure~\ref{fig:sig:wave:12} for a visualization of $\z_3$. The resulting digit $\x_3 = \Dict_{\text{dct-2}} \cdot \z_3$ is displayed in Figure~\ref{fig:sig:wave:6}. A minimal $\ell^1$-representer $\zl^{\smash{3}}$, which is needed to apply Theorem~\ref{thm:sig}, is shown in Subfigure~\ref{fig:sig:wave:9}.

\paragraph{Results} Our conclusions on the numerical experiments shown in Figure~\ref{fig:sig:wave}
are reported in the following:

\begin{enumerate}[(i)] 
\setcounter{enumi}{5}

  \item The convex program \sigNoisefree{} obeys a sharp phase transition in the number of measurements. Due to the equivalent, gauge-based reformulation~\eqref{eq:gauge} of \sigNoisefree{}, this observation is predicted by the works~\cite{amelunxen2014edge,tropp2014convex}. However, a recovery of a coefficient representation via solving \coefNoisefree{} is impossible in all three examples, even for $m=n$.

  \item For any $\zl \in \Zlset$, the conic mean width $\cmw[2]{\Dict \cdot \ds{\Onenorm;\zl}}$ accurately describes the sampling rate of \sigNoisefree{}, as predicted by Theorem~\ref{thm:sig}. Indeed,  in all three cases of Figure~\ref{fig:sig:wave}, the estimated $\cmw[2]{\Dict \cdot \ds{\Onenorm;\zl^{{i}}}}$ matches precisely the 50\% recovery rate.
 
  \item In contrast, for any other sparse representation $\z \not\in \Zlset$, the conic width ${\cmw[2]{\Dict \cdot \ds{\Onenorm;\z}}}$ does not describe the sampling rate of \sigNoisefree{}, in general. Indeed, observe that we have $\cmw[2]{\Dict \cdot \ds{\Onenorm;\z_i}} \approx n$ in all three examples. Also note that $\norm{\vec{z}_1}_0 = 35 = \norm{\vec{z}_2}_0$, however, the locations of the phase transitions deviate considerably. Similarly, although $\|\zl^{{1}}\|_0 < \|\zl^{{2}}\|_0$, we have that $\cmw[2]{\Dict \cdot \ds{\Onenorm;\zl^{\smash{1}}}} > \cmw[2]{\Dict \cdot \ds{\Onenorm;\zl^{\smash{2}}}}$. This observation is yet another indication that sparsity alone is not a good proxy for the sampling rate of $\ell^{\smash{1}}$-synthesis, in general. 
  
 
\end{enumerate}

\subsection{Creating a ``Full'' Phase Transition}
\label{sec:2d_pt}

Let us now focus on the phase transition plots shown in Figure~\ref{fig:sig:pt_2}. 
Up to now, we have only considered one specific signal at a time. However, it is also of interest to assess the quality of our results if the ``complexity'' of the underlying signals is varied. In the classical situation of $\Dict$ being an ONB, the location of the phase transition is entirely determined by the sparsity of the underlying signal. Hence, it is a natural choice to create phase transitions over the sparsity, as it is for instance done in \cite{amelunxen2014edge}. Recalling Claims (iv) and (viii), it might appear odd to do the same in the case of a redundant dictionary. However, as the result of Figure~\ref{fig:sig:pt_2} shows, if the support is chosen uniformly at random, sparsity is still a somewhat reasonable proxy for the sampling rate.  Indeed, these plots are created by running Experiment~\ref{exp:pt} with $\Dict = \Dictw \in \R^{256\times 1024}$, maximal sparsity $s_0 = 125$ and displaying the empirical success rates of coefficient and signal recovery, respectively. Additionally the dotted line shows the averaged conic mean width values $\cmw[2]{\Dict \cdot \ds{\Onenorm;\zl}}$.

\begin{experiment}[Phase transition of Figure~\ref{fig:sig:pt_2}]\leavevmode
\label{exp:pt}
\vspace{-.25\baselineskip}\hrule\vspace{.5\baselineskip}

\myhangindent{Input: \ }
\expkwd{Input:} Dictionary $\Dict \in \R^{n \times d}$, maximal sparsity $s_0 \in [n]$. 

\vspace{.5\baselineskip}
\expkwd{Compute:} Repeat the following procedure $500$ times for each $s \in \{1,\dots,s_0\}$:
\begin{expstep}
	\item 
	 Select a set $S \subset [n]$ uniformly at random with $\# S = s$. Then draw a standard Gaussian random vector $\vec{c} \in \R^s$ and define $\gtz$ by setting $(\gtz)_S = \vec{c}$ and $(\gtz)_{\setcompl{S}} = \vec{0}$. 
	\item 
	 Define the signal $\gt = \Dict \gtz$ and compute a minimal $\ell^1$-representation $\zl\in\Zlset$ of $\gt$. Compute the conic mean width $\cmw[2]{\Dict \cdot \ds{\Onenorm;\zl}}$. 
	\item 
     Run Experiment~\ref{exp:coef} with $\Dict$ and $\zl$ as input, where the number of repetitions is lowered to 5. In the third step, coefficient/signal recovery is declared successful if $\norm{\zl-  \solz}_2 < 10^{-5}$ or if $ \norm{\Dict  \zl - \sol}_2 = \norm{\gt - \Dict  \solz}_2 < 10^{-5}$, respectively.
\end{expstep}
\end{experiment}

First note, that the averaged values of the conic mean width perfectly match the center of the phase transition of\sigNoisefree{} in Subfigure~\ref{fig:sig:pt_2:2}, as it is predicted by Theorem~\ref{thm:sig}. However, observe that for sparsity values between $s \approx 20$ and $s\approx 80$ the transition region is spread out in the vertical direction, cf.~\cite{amelunxen2014edge}. This phenomenon can be related to Claim (viii): Given that sparsity alone is not a good proxy for the sample complexity of \sigNoisefree{}, averaging over different instances with the same sparsity necessarily results in a smeared out transition area.  

Regarding the phase transition in Subfigure~\ref{fig:sig:pt_2:1}, we observe that the location of the phase transition is also determined by $\cmw[2]{\Dict \cdot \ds{\Onenorm;\zl}}$; provided that coefficient recovery is possible, i.e., if $\lmin{\Dict}{\dc{\Onenorm;\zl}} > 0$. For sparsity values $s\geq50$, recovery of coefficients appears to be impossible, whereas signal and coefficient recovery seems to be equivalent for very small sparsity values (i.e., $s\leq 5$). The interval in between forms a transition region, in which coefficient recovery becomes gradually less likely. We suspect that with more repetitions in Experiment~\ref{exp:pt}, the empirical success rates on this interval would eventually smooth out. 

We conclude that:
\begin{enumerate}[(i)]
	\setcounter{enumi}{8}
  \item Whether $\lmin{\Dict}{\dc{\Onenorm;\zl}} \neq 0$, i.e., whether coefficient recovery is possible, is again a property that is non-uniform in the sparsity of $\zl$.
\end{enumerate}

\begin{figure}[!t]
	\centering
	\begin{subfigure}[t]{0.30\textwidth}
		\centering
		\includegraphics[width=\textwidth]{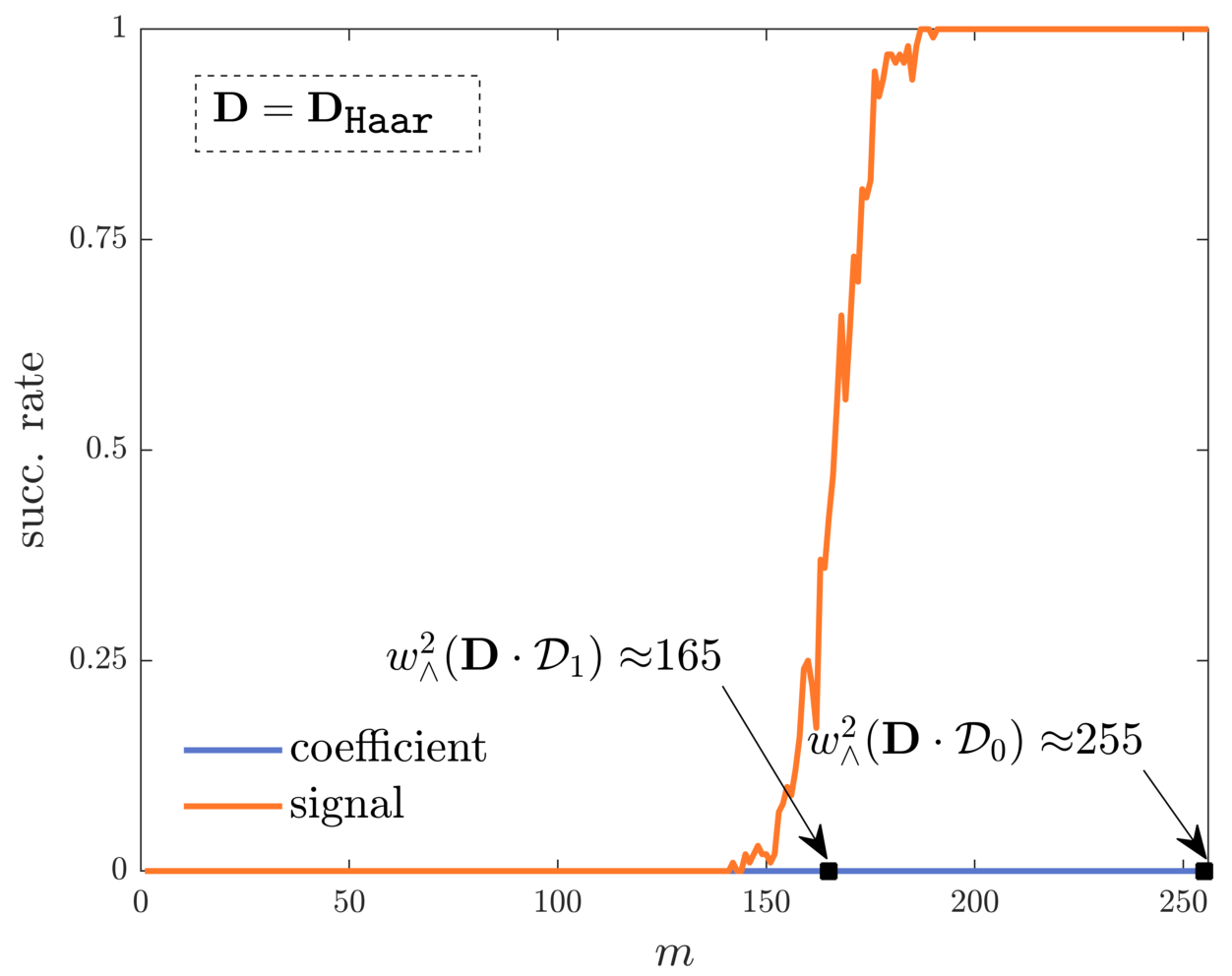}
		\caption{}
		\label{fig:sig:wave:1}
	\end{subfigure}%
	\qquad
	\begin{subfigure}[t]{0.30\textwidth}
		\centering
		\includegraphics[width=\textwidth]{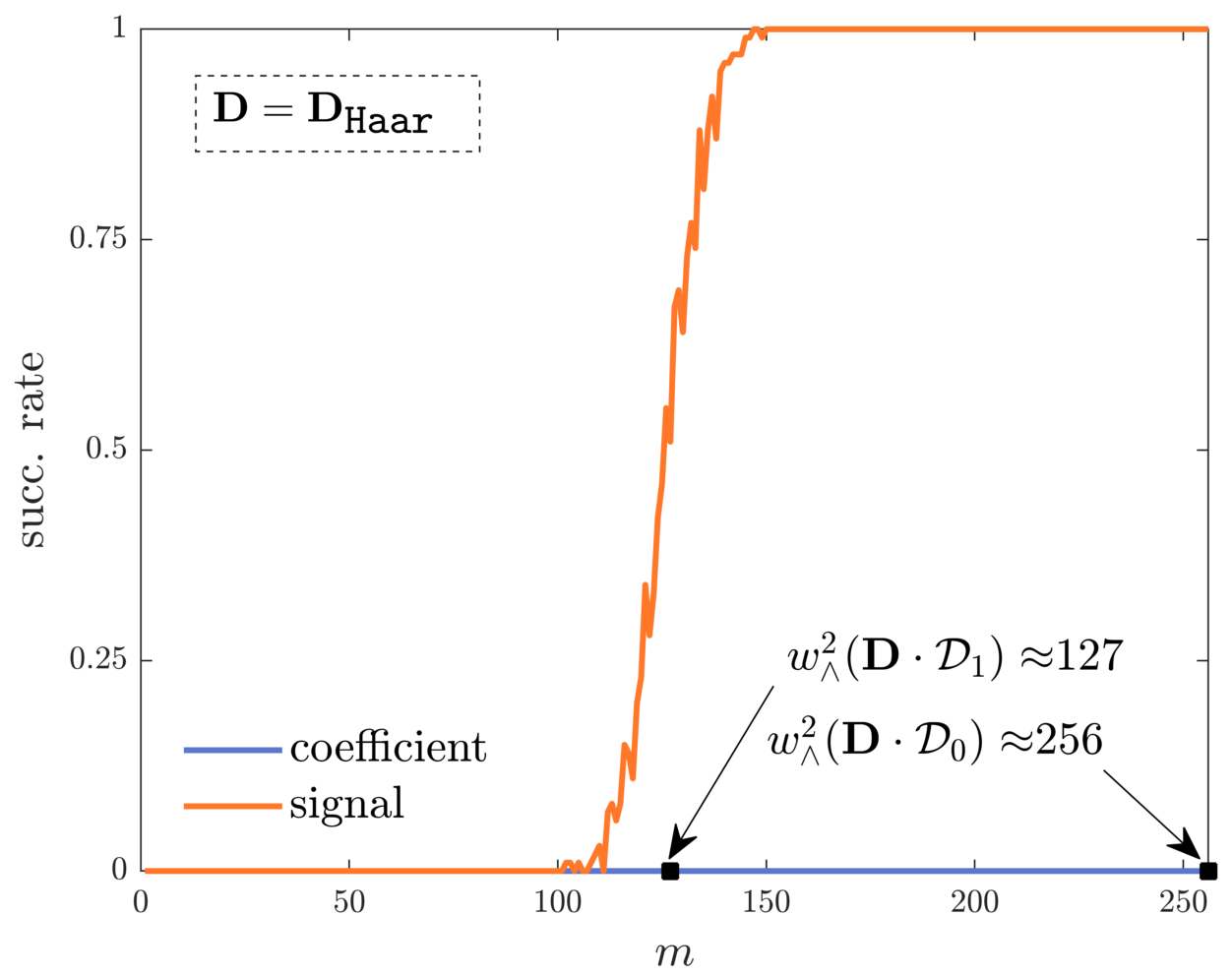}
		\caption{}
		\label{fig:sig:wave:2}
	\end{subfigure}%
	\qquad
	\begin{subfigure}[t]{0.30\textwidth}
		\centering
		\includegraphics[width=\textwidth]{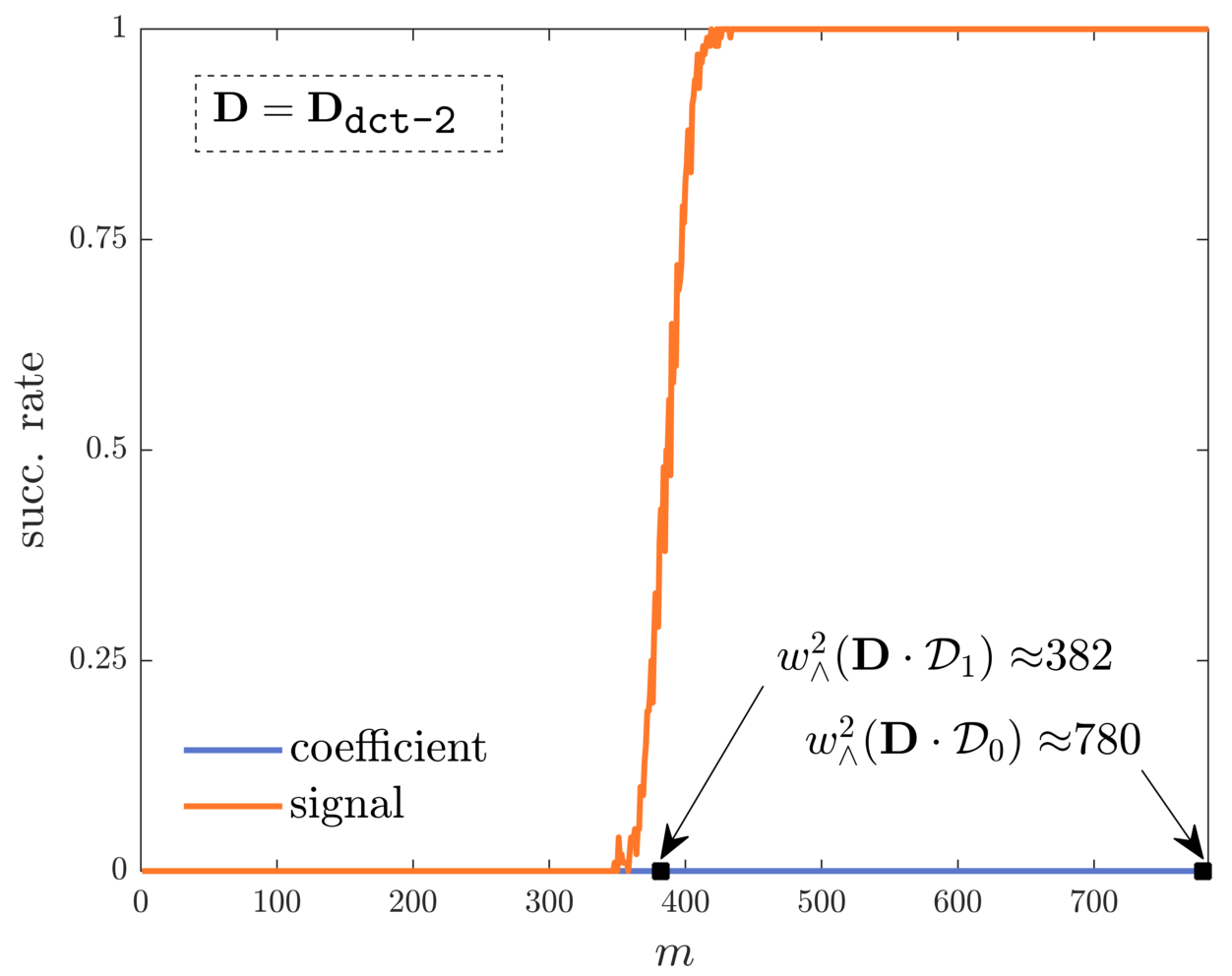}
		\caption{}
		\label{fig:sig:wave:3}
	\end{subfigure}%
	\vspace{.5\baselineskip}
		\begin{subfigure}[t]{0.3\textwidth}
		\centering
		\includegraphics[width=\textwidth]{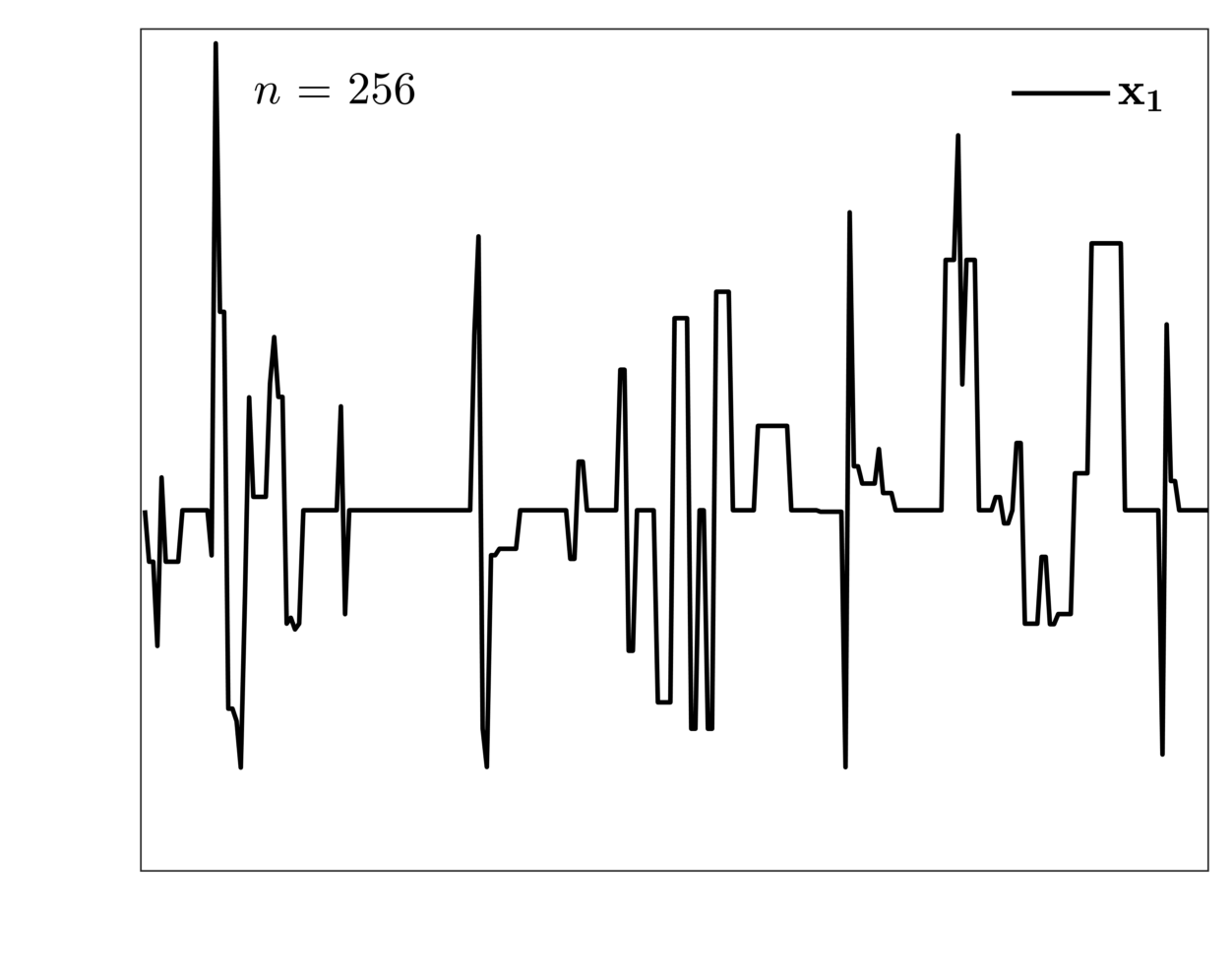}
		\caption{}
		\label{fig:sig:wave:4}
	\end{subfigure}%
	\qquad
	\begin{subfigure}[t]{0.3\textwidth}
		\centering
		\includegraphics[width=\textwidth]{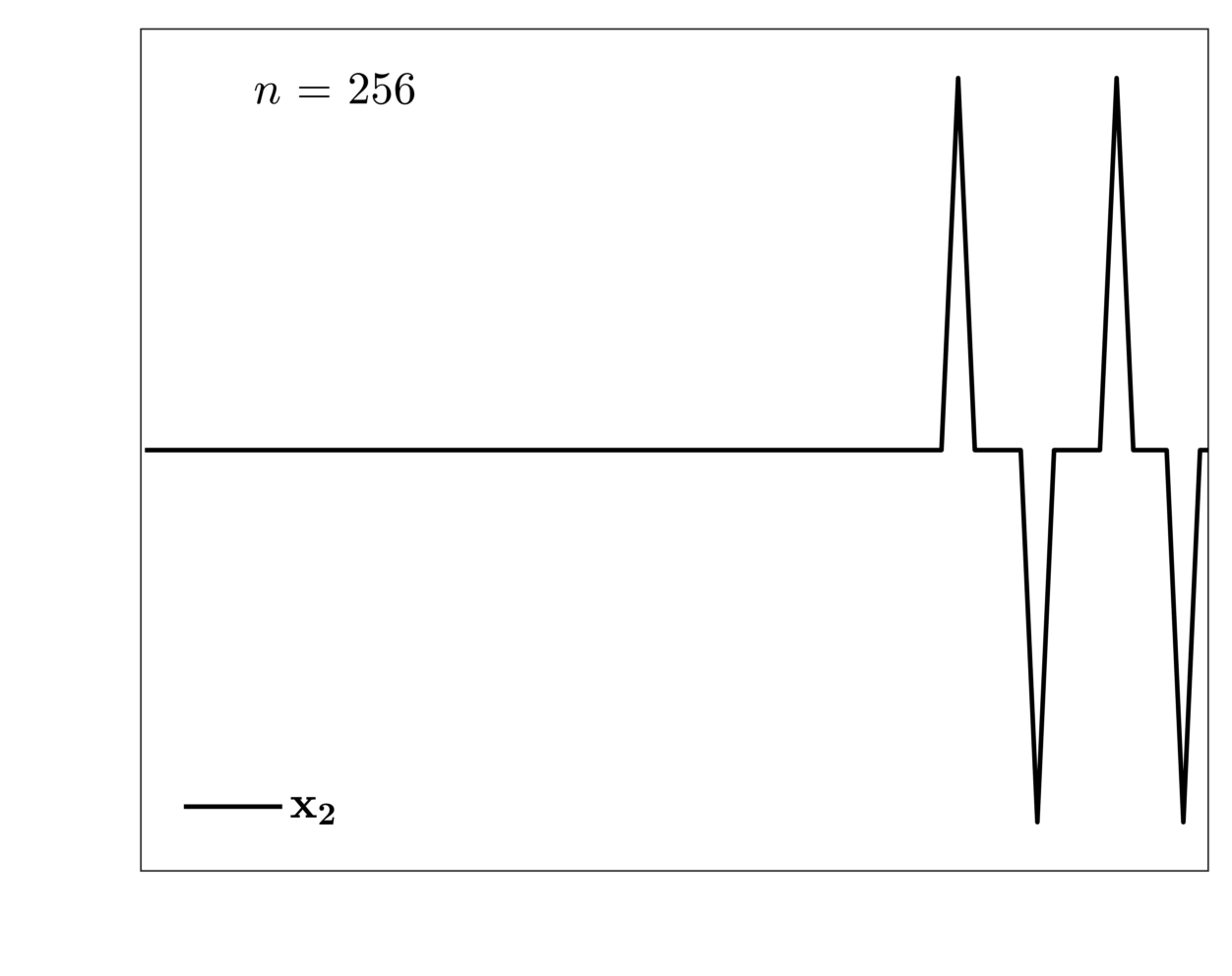}
		\caption{}
		\label{fig:sig:wave:5}
	\end{subfigure}%
	\qquad
	\begin{subfigure}[t]{0.3\textwidth}
		\centering
		\includegraphics[width=\textwidth]{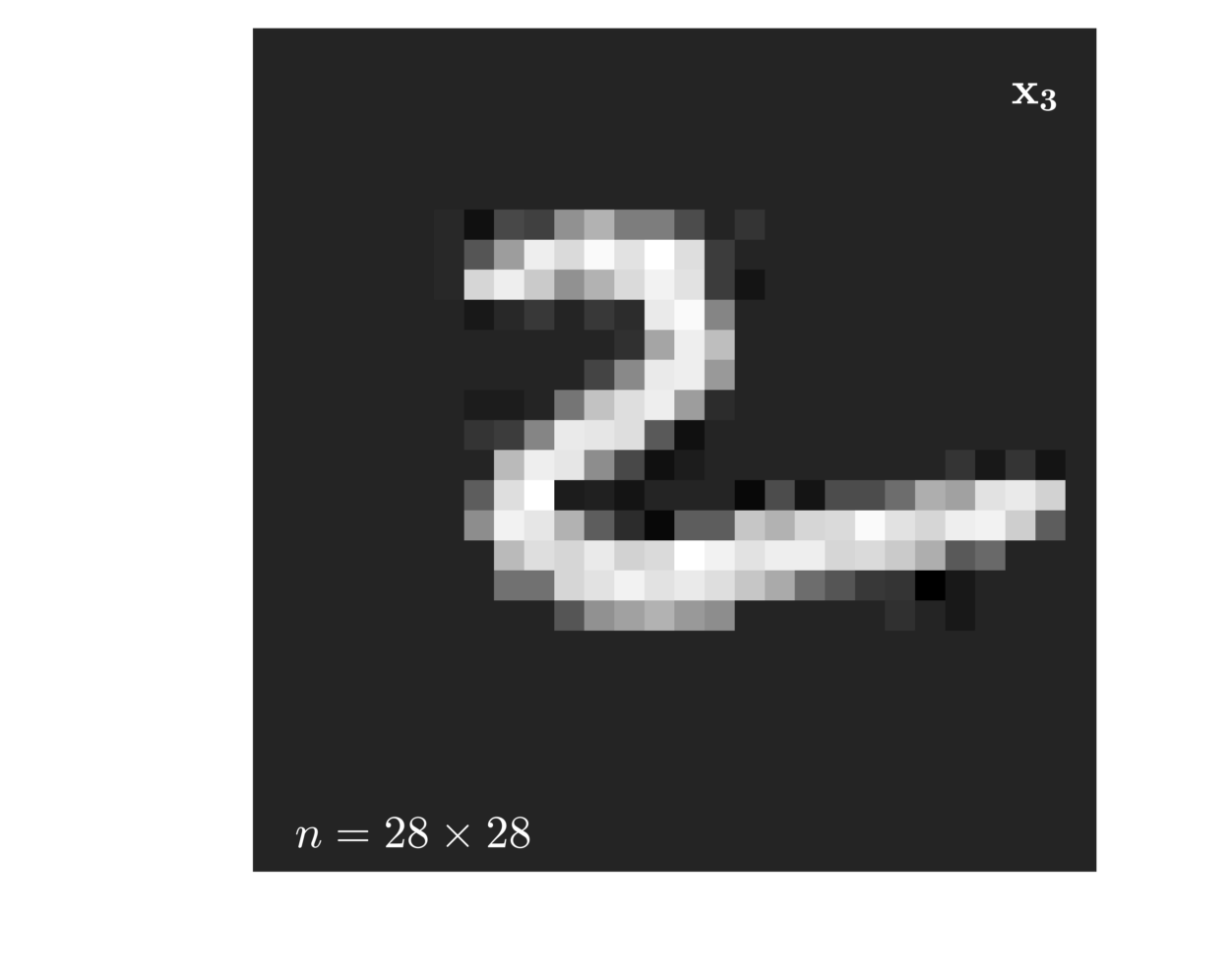}
		\caption{}
		\label{fig:sig:wave:6}
	\end{subfigure}%
	\vspace{.5\baselineskip}
	\begin{subfigure}[t]{0.30\textwidth}
		\centering
		\includegraphics[width=\textwidth]{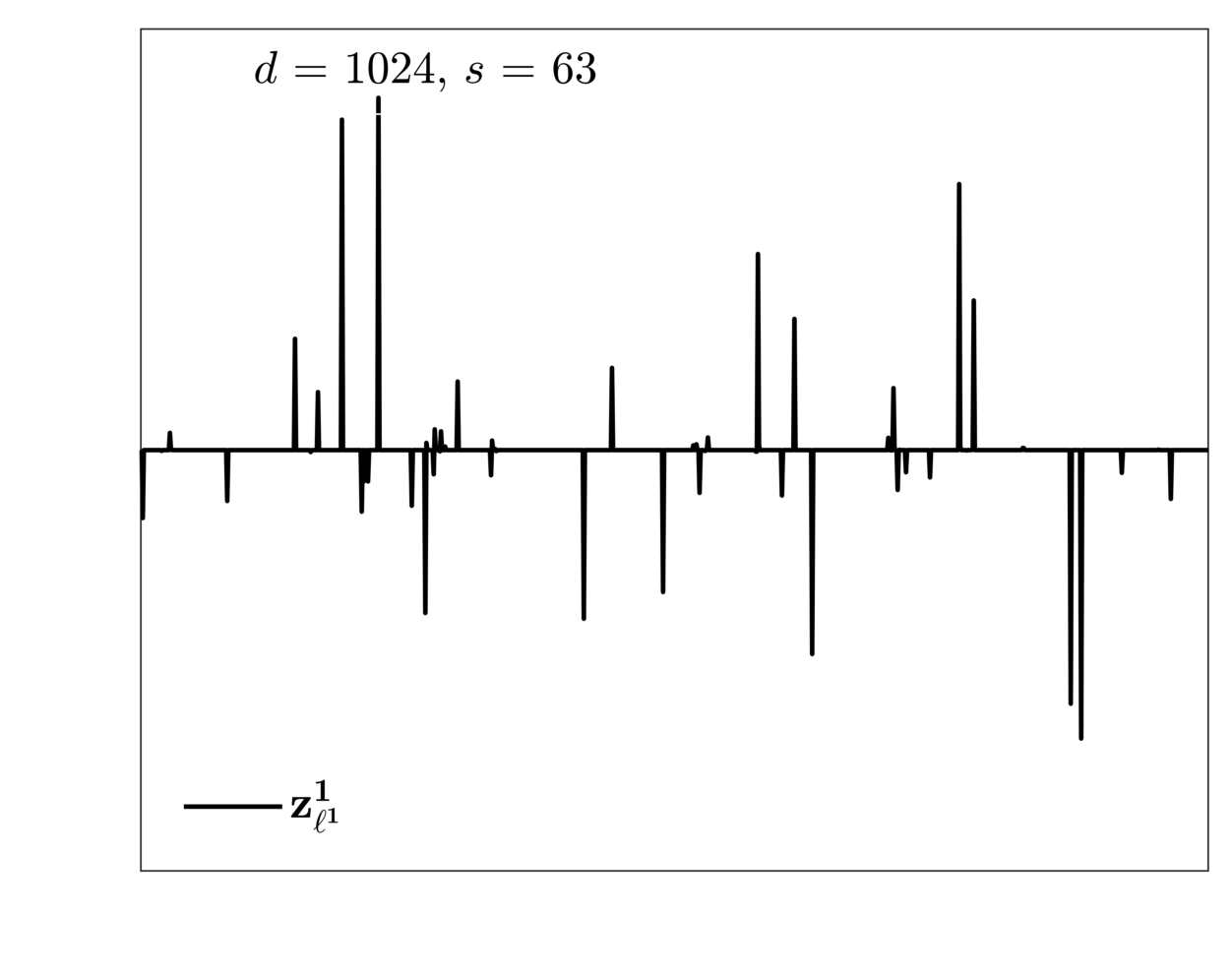}
		\caption{}
		\label{fig:sig:wave:7}
	\end{subfigure}%
	\qquad
	\begin{subfigure}[t]{0.30\textwidth}
		\centering
		\includegraphics[width=\textwidth]{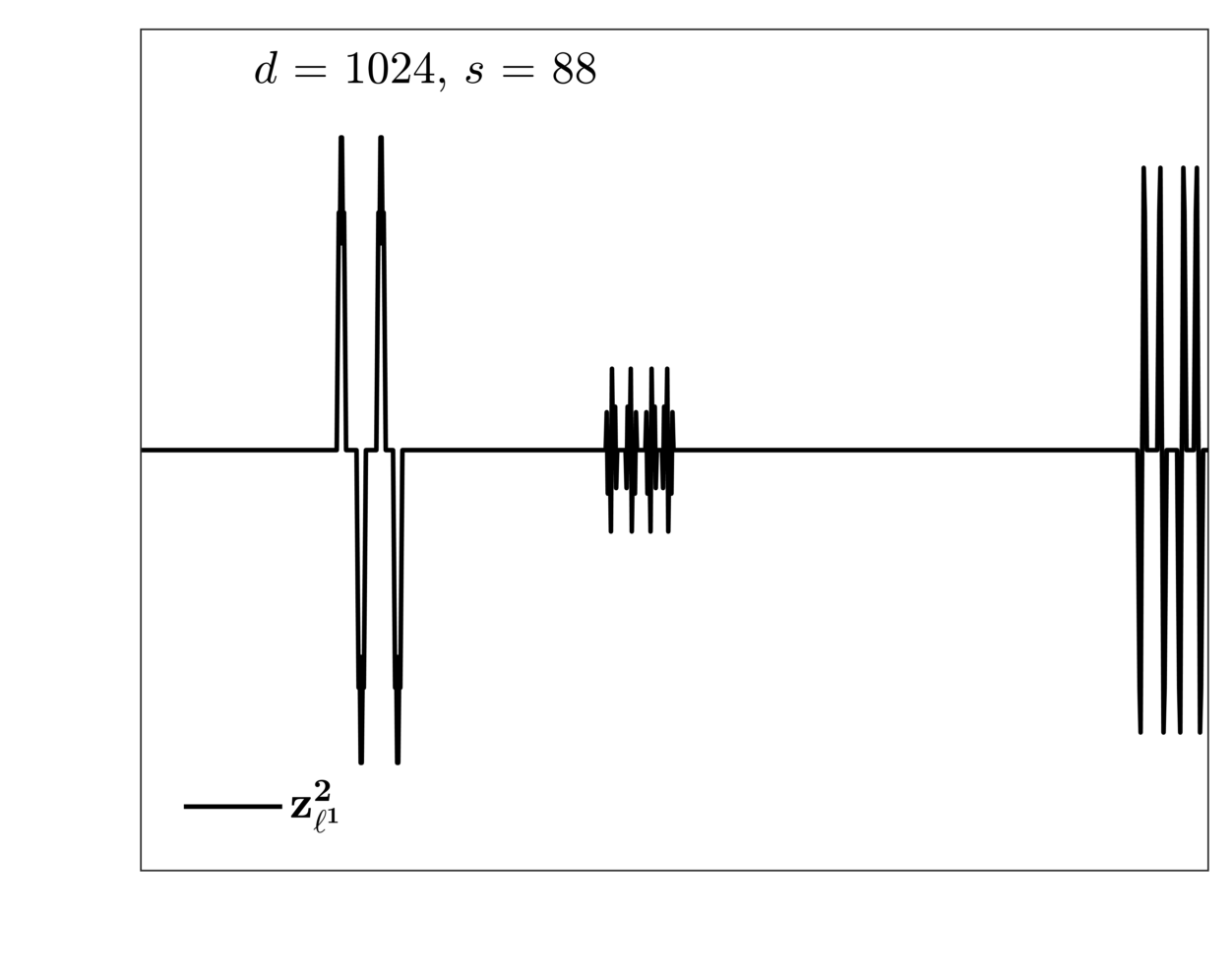}
		\caption{}
		\label{fig:sig:wave:8}
	\end{subfigure}%
	\qquad
	\begin{subfigure}[t]{0.30\textwidth}
		\centering
		\includegraphics[width=\textwidth]{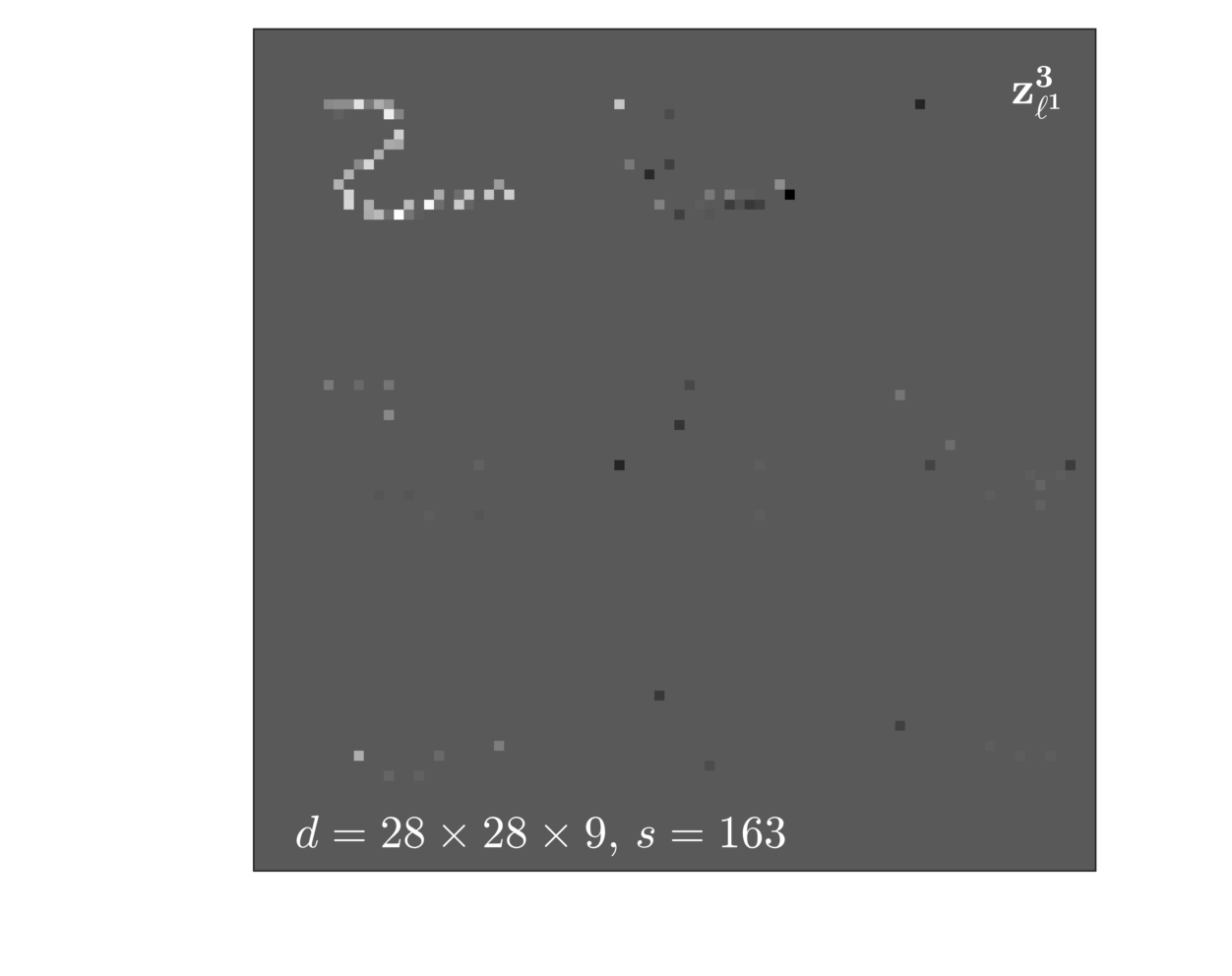}
		\caption{}
		\label{fig:sig:wave:9}
	\end{subfigure}%
	\vspace{.5\baselineskip}
	\begin{subfigure}[t]{0.30\textwidth}
		\centering
		\includegraphics[width=\textwidth]{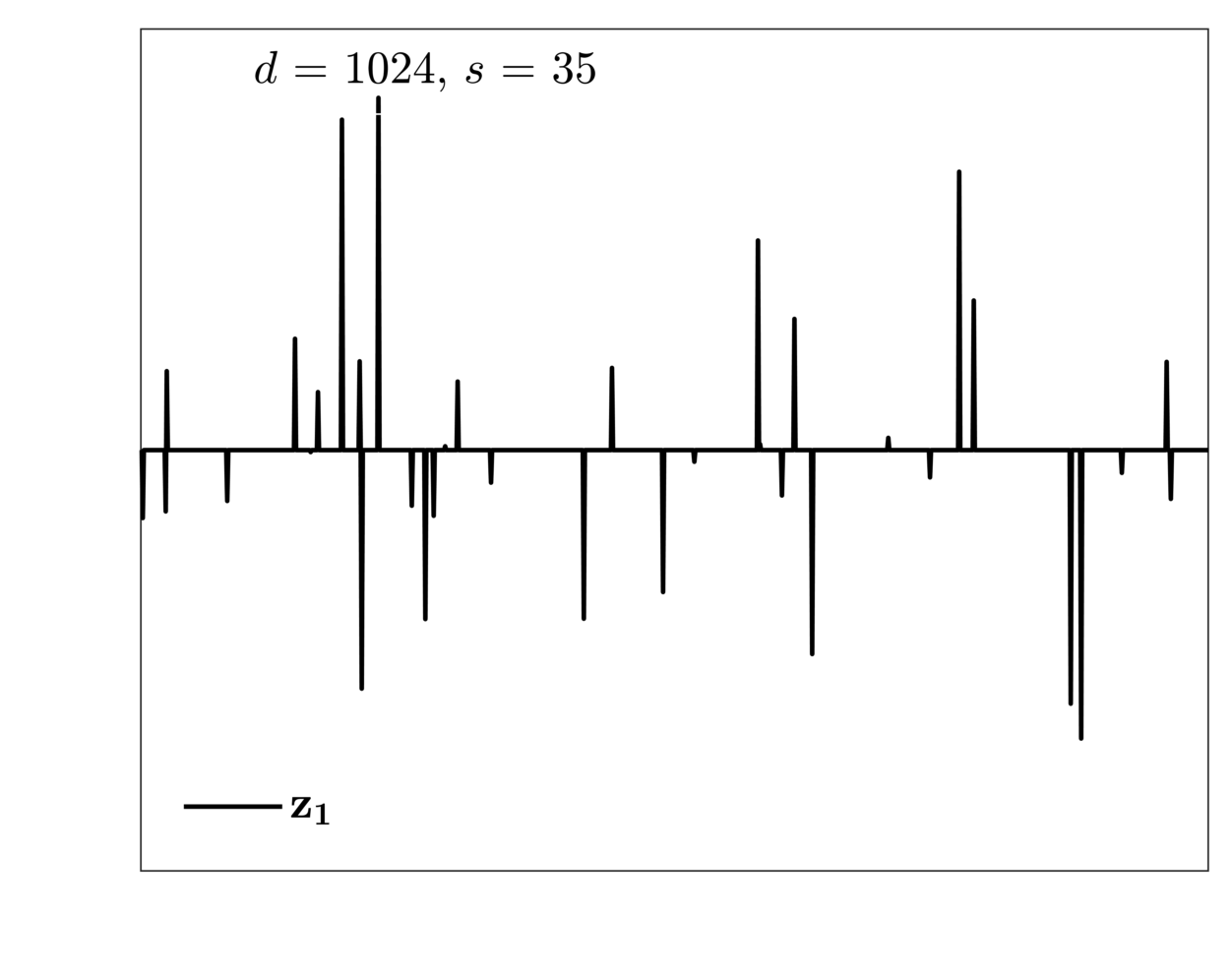}
		\caption{}
		\label{fig:sig:wave:10}
	\end{subfigure}%
	\qquad
	\begin{subfigure}[t]{0.30\textwidth}
		\centering
		\includegraphics[width=\textwidth]{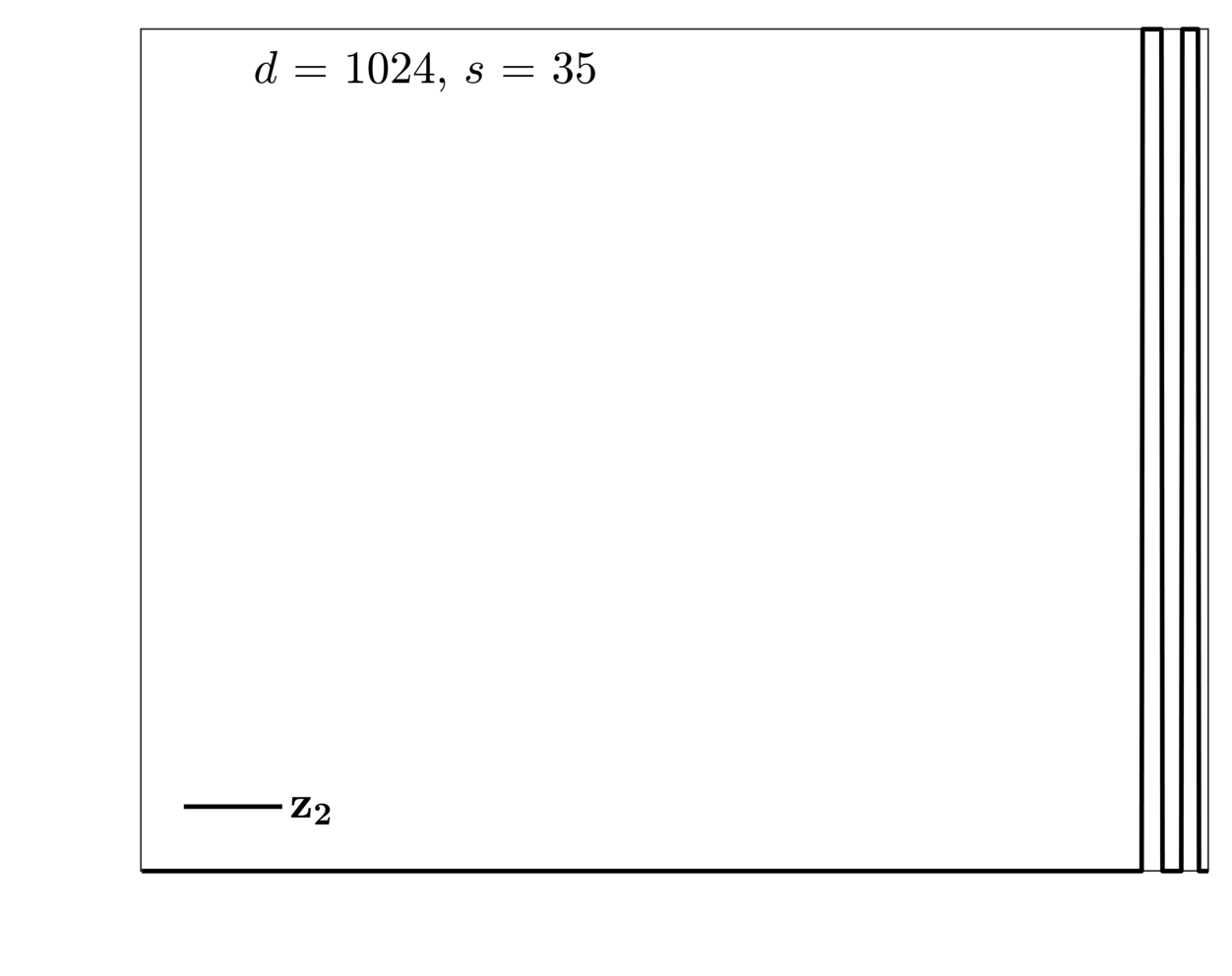}
		\caption{}
		\label{fig:sig:wave:11}
	\end{subfigure}%
	\qquad
	\begin{subfigure}[t]{0.30\textwidth}
		\centering
		\includegraphics[width=\textwidth]{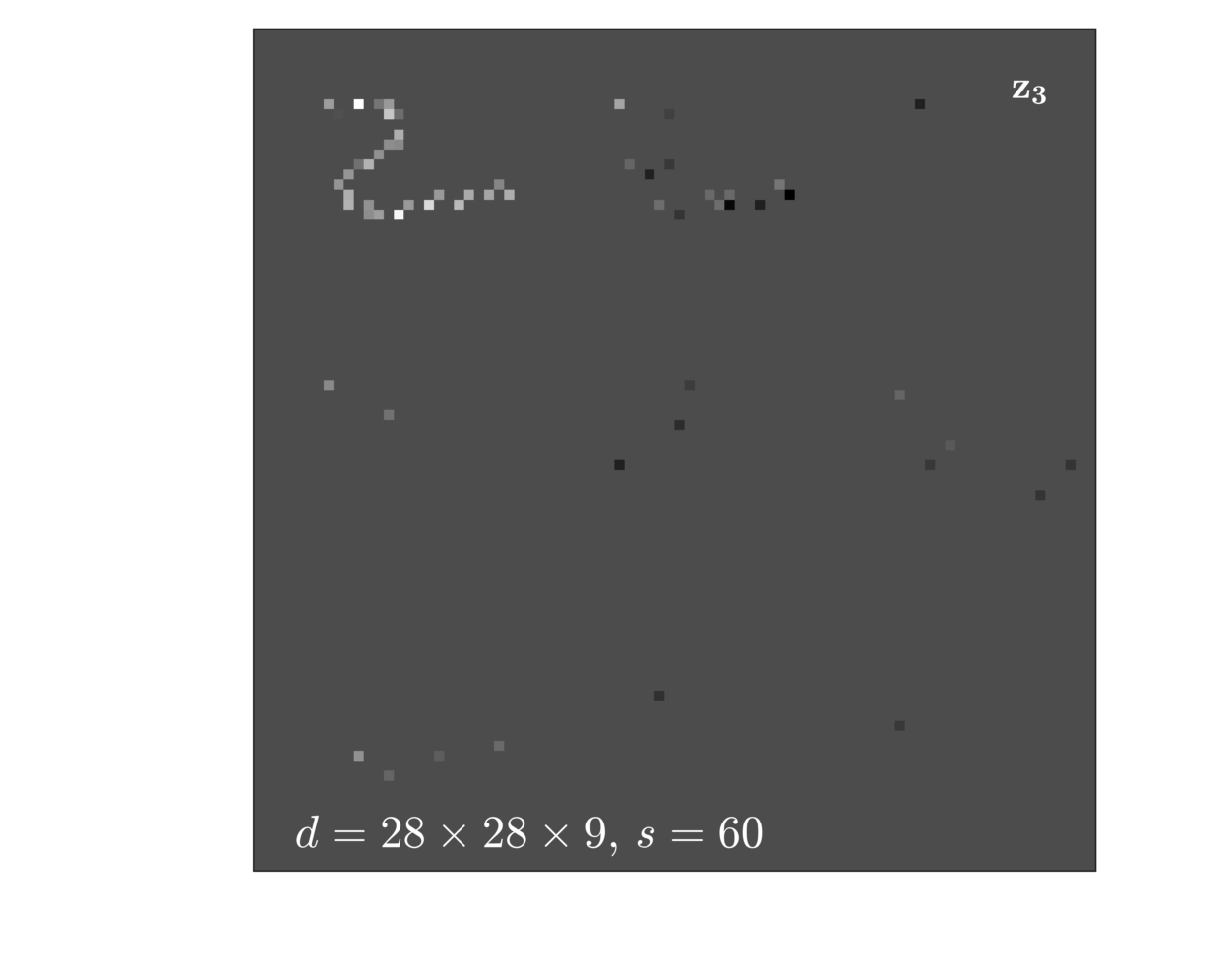}
		\caption{}
		\label{fig:sig:wave:12}
	\end{subfigure}%
	\caption{\textbf{Phase transitions of signal recovery by solving \sigNoisefree{}.} Empirical success rates and other key figures are reported in the first row, where we use the notation $\mathcal{D}_1 = \ds{\Onenorm,\zl^{\smash{i}}}$ and $\mathcal{D}_0 = \ds{\Onenorm,\z_i}$. The signals $\x_i$ that are used in each experiment are shown in the second row. The associated minimal $\ell^1$-representers $\zl^{\smash{i}}$ are displayed in the third row and the original coefficient representations $\z_i$ are shown in the fourth row. The first two columns are relying on a redundant Haar wavelet frame and the third column is based on the dct-2. Note that in all three examples, coefficient recovery is not possible since $\smash{\lmin{\Dict}{\dc{\Onenorm,\zl^{\smash{i}}}} = 0}$.}
	\label{fig:sig:wave}
\end{figure}

\subsection{Robustness to Noise}
\label{sec:robustness_noise}

The purpose of this last numerical simulation is to analyze coefficient and signal recovery with respect to robustness to measurement noise. To that end, we run Experiment~\ref{exp:robust} for different setups, which are reported below.

\begin{experiment}[Robustness to Measurement Noise]\leavevmode
\label{exp:robust}
\vspace{-.25\baselineskip}\hrule\vspace{.5\baselineskip}

\myhangindent{Input: \ }
\expkwd{Input:} Dictionary $\Dict \in \R^{n \times d}$, number of measurements $m$, minimal $\ell^1$-representation $\zl$ of the signal $\gt \in \R^n$, range of noise levels $H$.

\vspace{.5\baselineskip}
\expkwd{Compute:} Repeat the following procedure $100$ times for every $\noiseparam \in H$:
\begin{expstep}
	 \item 
		Draw a standard i.i.d.\ Gaussian random matrix $\Meas \in \R^{m\times n}$ and determine the noisy measurement vector $\y = \Meas\gt + \noiseparam \cdot \noise$, where $\norm{\noise}_2 = 1$.
	\item 
		Solve the program \eqref{eq:coef} to obtain an estimator $\solz \in \R^d$. 
	\item
		Compute and store the recovery errors $\norm{\zl - \solz}_2$ and $\norm{\gt - \sol}_2 = \norm{\Dict \zl - \Dict \solz}_2$.
\end{expstep}
\end{experiment}

\paragraph{Simulation settings} First, we choose the same dictionary and signal combination as in Section~\ref{sec:num_coef} and restrict the noise level to $H = \left\{0,0.05,0.1,0.15,\dots,1 \right\}$. Furthermore, we consider the 1D examples of Section~\ref{sec:sampling_sig}, together with the noise range $H = \left\{0,0.005,0.01,\dots,0.1 \right\}$. Recall that the difference of these two setups is that $\lmin{\Dict}{\dc{\Onenorm,\zl^{\smash{i}}}} > 0$ in the first case, whereas $\lmin{\Dict}{\dc{\Onenorm,\zl^{\smash{i}}}} = 0$ in the second case. In all experiments, we roughly pick the number of measurements as $m \approx \cmw[2]{\Dict \cdot \ds{\Onenorm;\zl^{\smash{i}}}} +40$ to ensure that Theorem~\ref{thm:coeff} (or Theorem~\ref{thm:sig}, respectively) is applicable. The averaged coefficient and signal recovery errors are displayed in Figure~\ref{fig:noise}, together with the theoretical upper bound on the signal error of Equation~\eqref{eq:rec_sig1}. Note that it is not possible to show the corresponding error bound for coefficient recovery. In the first set of examples, we do not have access to $\lmin{\Dict}{\dc{\Onenorm,\zl^{\smash{i}}}}$ and in the last two cases, $\lmin{\Dict}{\dc{\Onenorm,\zl^{\smash{i}}}} = 0$ and therefore Theorem~\ref{thm:coeff} is not applicable. Nevertheless, it is possible to obtain an upper bound for the latter quantity, as outlined in the Appendix~\ref{sec:num_details}. If $\Dict = \Dict_{\texttt{rand}}$, it is additionally possible to use the result on minimum conic singular values of Gaussian matrices to get a lower bound with high probability~\cite[Proposition 3.3]{tropp2014convex}.

\paragraph{Results} We summarize the findings of the results shown in Figure~\ref{fig:noise} below:
\begin{enumerate}[(i)] 
\setcounter{enumi}{9}
  \item If the number of measurements exceeds the sampling rate in Theorem~\ref{thm:sig}, signal recovery via solving the Program~\eqref{eq:sig} is robust to measurement noise. This phenomenon holds true without any further assumptions. Indeed, observe  that in all simulations of Figure~\ref{fig:noise}, the signal error $\norm{\gt - \sol}$ lies below its theoretical upper bound of Equation~\eqref{eq:rec_sig1}. 
  
  \item If $\zl$ is the unique minimal $\ell^1$-representation of $\gt$ (i.e., if $\lmin{\Dict}{\dc{\Onenorm,\zl}} > 0$) and if the number of measurements exceeds the sampling rate in Theorem~\ref{thm:coeff}, it is possible to robustly recovery $\zl$. However, in contrast to signal recovery, the robustness is influenced by the ``stability'' of the minimal $\ell^1$-representation of $\zl$ in $\Dict$, i.e., by the value of $\lmin{\Dict}{\dc{\Onenorm,\zl}}$ in the error bound~\eqref{eq:rec_coef}. Indeed, it is possible that the signal $\gt$ is more robustly recovered than its coefficients $\zl$, or vice versa. This can be seen by comparing coefficient and signal recovery in the Subfigures~\ref{fig:noise:1}-\ref{fig:noise:4}.\footnote{Note that the quantity  $\lmin{\Dict_{\texttt{super}}}{\dc{\Onenorm,\zl^4}}$ is very small. Hence, even for a small amount of noise the error of Equation~\eqref{eq:rec_coef} explodes. For $\noiseparam>0.1$ the error stays roughly constant since the solution $\solz$ of \eqref{eq:coef} is always close to $\vec{0}$.}  If $\lmin{\Dict}{\dc{\Onenorm,\zl}}  \ll 1$, coefficient recovery is less robust than signal recovery, see Subfigures~\ref{fig:noise:1}, \ref{fig:noise:2} and \ref{fig:noise:4}. However, if $\lmin{\Dict}{\dc{\Onenorm,\zl}}  \gg 1$, the contrary holds true, see Subfigure~\ref{fig:noise:3}. 
  
\end{enumerate}
  

\begin{figure}
	\centering
	\begin{subfigure}[t]{0.30\textwidth}
		\centering
		\includegraphics[width=\textwidth]{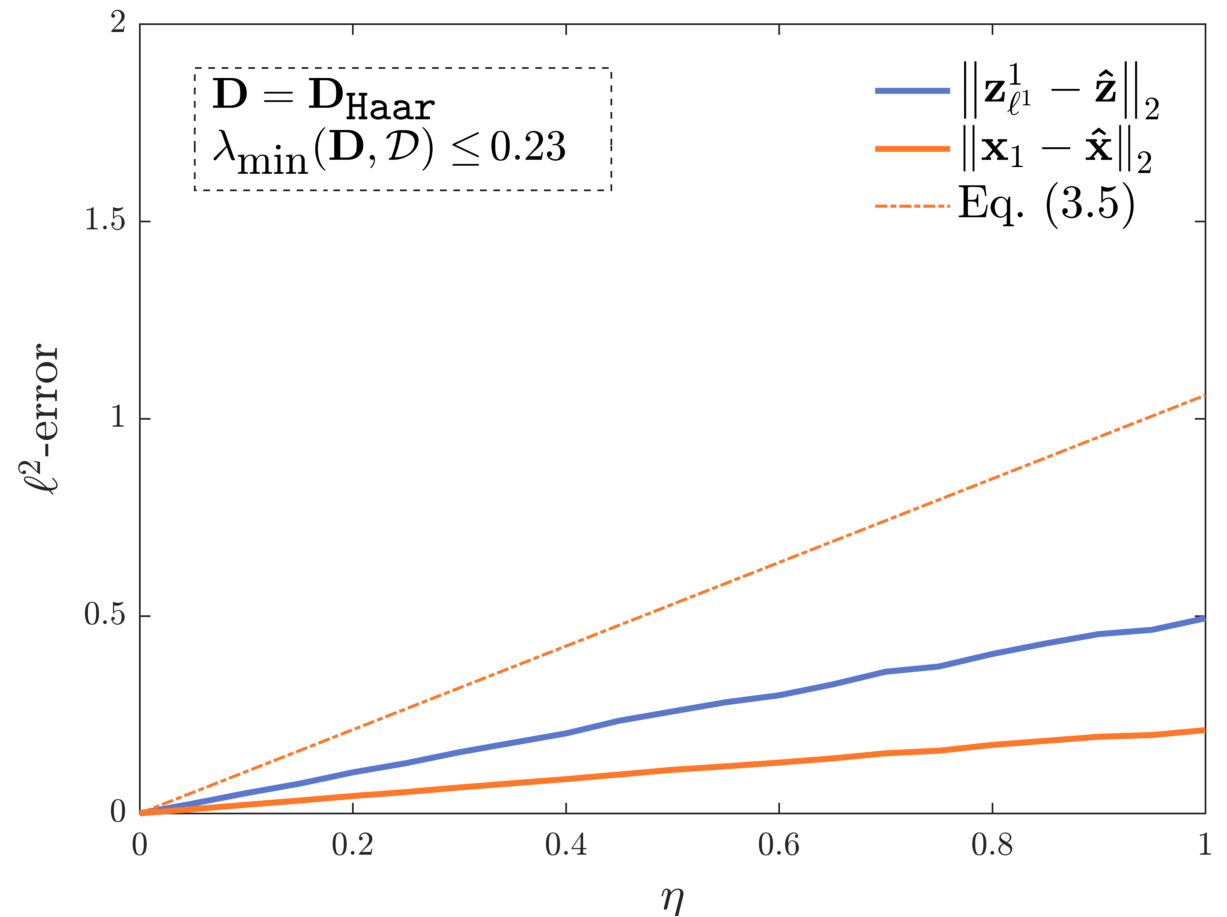}
		\caption{}
		\label{fig:noise:1}
	\end{subfigure}%
	\qquad
	\begin{subfigure}[t]{0.30\textwidth}
		\centering
		\includegraphics[width=\textwidth]{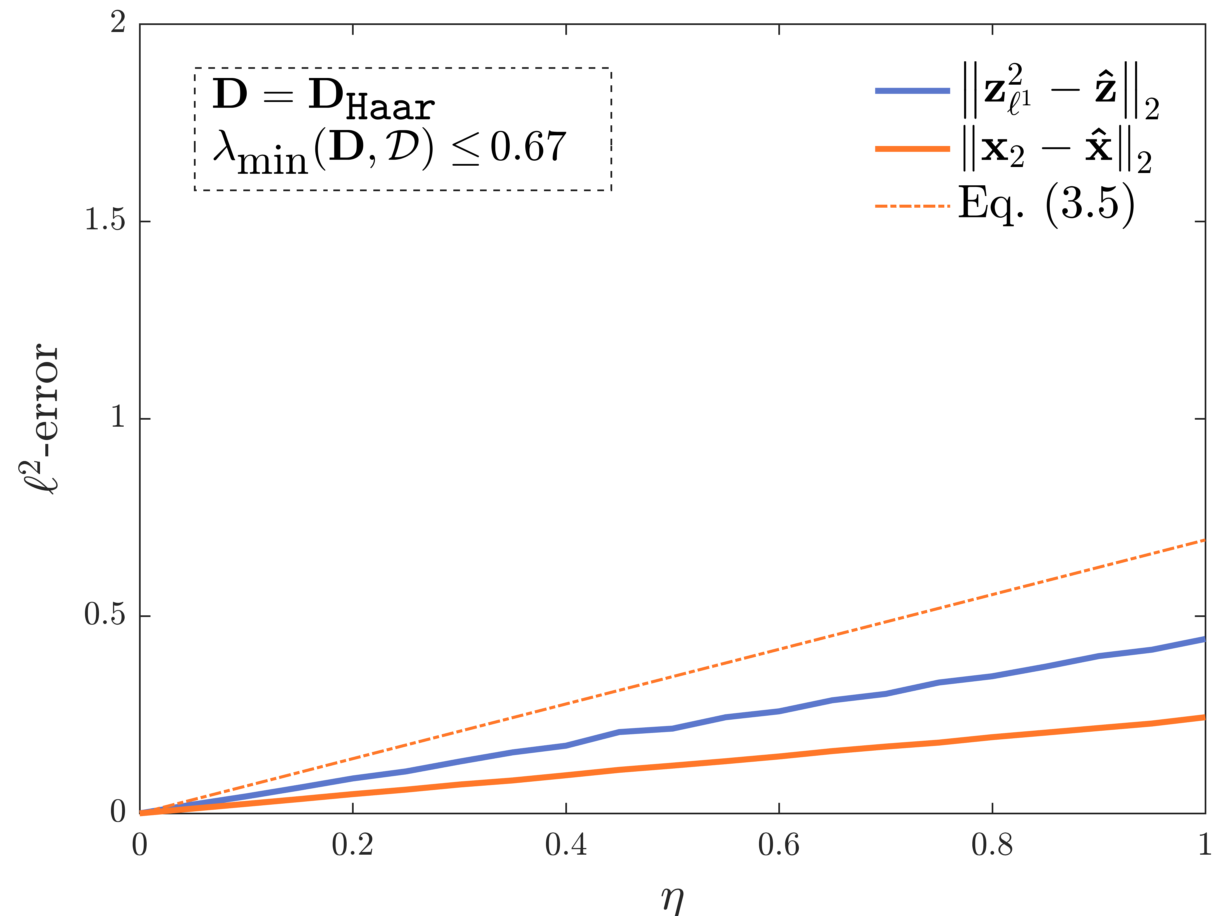}
		\caption{}
		\label{fig:noise:2}
	\end{subfigure}%
	\qquad
	\begin{subfigure}[t]{0.30\textwidth}
		\centering
		\includegraphics[width=\textwidth]{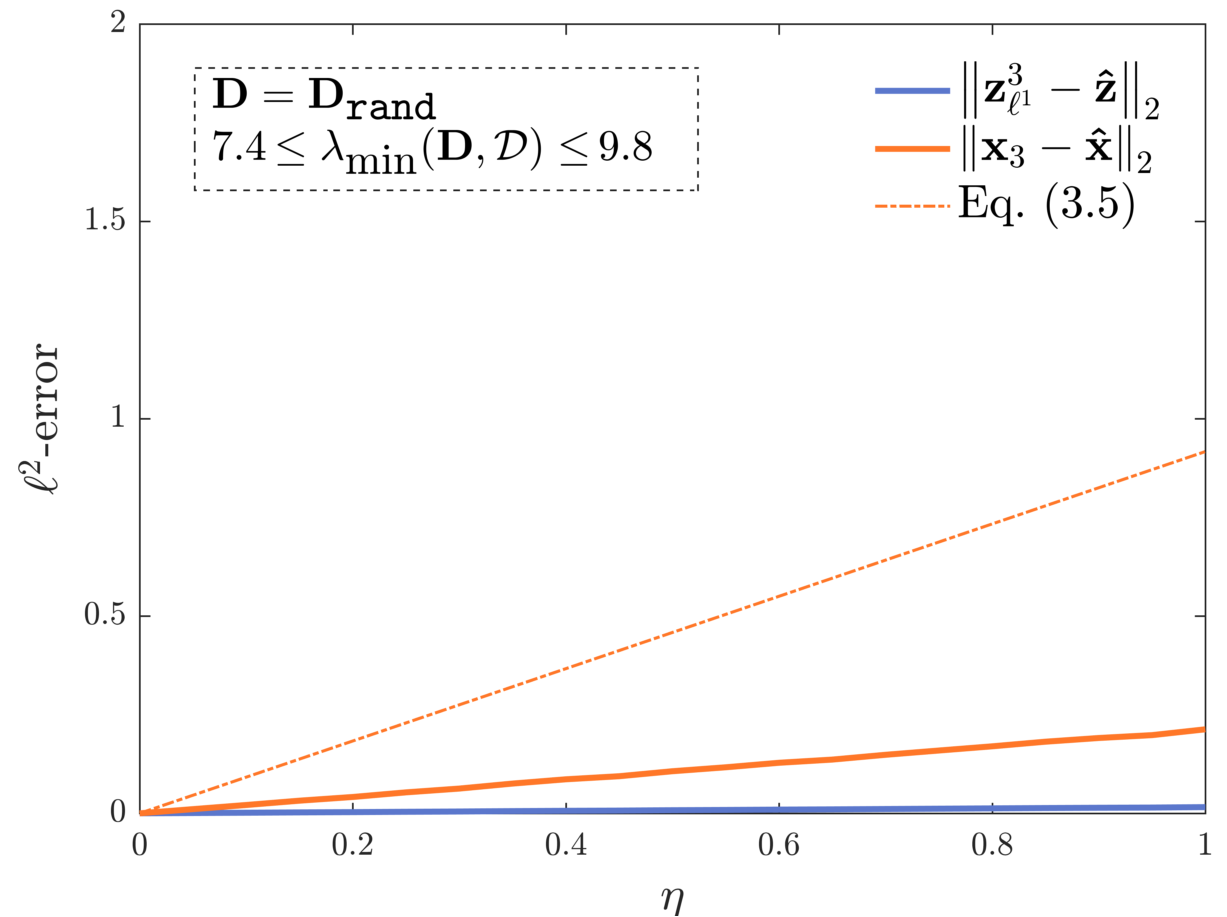}
		\caption{}
		\label{fig:noise:3}
	\end{subfigure}
    \vspace{.5\baselineskip}
    \begin{subfigure}[t]{0.30\textwidth}
		\centering
		\includegraphics[width=\textwidth]{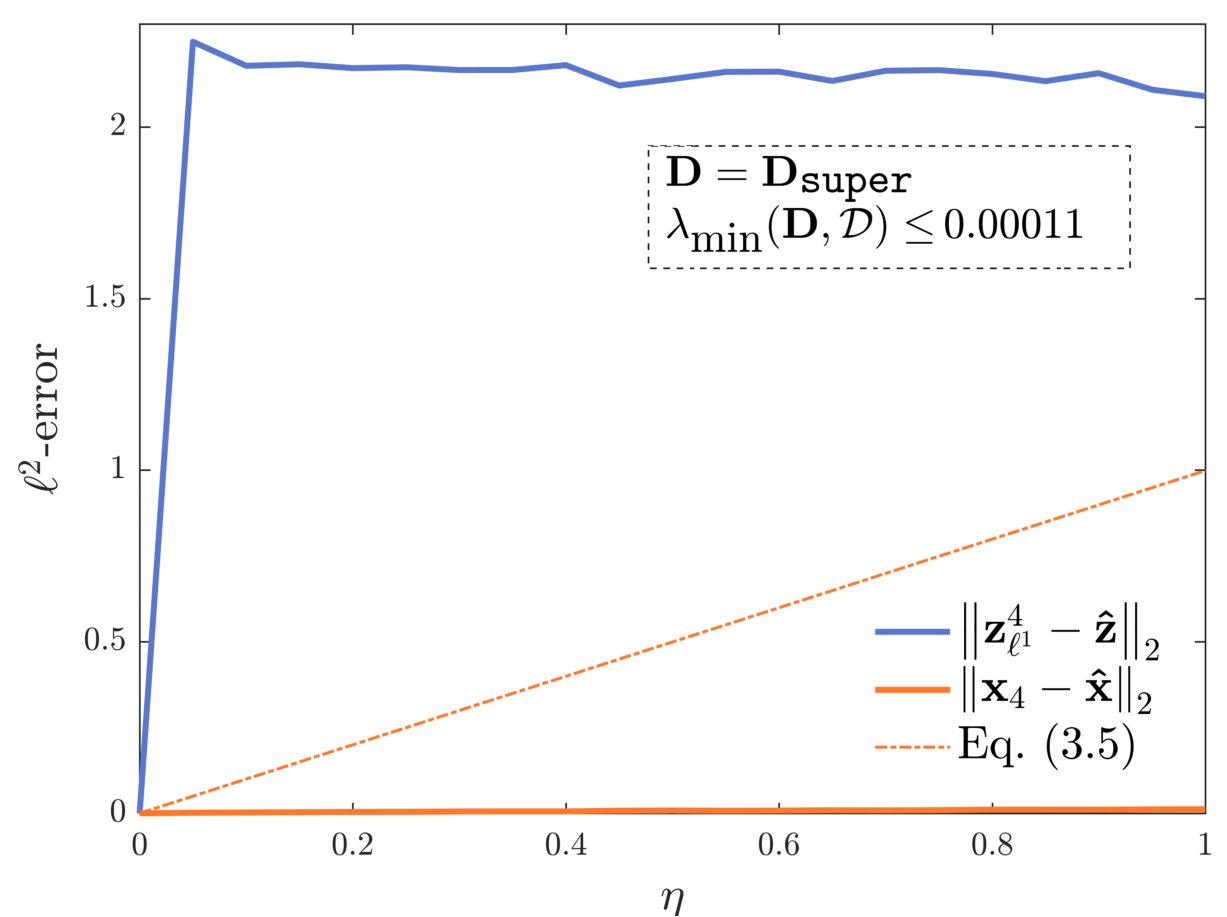}
		\caption{}
		\label{fig:noise:4}
	\end{subfigure}%
	\qquad
    \begin{subfigure}[t]{0.30\textwidth}
		\centering
		\includegraphics[width=\textwidth]{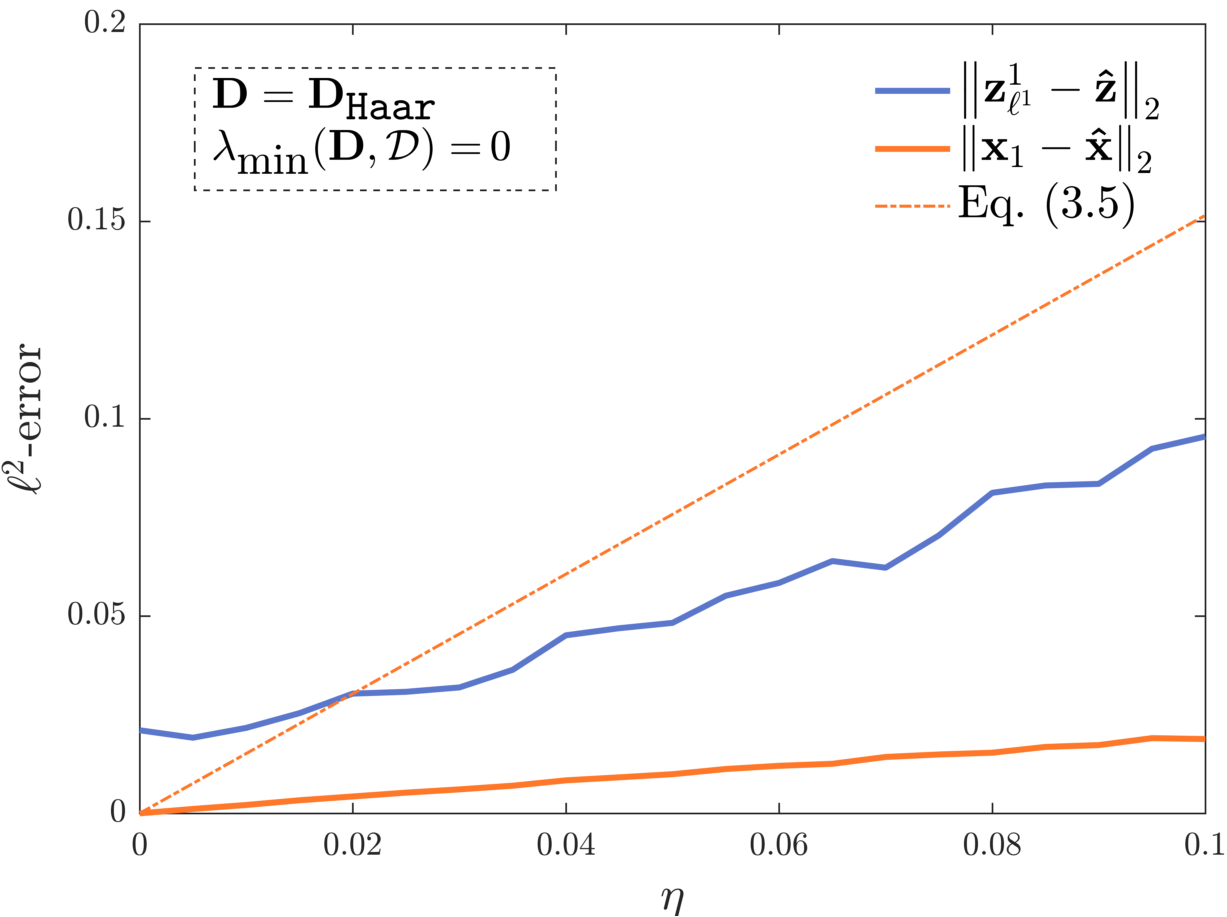}
		\caption{}
		\label{fig:noise:5}
	\end{subfigure}%
	\qquad
	\begin{subfigure}[t]{0.30\textwidth}
		\centering
		\includegraphics[width=\textwidth]{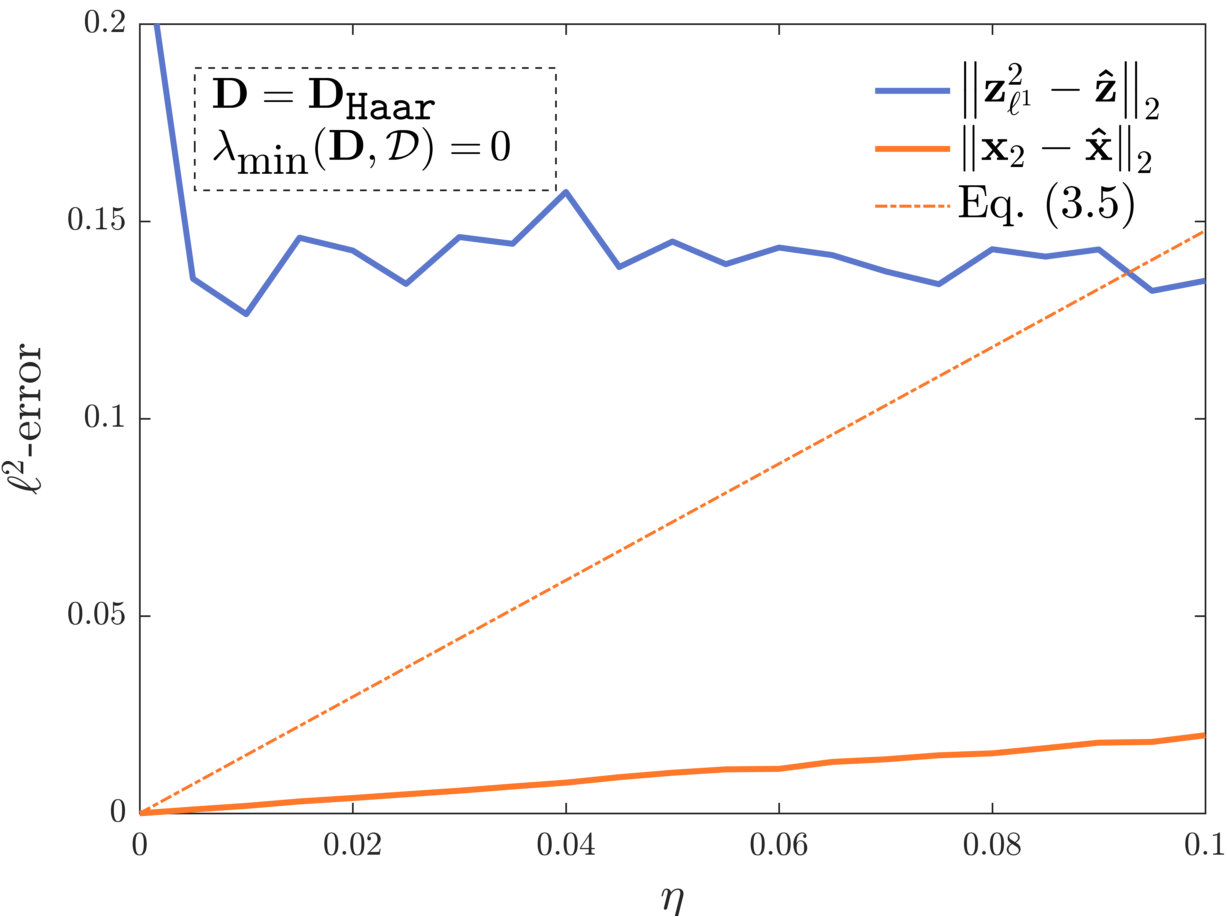}
		\caption{}
		\label{fig:noise:6}
	\end{subfigure}%
	\caption{\textbf{Robustness to measurement noise.} We display the reconstruction errors for a recovery from noisy measurements with an increasing noise level. The first four Subfigures are based on the examples for coefficient recovery of Section~\ref{sec:num_coef}, and the the last two on the  examples based on Haar wavelet of Section~\ref{sec:sampling_sig}. We use the notation $\mathcal{D} = \ds{\Onenorm,\zl^i}$, where $\zl^{\smash{i}} \in \Zlset$.}
	\label{fig:noise}
\end{figure}

\section*{Acknowledgments}
{\smaller Acknowledgements:
M.M. is supported by the DFG Priority Programme DFG-SPP 1798. He wishes to thank Martin Genzel for inspiring discussions. C.B. have been supported by a PEPS « Jeunes chercheuses et jeunes chercheurs » funding in 2017 and 2018 for this work.
P.W. and J.K. are supported by the ANR JCJC OMS.}

\renewcommand*{\bibfont}{\smaller}
\printbibliography[heading=bibintoc]

\appendix


\section{Proofs of Section~\ref{sec:coef_sig}}

\subsection{Proof of Lemma \ref{lem:observation} (Minimal $\ell^1$-Representers)}
\label{sec:proof_lem_observation}
\begin{enumerate}

\item[(b)]  ``$\Rightarrow$'': Assume that $\X = \left\{ \gt \right\}$. First, we show that $\Zlset \subseteq \Zset$: Let $\zl \in \Zlset$. Since $\Dict \zl = \gt$, $\zl$ is in the feasible set of \coefNoisefree. Furthermore, due to $\X = \left\{ \gt \right\}$ it holds true that $\Dict \solz  = \gt$  for each $\solz \in \Zset$. Thus also $\norm{\zl}_1 \leq \norm{\solz}_1$ and therefore $\zl \in \Zset$.  Secondly, we show $\Zset \subseteq \Zlset$. For a $\solz \in \Zset$ it holds true that $D\solz = \gt$ since we have assumed that $\X = \left\{ \gt \right\}$. Thus, $\solz$ is also feasible for $\Zlset$. Since each $\zl \in \Zlset$ is in turn feasible for $\Zset$, we obtain that $\norm{\solz}_1 \leq \norm{\zl}_1$  and therefore $\solz \in \Zlset$.  \newline
``$\Leftarrow$'': If $\Zset = \Zlset$, then $\X = \Dict \cdot \Zset = \Dict \cdot \Zlset = \left\{\gt\right\}$.


\item[(a)] If $\Zset = \left\{ \gtz \right\}$, then also $\X = \Dict \cdot \Zset = \left\{ \gt \right\}$ and (b) would imply $\Zlset = \Zset = \left\{ \gtz \right\}$.


\end{enumerate}

\subsection{Proof of Lemma \ref{lem:gauge_formulation} (The Gauge Formulation)}
\label{sec:proof_gauge_formulation}
By definition,
\begin{align}
  \X & = \Dict \cdot \left( \argmin_{\z \in \R^d} \norm{\z}_1 \quad \mbox{ s.t } \quad  {\norm{\y - \Meas  \Dict  \z}_2 \leq \noiseparam} \right)  \\
  & = \Dict  \cdot \left( \argmin_{\z \in \R^d}  \inf \left\{ \lambda >0 : \z \in \lambda\cdot \Bn[1]{d}  \right\} 
 \quad \mbox{ s.t. } \quad \norm{\y - \Meas  \Dict  \z}_2 \leq \noiseparam \right) \\
 & = \argmin_{\x \in \R^n} \inf \left\{ \lambda >0 : \x \in \lambda \cdot \Dict \cdot\Bn[1]{d}  \right\} \quad \mbox{ s.t. } \quad \norm{\y - \Meas  \x}_2 \leq \noiseparam. 
\end{align}

\subsection{Proof of Lemma \ref{lem:dc} (Descent Cone of the Gauge)}
\label{sec:proof_lem_dc}
We will only prove the first equality and note that the other one follows essentially the same argumentation. Pick any $\zl \in \Zlset$ and  note that $p_{\Dict\cdot\Bn[1]{d}} (\gt) = \norm{\zl}_1$.\\
 ``$\supseteq$'': Let $\h \in \dc{\norm{\cdot}_1,\zl}$, i.e., there exists a $\tau>0$ such that $\norm{\zl + \tau \h}_1 \leq \norm{\zl}_1$. Hence, 
 \[
  p_{\Dict\cdot\Bn[1]{d}} (\Dict  \zl + \tau  \Dict  \h )  = p_{\Dict\cdot\Bn[1]{d}} (\Dict \cdot ( \zl + \tau  \h )) \leq \norm{\zl + \tau \cdot \h}_1 \leq \norm{\zl}_1 = p_{\Dict\cdot\Bn[1]{d}} (\gt),
 \]
 and therefore $\Dict  \h \in \dc{p_{\Dict\cdot\Bn[1]{d}},\gt}$. \\
 ``$\subseteq$'': Let $\x \in \dc{p_{\Dict\cdot\Bn[1]{d}},\gt}$, i.e., there exists $\tau >0$ such that 
 \[
 R \coloneqq p_{\Dict\cdot\Bn[1]{d}} (\gt + \tau  \x) \leq p_{\Dict\cdot\Bn[1]{d}} (\gt) = \norm{\zl}_1.
 \]
 Now, choose $\h \in \Bn[1]{d}$ such that $R \cdot \Dict  \h = \gt + \tau  \x$ and write $\x = \Dict \cdot (R/\tau\cdot \h - 1/\tau \cdot \zl) \eqqcolon \Dict (\zz)$. Observe that 
 \[
  \norm{\zl + \tau \zz}_1 = R \cdot \norm{\h}_1 \leq R \leq \norm{\zl}_1,
 \]
 and therefore $\zz \in \dc{\norm{\cdot}_1,\zl}$.
 
 \begin{remark}
  \label{rem:gauge}
  The proof of ``$\subseteq$'' shows that $\zl$ could be replaced by any other $\gtz$ with $\gt = \Dict\gtz$, which is not necessarily a minimal $\ell^1$-representer of $\gt$. Hence, $\dc{p_{\Dict\cdot\Bn[1]{d}},\gt} \subseteq \Dict \cdot \dc{\norm{\cdot}_1,\gtz}$ and $\ds{p_{\Dict\cdot\Bn[1]{d}},\gt} \subseteq \Dict \cdot \ds{\norm{\cdot}_1,\gtz}$ for any $\gtz$ with $\gt = \Dict\gtz$, with inequality if $\gtz \in \Zlset$.
 \end{remark}

\subsection{Proof of Theorem \ref{thm:coeff} (Coefficient Recovery)}
\label{sec:proof_coeff}

Recalling Proposition \ref{prop:recover}, the goal of the proof is to find a lower bound for the minimum conic singular value $\lmin{\MeasDict}{C} = \inf \left\{\norm{\MeasDict \z}_2 : \z \in C \cap \Sn{d-1} \right\}$, where we use the abbreviated notation $C \coloneqq \dc{\Onenorm;\zl}$ and $\MeasDict \coloneqq \Meas\Dict$. 
Note that by assumption $\norm{\Dict \z}_2 \geq \lmin{\Dict}{C} >0$ for all $\z \in C \cap \Sn{d-1}$.  Thus, we obtain
\begin{align}
 \lmin{\MeasDict}{C} & = \inf \left\{  \frac{\norm{\Meas \Dict \z}_2}{\norm{\Dict \z}_2} \norm{\Dict \z}_2 : \z \in  C \cap \Sn{d-1} \right\} \\
 & \geq \lmin{\Dict}{C} \cdot \inf \{  \norm{\Meas \x}_2 : \x \in \Dict C \cap \Sn{n-1}\} \label{eq:esti}.
\end{align}
Theorem \ref{thm:versh} now implies that there is a numerical constant $\constant>0$ such that with probability at least $1-\e^{-u^2/2}$, we have
\begin{equation}
 \inf \{  \norm{\Meas \x}_2 : \x \in \Dict C \cap \Sn{n-1} \} > \sqrt{m-1} - \constant \cdot \gaussparam^2 \cdot (\cmw{\Dict C} + u).
\end{equation}
Thus, with probability at least $1-\e^{-u^2/2}$, we conclude from the previous steps that
\begin{equation}
\lmin{\MeasDict}{C} > \lmin{\Dict}{C} \cdot \left(\sqrt{m-1} - \sqrt{m_0 - 1} \right).
\end{equation}
The claim of the theorem is then a direct consequence of Proposition \ref{prop:recover}.

\begin{remark}
 The above argumentation may be compared with \cite[Theorem 3.1]{chen2017}. We also show that if $\Dict$ is bounded away from 0 on the intersection $S$ of a closed convex cone $C$ and the sphere, then $\Meas\Dict$ also stays away from 0 on $S$ with high probability. However, an important difference is that our result does not involve $\lmin{\Dict}{C}^{-1}$ as a multiplicative factor in the rate $m_0 \approx \cmw[2]{\Dict\cdot C}$. It therefore allows for a tight description of the sampling rate in the case of noiseless Gaussian measurements; cf.~the discussion subsequent to Theorem~\ref{thm:versh}. Indeed, the numerical experiments of Section~\ref{sec:num_coef} reveal that $\lmin{\Dict}{\dc{\Onenorm;\zl}}$ may be very small, while \coefNoisefree\ still allows for exact recovery from $m\approx\cmw[2]{\Dict\cdot\ds{\Onenorm;\zl} }$ measurements.
\end{remark}

\subsection{Proof of Proposition~\ref{prop:epsilongaussian} (Stable Recovery)}
\label{sec:proof_stable}

Let $\za\in \R^d$ be chosen according to \eqref{eq:suro}, i.e., it satisfies $\norm{\gt - \Dict \za}_2 \leq \varepsilon$ and $\norm{\za}_1 = \norm{\gtz}_1$, where $\gtz$ is any vector with $\gt = \Dict \gtz$. The goal is to invoke \cite[Theorem 6.4]{genzel2017} with
\begin{equation}
  t = \max \left\{ r \cdot \norm{\gt - \Dict \za}_2,\frac{2\noiseparam}{\sqrt{m-1} - c\cdot\gamma^2 \cdot \left( \tfrac{r+1}{r} \cdot (\cmw{\Dict \cdot \ds{\Onenorm;\za}} +1) + u\right) } \right\}.
 \end{equation}
 Thus, we need to verify that $t$ satisfies
 \begin{equation}
  t \geq \frac{2\noiseparam}{\sqrt{m-1} - c\cdot\gamma^2 \cdot \left( w_t(\ds{p_{\Dict\cdot\Bn[1]{d}};\gt}) + u \right)},
 \end{equation}
 where $w_t$ denotes the \emph{local mean width} at scale $t>0$; see \cite[Definition 6.1]{genzel2017} for details on this notion.
To that end, we first observe that
 \begin{align}
  w_t(\ds{p_{\Dict\cdot\Bn[1]{d}};\gt}) & \leq w_t (\Dict \cdot \ds{\Onenorm;\gtz}  + \gt - \gt ) \\
  & \stackrel{\mathllap{t\geq r\cdot\norm{\gt - \Dict\za}_2}}{\leq} w_{r\cdot\norm{\gt - \Dict\za}_2} ( (\Dict \cdot \ds{\Onenorm;\gtz} + \gt) - \gt),
 \end{align}
 where we have used that $\ds{p_{\Dict\cdot\Bn[1]{d}};\gt}  \subseteq \Dict \cdot \ds{\Onenorm;\gtz}$ for any $\gtz$ with $\gt =\Dict \gtz$ in the first step (see Remark~\ref{rem:gauge}), and the monotonicity of the local mean width in the second step. Observe that we have $\Dict \za \in  \Dict \cdot \ds{\Onenorm;\gtz} + \gt$, due to the assumption $\norm{\za}_1 = \norm{\gtz}_1$. Hence, we can make use of \cite[Lemma A.2]{genzel2017} for $K=\Dict \cdot \ds{\Onenorm;\gtz} + \gt$ in order to obtain
 \begin{align}
  w_{r\cdot\norm{\gt - \Dict\za}_2} ( (\Dict \cdot \ds{\Onenorm;\gtz} + \gt) - \gt) & \leq \frac{r+1}{r} \cdot \left( \cmw{(\Dict \cdot \ds{\Onenorm;\gtz} + \gt) - \Dict\za }   + 1 \right) \\
  & = \frac{r+1}{r} \cdot (\cmw{\Dict \cdot \ds{\Onenorm;\za}} +1),
 \end{align}
where the equality follows from:
 \begin{align}
  \Dict \cdot \ds{\Onenorm;\gtz} + \gt  - \Dict \za  & = \Dict \cdot \{\vec{h} \in \R^d : \norm{\gtz + \vec{h}}_1 \leq \underbrace{\norm{\gtz}_1}_{= \norm{\za}_1} \} + \gt - \Dict \za \\ 
  & \stackrel{\mathllap{\vec{h}' = \vec{h} + \gtz - \za}}{=} \Dict \cdot \{\vec{h}' \in \R^d : \norm{\za + \vec{h}'}_1 \leq  \norm{\za}_1 \} = \Dict \cdot  \ds{\Onenorm;\za}.
 \end{align}
We conclude that 
 \begin{align}
  t &\geq \frac{2\noiseparam}{\sqrt{m-1} - c\cdot\gamma^2 \cdot \left( \tfrac{r+1}{r} \cdot (\cmw{\Dict \cdot \ds{\Onenorm;\za}} +1) + u\right)} \\
  & \geq \frac{2\noiseparam}{\sqrt{m-1} - c\cdot\gamma^2\cdot \left(w_t(\ds{p_{\Dict\cdot\Bn[1]{d}};\gt})  + u \right) }.
 \end{align}
 Hence, \cite[Theorem 6.4]{genzel2017} then implies that any minimizer of \eqref{eq:sig} satisfies $\norm{\gt - \sol}_2 \leq t$.\footnote{This argument does actually not cover the case of $\noiseparam = 0$, 
 but here we can simply use that $t = r \cdot \norm{\gt -\Dict \za}_1 > 0$ if $\gt \neq \Dict \za$.} 

\section{Proof of Proposition \ref{prop:conditiongaussian} (Width and Condition Number)}
\label{sec:proof_conditiongaussian}

Let us start with a preliminary lemma, which generalizes Proposition 10.2 in \cite{amelunxen2014edge}.
\begin{lemma}
\label{lem:statdim}
 For a closed convex cone $C\subseteq \R^d$, a dictionary $\Dict \in \R^{n\times d}$ and  a standard Gaussian vector  $\g \sim \mathcal{N} (\vec{0}, \I_n)$, we have that
 \begin{equation}
  \E \left[ \left( \sup_{\z \in C \cap \Bn[2]{d}} \sp{\g}{\Dict\z} \right)^2\right] \leq \left( \E \left[ \sup_{\z \in C \cap \Sn{d-1}} \sp{\g}{\Dict\z} \right] \right)^2 + \lmax[2]{\Dict}{C}.
 \end{equation}
\end{lemma}

\begin{proof}
Define the random variable $Z=Z(\g):= \sup_{\z \in C \cap  \Sn{d-1}} \langle \g, \Dict \z \rangle$. 
In a first step, we prove that 
\begin{align}
\label{eq:lem_statdim_aux}
\E \left[\left( \sup_{\z \in C \cap \Bn[2]{d}} \langle \g, \Dict \z \rangle\right)^2 \right]
\leq \E \left[Z^2 \right].
\end{align}
Indeed, since $Z^2$ is a nonnegative random variable, we obtain
\begin{align*}
\E [Z^2] \geq \E \left[Z^2 \cdot \mathbf{1}_{\R^d\backslash C^{\circ}}(\Dict^*\g) \right] =  \E \left[\left(\sup_{\z \in C \cap  \Sn{d-1}} \langle \g, \Dict \z \rangle\right)^2 \cdot \mathbf{1}_{\R^d\backslash C^{\circ}}(\Dict^*\g) \right],
\end{align*}
where $C^{\circ}$ denotes the \emph{polar cone} of $C$.
Furthermore, it holds true that
\begin{align*}
\E \left[\left(\sup_{\z \in C \cap  \Sn{d-1}} \langle \g, \Dict \z \rangle\right)^2 \cdot \mathbf{1}_{\R^d\backslash C^{\circ}}(\Dict^*\g) \right] = \E \left[\left( \sup_{\z \in C \cap \Bn[2]{d}} \langle \g, \Dict \z \rangle\right)^2 \right].
\end{align*}
Indeed, for an $\x \in \R^n$ such that $\Dict^*\x \notin C^{\circ}$ the equality $\sup_{\z \in C \cap  \Sn{d-1}} \langle \x, \Dict \z \rangle = \sup_{\z \in C \cap \Bn[2]{d}} \langle \x, \Dict \z \rangle$ holds true, because the supremum over the ball occurs at a vector of length 1. On the other hand, when $\Dict^*\x \in C^{\circ}$, one has $\sup_{\z \in C \cap \Bn[2]{d}} \langle \x, \Dict \z  \rangle = 0$. Therefore, \eqref{eq:lem_statdim_aux} is established.

Moreover, observe that the function $\g \mapsto Z(\g)$ is $\lambda_{\max}(\Dict,C)$-Lipschitz. Indeed, for $\f,\g\in \R^n$ and $\z \in C \cap  \Sn{d-1}$ we obtain that
\begin{align*}
\langle \g, \Dict \z \rangle & = \langle \f, \Dict \z \rangle + \langle \g, \Dict \z \rangle - \langle
\f,\Dict \z \rangle  \leq \langle \f, \Dict \z \rangle + \norm{\f-\g}_2 \norm{\Dict \z}_2 \\
& \leq \langle \f, \Dict \z \rangle + \lambda_{\max}(\Dict,C) \norm{\f-\g}_2,
\end{align*}
and therefore by taking the supremum
\begin{align*}
\sup_{\z \in C \cap  \Sn{d-1}} \langle \g, \Dict \z \rangle \leq \sup_{\z \in C \cap  \Sn{d-1}}  \langle \f, \Dict \z \rangle + \lambda_{\max}(\Dict,C) \norm{\f-\g}_2.
\end{align*}
By swapping the roles of $\f$ and $\g$, an analogue estimate can be obtained, which verifies the claimed Lipschitz continuity. Thus, the fluctuation of $Z$ can be bounded as follows:
\begin{align*}
\E \left[ Z^2\right] - \E\left[Z\right]^2 = \E \left[(Z - \E Z)^2 \right] = \text{Var}(Z) \leq \lambda^2_{\max}(\Dict,C), 
\end{align*}
where the last estimate follows from Fact C.3 in \cite{amelunxen2014edge}.
\end{proof}

\paragraph{Back to the proof of Proposition \ref{prop:conditiongaussian}.}
In order to prove Proposition \ref{prop:conditiongaussian}, we continue as follows:
First observe that 
\begin{equation}
\cmw[2]{\Dict \cdot C} = w^2(\Dict \cdot C \cap \Sn{n-1} ) \leq \delta (\Dict \cdot C) = \E \left[\left( \sup_{\x \in \Dict \cdot C \cap \Bn[2]{n} } \sp{\g}{\x} \right)^2 \right],
\end{equation}
where $\delta$ denotes the \emph{statistical dimension}\footnote{The statistical dimension of a convex cone $C \subseteq \R^n$ can be defined as $\delta (C) = \E [(\sup_{\x \in C \cap \Bn[2]{n}} \sp{\g}{\x})^2]$; see \cite[Prop.~3.1]{amelunxen2014edge} for details. It holds true that $\cmw[2]{C} \leq \delta (C) \leq \cmw[2]{C} +1$, which is why both notions are often interchangeable~\cite[Prop.~10.2]{amelunxen2014edge}.}.
Next, it is straightforward to see that 
\begin{equation}
\label{eq:set}
\Dict \cdot C \cap \Bn[2]{n} \subseteq \frac{\Dict}{\lmin{\Dict}{C}} \cdot \left( C \cap \Bn[2]{d} \right),
\end{equation}
which immediately implies that
\begin{equation}
 \delta (\Dict \cdot C)   \leq \frac{1}{\lmin[2]{\Dict}{C}} \cdot \E \left[\left( \sup_{\x \in \Dict\left( C \cap \Bn[2]{d} \right)} \sp{\g}{\x} \right)^2 \right].
\end{equation}
Exercise~7.5.4 in \cite{vershynin2018}  and Lemma \ref{lem:statdim} now allow us to derive the desired bound:
\begin{align}
 \cmw[2]{\Dict \cdot C}  \leq \delta (\Dict \cdot C) & \leq \frac{1}{\lmin[2]{\Dict}{C}} \E \left[\left( \sup_{\x \in \Dict\left( C \cap \Bn[2]{d} \right)} \sp{\g}{\x} \right)^2 \right] \\
 & \stackrel{(\ref{lem:statdim})}{\leq}  \frac{1}{\lmin[2]{\Dict}{C}} \left( \left( \E \left[ \sup_{\x \in \Dict\left( C \cap \Sn{d-1} \right)} \sp{\g}{\x} \right] \right)^2 + \lmax[2]{\Dict}{C} \right) \\
 & \stackrel{(7.5.4)}{\leq} \frac{\norm{\Dict}_2^2}{\lmin[2]{\Dict}{C}} \left( \left( \E \left[ \sup_{\z \in  C \cap \Sn{d-1} } \sp{\g}{\z} \right] \right)^2 + 1 \right) \\
 & = \condi[2]{\Dict}{C} \cdot \left( \cmw[2]{C} + 1\right).
 \end{align}

\section{Proofs of Section \ref{sec:circumangle}}

\subsection{Proof of Proposition \ref{prop:simple_circumcenter} (Circumangle of Polyhedral Cones)}
\label{sec:simple_circumcenter}

Consider a nontrivial pointed polyhedral cone $C=\mathrm{cone}(\x_1,\hdots,\x_k)$ with $\norm{\x_i}_2=1$ for $i\in [k]$ and let $\alpha$ denote its circumangle. Let $\Xb \coloneqq [\x_1,\hdots,\x_k]\in \R^{n\times k}$. Since $C$ does not contain a line, its circumcenter $\ax$ is unique, belongs to the cone and $0\leq \alpha < \pi/2$, see \cite{Henrion10}. This implies that $\Xb^*\ax > \vec{0}$ (element-wise). We have:
\begin{align*}
\cos(\alpha)& = \sup_{\vec{v}\in \Sn{n-1}} \inf_{\x\in C\cap \Sn{n-1}} \langle \x, \vec{v}\rangle=\sup_{\vec{v}\in \Sn{n-1}}  \inf_{\cc\geq \vec{0}, \norm{\Xb\cc}_2^2=1} \langle \Xb\cc, \vec{v}\rangle \\
&\leq \sup_{\vec{v}\in \Sn{n-1}} \inf_{i \in [k]} \langle \vec{e}_i, \Xb^*\vec{v}\rangle=\sup_{\vec{v}\in \Sn{n-1}} \min \Xb^*\vec{v},
\end{align*}
where $\vec{e}_i$ denotes $i$-th standard basis vector in $\R^k$ and we have used an inclusion of sets argument in the inequality. We now argue that the inequality is in fact an equality. To that end, observe that for any $\vec{v} \in \Sn{n-1}$ such that $\Xb^*\vec{v}\geq 0$, we also have that
\begin{align*}
 \inf_{\cc\geq \vec{0}, \norm{\Xb\cc}_2^2=1} \langle \cc, \Xb^*\vec{v}\rangle \geq \inf_{\cc\geq \vec{0}, \norm{\Xb\cc}_2^2\geq 1} \langle \cc, \Xb^*\vec{v}\rangle \geq \inf_{\cc\geq \vec{0}, \langle \one , \cc \rangle \geq 1} \langle \cc, \Xb^*\vec{v}\rangle = \min \Xb^*\vec{v}.
\end{align*}
In this sequence of inequalities, we first used the inclusion of sets, the triangular inequality together with the inclusion of sets and finally the fact that a linear program attains its minimum (if it exists) on an extremal point $\mathrm{Ext}\left(\{\cc\geq 0, \langle \one , \cc \rangle \geq 1\}\right)=\{\eb_1,\hdots,\eb_k\}$. Note that the condition $\Xb^*\vec{v}\geq 0$ ensures the existence of a solution.

Finally the infimum over $\Sn{n-1}$ can be relaxed to $\Bn[2]{n}$, since the supremum is attained on the boundary of the domain. This concludes the proof.

\subsection{Proof of Proposition \ref{prop:circ} (Maximal Width of Polyhedral Cones)}
\label{sec:proof_circ}

The proof of Proposition~\ref{prop:circ} necessitates a basic preliminary result on the Gaussian width of general convex polytopes, which we will proof first.

\paragraph{Bounding the Gaussian Width of Convex Polytopes}

\begin{lemma}
    \label{lem:convexpolywidth}
    Let $K$ be a convex polytope with $k\geq 5$ vertices that is contained in the unit ball of $\R^n$. Then:
    \begin{align*}
         w(K) & \leq \sqrt{2 \log\left(\frac{k}{\sqrt{2\pi}}\right)}   + \frac1{\sqrt{2 \log\left(\frac{k}{\sqrt{2\pi}}\right)}}.
    \end{align*}
\end{lemma}

\begin{proof}
    Denote by $\mathcal{V}(K)$ the set of vertices of $K$.
Since the maximum of the scalar products with points of a convex set is attained on a vertex, we get the following union bound for $x>0$ and $\g \sim \mathcal{N}(\vec{0},\Id)$:
    \begin{align}
        \Prob\left[\sup_{\ub \in K}\langle \ub , \g \rangle \geq x \right] & \leq  \sum_{\vb\in \mathcal{V}(K)}   \Prob\left[\langle \vb , \g \rangle \geq x \right] \notag \\
        & \leq \sum_{\vb \in \mathcal{V}(K)} \tfrac{1}{\sqrt{2\pi}\left\lVert \vb \right\rVert_2} \int_{x}^\infty \exp \left( - \frac{y^2}{2 \left\lVert \vb \right\rVert_2^2}   \right) \mathrm{d}y \\
        & \leq \tfrac{k}{\sqrt{2\pi}} \cdot \int_{x}^\infty \exp \left( - \tfrac{y^2}{2}   \right) \mathrm{d}y. \label{difalphagen}
    \end{align}
    Recall the following standard bound on the tail probability of a Gaussian, which will be used in the remainder of the proof:
    \begin{align*}
        \int_{x}^{\infty} \exp (-y^2/2) \mathrm{d}y & \leq \frac{\exp (-x^2/2)}{x}.
    \end{align*}
    Let $x_0 \coloneqq \sqrt{2 \log\left(\frac{k}{\sqrt{2\pi}}\right)} > 1$ and note that $\frac{k}{\sqrt{2\pi}} \exp \left( - \frac{x_0^2}{2} \right) = 1$. We may now use the bound \eqref{difalphagen} in order to obtain:
    \begin{align*}
    w(K)     & \leq \int_{\R_+}  \Prob\left[\sup_{\ub \in K}\langle \ub ,\g \rangle \geq x \right] \mathrm{d}x \\
               & \leq \int_{\R_+} \min \left\{1, \frac{k}{\sqrt{2\pi}} \int_{x}^\infty \exp \left( - \frac{y^2}{2} \right) \mathrm{d}y \right\}  \mathrm{d}x \\
               & \leq x_0 + \frac{k}{\sqrt{2\pi}}  \int_{x_0}^{\infty}   \exp \left( - \frac{x^2}{2} \right) \mathrm{d}x \\
               & \leq x_0 + \frac{1}{x_0}.
    \end{align*}
\end{proof}

\paragraph{Back to the proof of Proposition \ref{prop:circ}}
Equipped with this lemma, we can now prove Proposition \ref{prop:circ}.

 Let $C = \cone{\x_1,\dots,\x_k} \subset \R^n$ be a $k$-polyhedral $\alpha$-cone and let $\ax \in \Sn{n-1}$ be an axis vector such that $C\subseteq C(\alpha,\ax)$. Without loss of generality assume that $\norm{\x_i}_2 = 1$ for $i\in [k]$. Define the affine hyperplane $\mathcal{H} \coloneqq \left\{ \h \in \R^n : \sp{\h}{\ax} = 1 \right\}$ and let $K \coloneqq C \cap \mathcal{H}$. Observe that $K$ is a convex polyhedron with vertices belonging to the set $\left\{ \x_i / \sp{\x_i}{\ax} : i \in [k] \right\} \subseteq \Bn[2]{n}(\cos (\alpha)^{-1})$. Since $\norm{\vec{k}}_2\geq 1$ for all $\vec{k} \in K$, any $\h \in C \cap \Sn{n-1}$ can be written as $\h = \lambda (\kk) \cdot \kk$, where $\kk \in K$ and $0 < \lambda (\kk) \leq 1$.
  Hence, for all $\g \in \R^n$ it holds true that
 \begin{equation}
  \sup_{\h \in \Sn{n-1} \cap C} \sp{\h}{\g} = \sup_{\kk \in K} \left( \lambda (\kk) \cdot \sp{ \kk}{\g} \right) \leq \max \left\{0, \sup_{\kk \in K} \sp{\kk}{\g} \right\}.
 \end{equation}
 Let $\proj^{\ax}, \proj_{\perp}^{\ax}$ denote the orthogonal projections onto $\spann{(\ax)}$ and $\spann{(\ax)}^\perp$, respectively. Observe that for $\kk \in K$ and $\g \in \R^n$ it holds true that
 \begin{equation}
  \sp{\kk}{\g} = \sp{\ax}{\proj^{\ax}(\g)} + \sp{\proj_{\perp}^{\ax} (\kk)}{\proj_{\perp}^{\ax} (\g) } \leq \max \left\{0,\sp{\ax}{\proj^{\ax}(\g)} \right\} +  \sp{\proj_{\perp}^{\ax} (\kk)}{\proj_{\perp}^{\ax} (\g) },
 \end{equation} 
 where the first equality follows from $\proj^{\ax} (\kk) = \ax$. Furthermore, since $\ax \in C$, we have that $\vec{0} \in \proj_{\perp}^{\ax} (K)$ and hence for all $\g \in \R^n$, 
 \begin{equation}
  \sup_{\kk \in K} \sp{\proj_{\perp}^{\ax} (\kk)}{\proj_{\perp}^{\ax} (\g)} \geq 0.
 \end{equation}  
 This allows us to conclude that
 \begin{equation}
  \sup_{\h \in \Sn{n-1} \cap C} \sp{\h}{\g} \leq \max \left\{0, \sup_{\kk \in K} \sp{\kk}{\g} \right\} \leq \max \left\{0,\sp{\ax}{\proj^{\ax}(\g)} \right\} +  \sup_{\kk \in K} \sp{\proj_{\perp}^{\ax} (\kk)}{\proj_{\perp}^{\ax} (\g) }.
 \end{equation}
 Hence, we obtain
 \begin{align}
   \cmw{C} & = \E_{n} \left[ \sup_{\h \in \Sn{n-1} \cap C} \sp{\h}{\g} \right] \\
    & =  \E_{n} \left[ \max\left\{0,\sp{\ax}{\proj^{\ax}(\g)} \right\} + \sup_{\kk \in K} \sp{\proj_{\perp}^{\ax}(\kk)}{\proj_{\perp}^{\ax}(\g)} \right] \\
    & = \E_1 \left[ \max\left\{0,\sp{\ax}{\proj^{\ax}(\g)} \right\} \right] + \E_{n-1} \left[ \sup_{\kk \in \proj_{\perp}^{\ax}(K)} \sp{\kk}{\proj_{\perp}^{\ax}(\g)} \right] \\
    & = \frac{1}{\sqrt{2\pi}} + w (\proj_{\perp}^{\ax}(K)), \label{eq:ineqtemp}
 \end{align}
 where the last equality follows from $\sp{\ax}{\proj^{\ax}(\g) }\sim \mathcal{N} (0,1)$ and the fact that $\proj_{\perp}^{\ax} (\g)$ is an $(n-1)$-dimensional standard  Gaussian vector on $\spann{(\ax)}^\perp$.

    Since $K\subset \mathcal{H} \cap \Bn[2]{n}(\cos (\alpha)^{-1})$, its $(n-1)$-dimensional projection satisfies $ \proj_{\perp}^{\ax}(K) \subset \Bn[2]{n-1}(\tan (\alpha))$. 
    Now, Lemma \ref{lem:convexpolywidth} yields the following bound on the Gaussian width of a polyhedron included in a ball of radius $\tan \alpha$ with at most $k\geq 5$ vertices:
\begin{equation}
  w (\proj_{\perp}^{\ax}(K)) \leq {\tan \alpha} \cdot \left( \sqrt{2\log (k/\sqrt{2\pi})} + 1/\sqrt{2\log (k/\sqrt{2\pi})} \right).
 \end{equation}
The claimed inequality of Proposition~\ref{prop:circ} is then just a consequence of \eqref{eq:ineqtemp}. 

\subsection{Proofs of Section~\ref{sec:geom_desc}}
\label{sec:proof_geom}

\paragraph{Descent Cone of $\ell^1$-Norm (Lemma~\ref{lem:dc_lone})}
We begin by showing a polyhedral description of the descent cone of the $\ell^1$-norm:

Let $\vec{v}$ be any vector such that  $\|\vec{v}\|_1=s$ and $\sign \vec{v}=\sign \z$. Note that $\vec{v}$ and $\z$ enjoy the same descent cone associated to the $\ell^1$-norm, which is easy to see by observing that 
\begin{equation}
\dc{\norm{\cdot}_1,\z} = \left\{\h \in \R^d : \sum_{i \in \Suppc}  \abs{h_i} \leq - \sum_{i \in \Supp} \sign (z_i) \cdot h_i  \right\}. 
\end{equation}
 Therefore, the descent set of $\|\cdot\|_1$ at $\vec{v}$ can be obtained by scaling up the cross-polytope by the factor $\|\vec{v}\|_1=s$ and shifting it by $-\vec{v}$, i.e., 
$$
 \ds{\norm{\cdot}_1,\vec{v}} = \convhull{\pm s \cdot \vec{e}_i - \vec{v} : i \in [d]}.
$$
We conclude by taking the conic hull of the previous set to obtain
$$
\dc{\norm{\cdot}_1,\z} = \dc{\norm{\cdot}_1,\vec{v}}  = \cone{\pm s \cdot \vec{e}_i - \vec{v} : i \in [d]}.
$$

\paragraph{Lineality of Descent Cone of $\ell^1$-Norm (Lemma~\ref{lem:lineality})}
Next, we describe the lineality space and lineality of $\dc{\norm{\cdot}_1,\z}$:

The lineality space of the descent cone at point $\z$ corresponds to the span of the face of the $\ell^1$-ball of minimal dimension containing $\z$.
It can therefore be defined as the span of the vectors joining $\z$ to the vertices of this face, which are exactly the vectors $\sign(z_i)\cdot\vec{e}_i$. 

For a more formal proof for this fact, one could argue as follows: First note that (see for instance Appendix B in \cite{amelunxen2014edge})
\begin{equation}
 \dc{\norm{\cdot}_1,\z}^\circ = \bigcup_{\tau\geq 0} \tau \cdot \partial \norm{\z}_1.
\end{equation}
Since $\partial \norm{\z}_1 = \left\{ \h \in \R^d : \h_{\Supp} = \sign (\z)_{\Supp}, \h_{\Suppc} \in [-1,1]^{d-s}\right\}$, it follows that the polar cone is closed, pointed (i.e., $\dc{\norm{\cdot}_1,\z}^\circ \cap - \dc{\norm{\cdot}_1,\z}^\circ = \left\{ \vec{0} \right\}$) and therefore finitely generated by its extreme rays
\begin{equation}
 \dc{\norm{\cdot}_1,\z}^\circ = \cone{\z^j \in \R^d : j \in [2^{d-s}]},
\end{equation}
where $\z^j_{\Supp} = \sign (\z)_{\Supp}$ and on $\Suppc$ all $2^{d-s}$ combinations $\z^j_{\Suppc} = \left\{-1,1\right\}^{d-s}$. 
Hence, we obtain the following polyhedral description for the descent cone
\begin{equation}
 \dc{\norm{\cdot}_1,\z} = \left\{ \h \in \R^d : \sp{\h}{\z^j} \leq 0 \mbox{ for all } j \in [2^{d-s}] \right\}.
\end{equation}
Using the matrix $\vec{B} \coloneqq \left[\z^1, \dots, \z^{2^{d-s}}\right]^T \in \R^{2^{d-s}\times d}$, the lineality space can then be conveniently expressed as $L_{\dc{\norm{\cdot}_1,\z}} = \ker (\vec{B})$.

On the other hand, observe that for any $\h \in L_{\dc{\norm{\cdot}_1,\z}}$, we can find $\tau>0$ such that $\norm{\z + \tau \cdot \h}_1 \leq \norm{z}_1$
and therefore (by choosing $\tau>0$ small enough)
\begin{equation}
 \sum_{j\in\Supp} \sign (z_j) \cdot (z_j + \tau \cdot h_j) + \sum_{i\in\Suppc} |h_i| \leq \sum_{j \in \Supp} |z_j|.
\end{equation}
Similarly, since also $-\h \in \dc{\norm{\cdot}_1,\z}$, we obtain (again by choosing a small enough $\tau>0$)
\begin{equation}
 \sum_{j\in\Supp} \sign (z_j) \cdot (z_j - \tau \cdot h_j) + \sum_{i\in\Suppc} |h_i| \leq \sum_{j \in \Supp} |z_j|.
\end{equation}
Adding up these two inequalities, we obtain that $\sum_{i\in\Suppc} |h_i| \leq 0$ and hence $h_i=0$ for all $i \in \Suppc$.

Combining this fact with the previous observation, we obtain that
\begin{equation}
 L_{\dc{\norm{\cdot}_1,\z}} = \left\{ \h \in \R^d : \h_{\Suppc} = \vec{0}, \sp{\sign (\z)}{\h} = 0 \right\},
\end{equation}
which is of dimension $s-1$. From this description, we can conclude that for each $i \in \Supp$ the vector $ s \cdot \sign (z_i) \cdot \vec{e}_i - \sign (\z) $ is contained in the later space. Hence, if we can show that 
\begin{equation}
\dim \left(\spann \left(  s \cdot \sign (z_i) \cdot \vec{e}_i -\sign (\z)  : i \in \Supp \right) \right) = s-1,
\end{equation}
we have succeeded in proving the lemma. 
Indeed, consider the matrix $\vec{C} \in \R^{s\times s-1},$ where the columns are formed by $\left(s \cdot \sign (z_i) \cdot \vec{e}_i - \sign (\z) \right)_{\Supp},$ for each $i\in \Supp$, except for one. Then, the matrix $\vec{C}^T \cdot \vec{C} \in \R^{s-1 \times s-1}$ has the value $s^2-s$ on its diagonal and $-s$ everywhere else. Thus it is strictly diagonal dominant and invertible, implying that $\vec{C}$ is of full rank, as desired.

\paragraph{Lineality and Range for Gauge (Proposition~\ref{prop:lindc})}
Lastly, we characterize the range and lineality of $\dc{p_{\Dict\cdot\Bn[1]{d}},\gt}$:

First, observe that a combination of Lemma \ref{lem:dc} and  Lemma \ref{lem:dc_lone} yields that
\begin{align}
\dc{p_{\Dict\cdot\Bn[1]{d}},\gt} &= \Dict \cdot \dc{\norm{\cdot}_1,\zl}  \\
 & = \Dict  \cdot \cone{\pm \bar{s} \cdot \vec{e}_i - \sign (\zl) : i \in [d]} \\
 & = \cone{\pm \bar{s} \cdot \dict_i - \Dict\sign (\zl) : i \in [d]}.
\end{align}
By Lemma~\ref{lem:lineality}, we know how to characterize the lineality of $\dc{\norm{\cdot}_1,\zl}$. Note that for any convex set $C \subseteq \R^d$, it holds true that $\left( {\Dict \cdot C} \right)_L \supseteq \Dict \cdot C_L$, however, the reverse inclusion is not satisfied, in general. Hence, Lemma~\ref{lem:dc} immediately implies $\left( {\dc{p_{\Dict\cdot\Bn[1]{d}},\gt}} \right)_L \supseteq \Dict \cdot \left( {\dc{\norm{\cdot}_1,\zl}} \right)_L$.
 For proving the reverse inclusion $\left( {\dc{p_{\Dict\cdot\Bn[1]{d}},\gt}} \right)_L \subseteq \Dict \cdot \left( {\dc{\norm{\cdot}_1,\zl}} \right)_L$,
  we will now show that if ${ \left( {\dc{p_{\Dict\cdot\Bn[1]{d}},\gt}} \right)_L \not\subseteq \Dict \cdot \left( {\dc{\norm{\cdot}_1,\zl}} \right)_L}$, 
  then $\zl$ did not have maximal support. To that end, pick any vector $\x\in \left( {\dc{p_{\Dict\cdot\Bn[1]{d}},\gt}} \right)_L \setminus \Dict \cdot \left( {\dc{\norm{\cdot}_1,\zl}} \right)_L$ and write $\x = \Dict \cdot \zt$, where $\zt \in \dc{\norm{\cdot}_1,\zl} \setminus \left( {\dc{\norm{\cdot}_1,\zl}} \right)_L$. Since $\x \in \left( {\dc{p_{\Dict\cdot\Bn[1]{d}},\gt}} \right)_L$, we can also chose a $\zo \in \dc{\norm{\cdot}_1,\zl} \setminus \left( {\dc{\norm{\cdot}_1,\zl}} \right)_L$ with $-\x = \Dict \cdot \zo$. Due to $\zi \not\in \left( {\dc{\norm{\cdot}_1,\zl}} \right)_L$ for $i=1,2$, we have that for all $\varepsilon>0$ 
 \begin{equation}
 \label{eq:minus}
  \norm{\zl - \varepsilon \cdot \zi}_1 > \norm{\zl}_1,
 \end{equation}
 however, there exists a small enough $\varepsilon>0$ such that 
 \begin{equation}
 \label{eq:plus}
  \norm{\zl + \varepsilon \cdot \zi}_1 \leq \norm{\zl}_1.
 \end{equation}
For small enough $\varepsilon>0$, inequality \eqref{eq:minus} implies that 
\begin{equation}
 \sum_{j \in \bar{\Supp}} \sign (z_{\ell^1,j}) \cdot  z^i_j - \sum_{j \in \bar{\Supp}^c} |z^i_j| < 0,
\end{equation}
whereas \eqref{eq:plus} means that
\begin{equation}
 \sum_{j \in \bar{\Supp}}\sign (z_{\ell^1,j}) \cdot  z^i_j + \sum_{j \in \bar{\Supp}^c} |z^i_j| \leq 0.
\end{equation}
Summing up the previous two inequalities, we obtain that $\sum_{j \in \bar{\Supp}} \sign (z_{\ell^1,j}) \cdot (z^1_j + z^2_j) < 0$.
Now, define $\z^{\delta} \coloneqq \zl + \delta \cdot (\zt + \zo)$ and observe that for all $\delta>0$ it holds true that $\x = \Dict \cdot \z^{\delta}$. Furthermore, for a small enough $\delta>0$, we have that $\norm{\z^{\delta}}_1 \leq \norm{\zl}_1$. Hence, we can conclude that $\z^{\delta} \in \Zlset$ and therefore even $\norm{\z^{\delta}}_1 = \norm{\zl}_1$. If $\delta>0$ is chosen small enough, this allows us to write
\begin{equation}
  \norm{\z^{\delta}}_1 = \norm{\zl}_1  + \delta \cdot \sum_{j \in \bar{\Supp}} \sign (z_{\ell^1,j}) \cdot (z^1_j + z^2_j) + \delta \cdot \sum_{j \in \bar{\Supp}^c} |z^1_j + z^2_j|,
\end{equation}
and we can conclude that $\sum_{j \in \bar{\Supp}^c} |z^1_j + z^2_j|>0$. However, this means that there is at least one $j \in \bar{\Supp}^c$ such that $\z^{\delta}_j \neq 0$, which shows that $\z$ was indeed not maximal. 
Finally,  Lemma~\ref{lem:lineality} implies that
\begin{equation}
 \dim \left( \left( {\dc{p_{\Dict\cdot\Bn[1]{d}},\gt}} \right)_L \right) = \dim \left( \left( {\dc{\norm{\cdot}_1,\zl}} \right)_L \right) - \dim \left( \ker \Dict_{|\left( {\dc{\norm{\cdot}_1,\zl}} \right)_L } \right) \leq \bar{s}-1,
\end{equation} 
which concludes the proof of first part of the proposition concerning the lineality of $\dc{p_{\Dict\cdot\Bn[1]{d}},\gt}$.

The characterization of the range follows easily. Indeed, let $i\in \Supp$ and consider the vector $\vec{r}_i^- = -\bar{s}\cdot \sign (z_{\ell^1,i}) \cdot \dict_i -\Dict \cdot \sign (\zl)$. Observe that we can write $\vec{r}_i^- = -2 \cdot \Dict \cdot \sign (\zl) - \vec{r}_i^+$. Hence, for any $j\in \bar{\Supp}^c \neq \emptyset$ we obtain that
\begin{equation}
 P_{C_L^\perp} (\vec{r}_i^-) = -2\cdot P_{C_L^\perp} (\Dict \cdot \sign (\zl)) =  \vec{r}_j^{+\perp} +  \vec{r}_j^{-\perp}.
\end{equation}
Thus, $P_{C_L^\perp} (\vec{r}_i^-) \in \mathrm{cone}(\r_j^{\pm \perp}, j\in \bar \Supp^c)$, which concludes the proof.

\subsection{Proof of Theorem~\ref{thm:decomp}}
\label{sec:proof_decomp}

 Let $C=\dc{p_{\Dict \cdot \Bn[1]{d}},\gt}$ and use the orthogonal decomposition provided in Proposition \ref{prop:lindc}:
 \begin{equation}
  \dc{p_{\Dict\cdot\Bn[1]{d}},\gt} = C_L \oplus C_R.
 \end{equation}
 This allows us to estimate
 \begin{equation}
  \cmw[2]{C} \stackrel{(1)}{\leq} \delta (C) \stackrel{(2)}{\leq} \delta (C_L) + \delta (C_R) \stackrel{(3)}{\leq}  \dim (C_L) +  \cmw[2]{C_R} + 1, 
  \label{eq:intermed}
  \end{equation}
  where $\delta$ denotes the \emph{statistical dimension}; see proof of Proposition~\ref{prop:conditiongaussian} in Appendix~\ref{sec:proof_conditiongaussian} for further details on this notion and a justification of $(1)$. Using the statistical dimension as a summary parameter for convex cones brings several advantages. For a direct sum $C_1 \oplus C_2$ of two closed convex cones $C_1,C_2\subseteq \R^n$ it holds true that $\delta (C_1 \oplus C_2) = \delta (C_1) + \delta (C_2)$, explaining $(2)$ in the previous inequalities. Furthermore, for a subspace $C_L \subseteq \R^n$ we have that $\delta (C_L) = \dim (C_L)$, which, together with $\delta (C_R) \leq \cmw[2]{C_R} +1$, justifies $(3)$.  Observe that the estimate of \eqref{eq:intermed} is essentially tight. 
   
   Proposition~\ref{prop:lindc} allows to upper bound $\dim (C_L) + 1$ by $\bar{s}$. The statement then follows by applying Proposition \ref{prop:circ} to the $2(d-\bar{s})$-polyhedral $\alpha$-cone $C_R$. 

\subsection{Proof of Proposition \ref{prop:coherence_bound} (Coherence Bound)}
\label{sec:coherence_bound}

  First, observe that we have
  \begin{align}
	 \tan^2 \left( \angle(\vec{a},\vec{a} + \vec{b}) \right)= \frac{\|\vec{a}\times(\vec{a}+\vec{b})\|_2^2}{\langle \vec{a},\vec{a}+\vec{b}\rangle^2}
	 =\frac{\norm{\vec{a}}_2^2\norm{\vec{b}}_2^2 - \sp{\vec{a}}{\vec{b}}^2}{\left( \norm{\vec{a}}_2^2 +   \sp{\vec{a}}{\vec{b}} \right)^2}&\leq  \frac{\norm{\vec{a}}_2^2\norm{\vec{b}}_2^2}{\left( \norm{\vec{a}}_2^2 +   \sp{\vec{a}}{\vec{b}} \right)^2}, \label{eq:tanalpha}
  \end{align}
  where $\vec{a},\vec{b}\in \R^n$ with $\vec{a}\neq 0$ and $\vec{a}+\vec{b}\neq 0$. 

    Obsere that the assumptions of Proposition~\ref{prop:lindc} are satisfied. Indeed, $s < \tfrac{1}{2}(1+ \mu^{-1}(\Dict))$ guarantees that $\zl$ is the unique minimal $\ell^1$-representer of the associated signal $\Dict \zl$ and that $\Dict \zl \neq \vec{0}$~\cite{Donoho2003,1255564}.  Hence, we want to evaluate the circumangle of the cone generated by the vectors $\r_j^{\pm\perp}=\proj_{C_L^\perp}(\pm s \cdot \dict_j -\Dict \sign(\zl))$ for $j\in S^c$, where $\Supp=\supp(\zl)$. As a proxy for the circumcenter, we can consider the vector $\vec{v} = - \proj_{C_L^\perp}(\Dict \sign(\zl))$ and therefore obtain:
    \begin{align*}
     \tan^2 \alpha\leq \sup_{j\in S^c} \tan^2(\angle(\vec{v},\r_j^{\pm\perp})) =  \sup_{j\in S^c}\tan^2\left(\angle(\vec{v},\vec{v} + \proj_{C_L^\perp}(s \cdot \dict_j))\right).
    \end{align*}
  We can now use the inequality \eqref{eq:tanalpha} with $\vec{a}=\vec{v}$ and $\vec{b}=s\cdot\proj_{C_L^\perp}(\dict_j)$; note that $\vec{v} \neq \vec{0}$, since otherwise we would have $\Dict \zl = \vec{0}$. The expression \eqref{eq:tanalpha} is decreasing w.r.t.~$\norm{\vec{a}}_2^2$. Hence, we shall find a lower bound for $\norm{\vec{v}}_2^2$. The projection $\proj_{C_L^\perp}(\Dict \sign(\zl))$ can be written as $\Dict \sign(\zl)+\vec{w}$ for some vector $\vec{w}\in C_L$. According to the characterization of the lineality space $C_L$ in Proposition~\ref{prop:lindc}, this amounts to saying that
  \begin{align}
  \label{eq:chara}
   \proj_{C_L^\perp}(\Dict \sign(\zl)) &= \sum_{i\in\Supp} c_i \cdot  \sign (z_{\ell^1,i}) \cdot \dict_i, \, \mbox{ with } \, \sum_{i\in \Supp}c_i=s.
  \end{align}
  This yields 
  \begin{align}
   \norm{\vec{v}}_2^2 &\geq \inf_{\vec{c}\in \R^s, \sum_{i\in \Supp} c_i=s} \norm{\sum_{i\in \Supp} c_i  \sign (z_{\ell^1,i}) \dict_i}_2^2
   \\
   & = \inf_{\vec{c}\in \R^s, \sum_{i\in \Supp} c_i=s} \|\vec{c}\|_2^2 + \sum_{i\in \Supp} \sum_{j\in \Supp, j\neq i} c_ic_j\langle \sign(z_{\ell^1,i}) \dict_i,\sign(z_{\ell^1,j}) \dict_j\rangle\\
   &\geq \inf_{\vec{c}\in \R^s, \sum_{i\in \Supp} c_i=s} \|\vec{c}\|_2^2 - \mu \sum_{i\in \Supp} \sum_{j\in \Supp, j\neq i} c_ic_j.  \label{eq:boundonnormv}
  \end{align}
The optimality conditions for this program yield the existence of a Lagrange multiplier $\lambda\in \R$ such that 
$c_i-\mu \sum_{j\neq i} c_j + \lambda=0$ and $\sum_{i\in \Supp} c_i=s$, i.e., $c_i=1$ for all $i\in \Supp$.  Plugging this expression in \eqref{eq:boundonnormv}, we obtain that 
\begin{align*}
\norm{\vec{v}}_2^2 \geq s-\mu(s\cdot(s-1)) \geq s (1-\mu s).
\end{align*}
Together with the following inequalities:
    \begin{align}
        |\langle \vec{a},\vec{b}\rangle| & =s\left\lvert\sp{\vec{v}}{\proj_{C_L^\perp}\dict_j}\right \rvert \leq s\left\lvert \sp{\vec{v}}{\dict_j} \right\rvert \stackrel{\eqref{eq:chara}}{\leq} s^2 \sup_{i\neq j} |\sp{\dict_i}{\dict_j}| =  s^2 \mu, \label{psvd} \\
        \norm{\vec{b}}_2^2 &= s^2\norm{\proj_{C_L^\perp} \dict_j}_2^2  \leq s^2\norm{\dict_j}_2^2 = s^2, \label{normd}
    \end{align}
    we obtain the desired bound
    \begin{align*}
        \tan^2 \alpha & \leq  \frac{\norm{\vec{a}}_2^2\norm{\vec{b}}_2^2}{\left( \norm{\vec{a}}_2^2 +   \sp{\vec{a}}{\vec{b}} \right)^2} \leq \frac{s(1-\mu s)\cdot s^2}{(s(1-\mu s)-s^2\mu)^2} = \frac{s(1-\mu s)}{(1-2\mu s)^2}.
    \end{align*}    

\section{Details on Numerical Experiments}
\label{sec:num_details}
In this subsection, we report on the setup that we have used in all our numerical experiments. 

\paragraph{Phase Transition Plots}

While our results encompass the more general class of subgaussian measurements, we only consider the benchmark of Gaussian matrices, as it is typically done in the compressed sensing literature. When illustrating the performance of results such as Theorem~\ref{thm:coeff}, we only report the quantity $\cmw{\Dict \cdot \ds{\Onenorm;\gtz}}$, ignoring for instance the probability parameter $u$, cf.~\cite{amelunxen2014edge}.   

\paragraph{Some Details on Computations}

Unless stated otherwise, we solve the convex recovery programs such as~\eqref{eq:coef} or \eqref{eq:minl1} using the \texttt{Matlab} toolbox \texttt{cvx} \cite{cvx1,cvx2}. We employ the default settings and set the precision to \texttt{best}. For creating phase transitions, a solution $\sol$ is considered to be ``perfectly recovered'' if the error to the ground truth vector $\gt$ satisfies $\norm{\gt - \sol}_2 \leq 10^{-5}$. This threshold produces stable transitions and seems to reflect the numerical accuracy of \texttt{cvx}.

\paragraph{Computing the Statistical Dimension}

When analyzing the sampling rate predictions of our results, we often report the conic mean width $\cmw[2]{C} = w (C \cap \mathcal{S}^{n-1} )$ of a convex cone $C \subseteq \R^n$. We will now briefly sketch how this quantity is numerically approximated: First recall that the conic mean width is essentially equivalent to the statistical dimension $\delta (C) = \E [\sup_{\x \in C \cap \Bn[2]{n}} \norm{\Pi_C(g)}_2^2]$; cf.~the proof of Proposition \ref{prop:conditiongaussian} in Appendix~\ref{sec:proof_conditiongaussian}. Due to the convexity of $C \cap \Bn[2]{n}$, the statistical dimension is preferred over the conic mean width for numerical simulations. In order to obtain an approximation of $\delta (C)$, we draw $k$ independent samples $\g_1,\dots,\g_k \sim \mathcal{N}(\vec{0},\Id)$ and for each of them we evaluate the projection $\Pi_C(\g_i)$ using quadratic programming. Due to a concentration phenomenon of empirical Gaussian processes, the arithmetic mean over $k=300$ samples yields tight estimates of $\delta (C)$. 

\paragraph{Minimal Conic Singular Values}

As already mentioned computing $\lmin{\Dict}{\dc{\Onenorm,\zl}}$ is out of reach in general. 
In our numerical experiments on coefficient recovery, we nevertheless provide empirical upper bounds on $\lmin{\Dict}{\dc{\Onenorm,\zl}}$. Those are obtained as follows: Let $\gt = \Dict\cdot \zl$ and consider the perturbed $\gtn = \gt + \hat{\noise}$, where $\hat{\noise} \in \R^n$ such that $\norm{\hat{\noise}}_2 \leq \hat{\noiseparam}$. We then define $\solz\in\R^d$ as a solution of the program
\begin{equation}
 \min_{\z \in \R^d} \norm{\z}_1 \quad \mbox{ s.t. } \quad \norm{\gtn - \Dict \z}_2 \leq \hat{\noiseparam}.
\end{equation}
Proposition~\ref{prop:recover} then implies that $\norm{\zl - \solz}_2 \leq 2\hat{\noiseparam} / \lmin{\Dict}{\dc{\Onenorm,\zl}}$. Rearranging the terms in the previous inequality then yields an upper bound for $\lmin{\Dict}{\dc{\Onenorm,\zl}}$. Note thereby that a clever choice of the perturbation $\hat{\noise}$ may result in a tighter bound.

\paragraph{Computing the Circumcenter and the Circumangle}

Computing the circumcenter amounts to solving:
\begin{equation}\label{eq:convex_circumcenter}
\ax \in \argmin_{\vec{v} \in \Bn[2]{n}} \max_{i \in [k]} \langle -\vec{v},\x_i\rangle, 
\end{equation}
where the vectors $\x_i$ are the normalized generators of a nontrivial pointed polyhedral cone; see Proposition~\ref{prop:simple_circumcenter}. This problem is closely related to the so-called \emph{smallest bounding sphere problem} \cite{sylvester1857question}, which has a long and rich history.

Let $g(\vec{v})=\max_{i \in [k]} \langle -\vec{v}, \x_i\rangle$ and $I(\vec{v})$ denote the set of active indices $i$, i.e., the indices satisfying $g(\vec{v})=-\langle  \vec{v},\x_i\rangle$. Then standard convex analysis results state that $\partial g(\vec{v})=\mathrm{conv}(-\x_i, i \in I(\vec{v}))$ and the optimality conditions read 
\begin{equation}
\ax\in \mathrm{conv}(\x_i, i \in I(\ax)) \qquad \mbox{ with } \qquad  \norm{\ax}_2=1, 
\end{equation}
i.e., the normal cone $\{-\ax\}$ to the constraint set should intersect the subdifferential $\partial g(\ax)$.

Problem \eqref{eq:convex_circumcenter} can be solved globally with projected subgradient descents or second order cone programming techniques available in \emph{CVX}.

\end{document}